\documentclass{theoretics}
\allowdisplaybreaks

\usepackage{ellipsis}

\usepackage{calc}
\usepackage{relsize}

\usepackage{csquotes}

\usepackage[shortcuts]{extdash}
\usepackage{bbm}
\usepackage{booktabs}

\usepackage{pgfplots}
\pgfplotsset{compat=1.13}
\usepackage{tikz}
\usetikzlibrary{shapes,positioning,arrows,arrows.meta,calc,automata,matrix,fit,backgrounds}
\tikzset{->,>=stealth'}

\colorlet{disabled}{lightgray}
\tikzstyle{state}=[draw,rectangle,inner sep=2pt,rounded corners=2pt,minimum size = 8mm, minimum height=6mm,outer sep=0pt]
\tikzstyle{action}=[font=\small,inner sep=0pt,outer sep=3pt]
\tikzstyle{actionnode}=[circle,draw=black,fill=black,minimum size=1mm,inner sep=0,outer sep=0]
\tikzstyle{actionedge}=[draw,-]
\tikzstyle{prob}=[font=\scriptsize,inner sep=0pt,outer sep=1pt]
\tikzstyle{probedge}=[draw,->]
\tikzstyle{directedge}=[draw,->]
\tikzset{chainarrow/.tip={Stealth[length=3pt]}}
\tikzset{>=chainarrow}

\usepackage{xparse}
\usepackage{mathtools}
\usepackage{environ}
\usepackage{etoolbox}
\usepackage{setspace}
\usepackage{float}

\usepackage[capitalize,noabbrev,nameinlink]{cleveref} 

 \DeclarePairedDelimiter{\delimabs}{\lvert}{\rvert}
\DeclarePairedDelimiter{\delimcardinality}{\lvert}{\rvert}
\DeclarePairedDelimiter{\delimnorm}{\lVert}{\rVert}

\NewDocumentCommand{\abs}{sm}{\IfBooleanTF{#1}{\delimabs*{#2}}{\delimabs{#2}}}
\NewDocumentCommand{\cardinality}{sm}{\IfBooleanTF{#1}{\delimcardinality*{#2}}{\delimcardinality{#2}}}
\NewDocumentCommand{\norm}{sm}{\IfBooleanTF{#1}{\delimnorm*{#2}}{\delimnorm{#2}}}
\NewDocumentCommand{\powerset}{r()}{\mathcal{P}(#1)}
\newcommand{\setcomplement}[1]{\overline{#1}}

\newcommand{\unionSym}{\cup}
\newcommand{\unionBin}{\mathbin{\unionSym}}

\newcommand{\intersectionSym}{\cap}
\newcommand{\intersectionBin}{\mathbin{\intersectionSym}}
\newcommand{\UnionSym}{\bigcup}

\newcommand{\IntersectionSym}{\bigcap}

\newcommand{\union}{\unionBin}

\newcommand{\intersection}{\intersectionBin}
\newcommand{\Union}{\UnionSym}

\newcommand{\Intersection}{\IntersectionSym}

\newcommand{\Naturals}{\mathbb{N}}
\newcommand{\Reals}{\mathbb{R}}

\DeclareMathOperator*{\argmax}{arg\,max}

\newcommand{\distribution}{d}
\NewDocumentCommand{\Distributions}{d()}{\IfNoValueTF{#1}{\mathcal{D}}{\mathcal{D}(#1)}}

\NewDocumentCommand{\Expectation}{s d[]}{\IfNoValueTF{#2}{\mathbb{E}}{\mathbb{E}\IfBooleanTF{#1}{\left[#2\right]}{[#2]}}}
\NewDocumentCommand{\Probability}{s d[]}{\IfNoValueTF{#2}{\mathbb{P}}{\mathbb{P}\IfBooleanTF{#1}{\left[#2\right]}{[#2]}}}

\newcommand{\lfalse}{\mathsf{false}}

\newcommand{\MC}{\mathsf{M}}
\newcommand{\MDP}{\mathcal{M}}

\newcommand{\States}{S}
\newcommand{\initialstate}{\hat{s}}
\newcommand{\Actions}{Act}
\NewDocumentCommand{\stateactions}{d()}{\IfNoValueTF{#1}{{Av}}{{Av}(#1)}}
\NewDocumentCommand{\mctransitions}{d()}{\IfNoValueTF{#1}{\delta}{\delta(#1)}}
\NewDocumentCommand{\mdptransitions}{d()}{\IfNoValueTF{#1}{\Delta}{\Delta(#1)}}
\NewDocumentCommand{\actionstate}{r<> r()}{\mathop{\mathsf{state}}(#2, #1)}

\newcommand{\infinitepath}{\rho}
\newcommand{\finitepath}{\varrho}
\NewDocumentCommand{\Infinitepaths}{d<>}{\IfNoValueTF{#1}{\mathsf{Paths}}{\mathsf{Paths}_{#1}}}
\NewDocumentCommand{\Finitepaths}{d<>}{\IfNoValueTF{#1}{\mathsf{FPaths}}{\mathsf{FPaths}_{#1}}}
\newcommand{\strategy}{\pi}
\NewDocumentCommand{\Strategies}{r<>}{\Pi_{#1}}
\NewDocumentCommand{\StrategiesM}{r<>}{\Pi_{#1}^{\mathsf{M}}}
\NewDocumentCommand{\StrategiesMD}{r<>}{\Pi_{#1}^{\mathsf{MD}}}
\newcommand{\last}[1]{last(#1)}

\DeclareMathOperator{\SccsOp}{SCC}
\DeclareMathOperator{\BsccsOp}{BSCC}
\DeclareMathOperator{\EcsOp}{EC}
\DeclareMathOperator{\MecsOp}{MEC}

\NewDocumentCommand{\Sccs}{r()}{\SccsOp(#1)}
\NewDocumentCommand{\Bsccs}{r()}{\BsccsOp(#1)}
\NewDocumentCommand{\Ecs}{d()}{\IfNoValueTF{#1}{\EcsOp}{\EcsOp(#1)}}
\NewDocumentCommand{\Mecs}{d()}{\IfNoValueTF{#1}{\MecsOp}{\MecsOp(#1)}}

\DeclareMathOperator{\support}{supp}

\NewDocumentCommand{\ProbabilityMC}{s r<> d[]}{\mathsf{Pr}_{#2}\IfNoValueF{#3}{\IfBooleanTF{#1}{\!\left[#3\right]\!}{[#3]}}}
\NewDocumentCommand{\ProbabilityMDP}{s r<> r<> d[]}{\mathsf{Pr}_{#2}^{#3}\IfNoValueF{#4}{\IfBooleanTF{#1}{\!\left[#4\right]\!}{[#4]}}}
\NewDocumentCommand{\ProbabilityMDPmax}{s r<> d[]}{\mathsf{Pr}_{#2}^{\max}\IfNoValueF{#3}{\IfBooleanTF{#1}{\!\left[#3\right]\!}{[#3]}}}
\NewDocumentCommand{\ProbabilityMDPsup}{s r<> d[]}{\mathsf{Pr}_{#2}^{\sup}\IfNoValueF{#3}{\IfBooleanTF{#1}{\!\left[#3\right]\!}{[#3]}}}

\NewDocumentCommand{\ExpectedSum}{m m}{#1\langle#2\rangle}
\NewDocumentCommand{\ExpectedSumMC}{m m m}{\ExpectedSum{#1(#2)}{#3}}
\NewDocumentCommand{\ExpectedSumMDP}{m m m m}{\ExpectedSum{#1(#2,#3)}{#4}}

\NewDocumentCommand{\ExpectedSumStrat}{m m}{#1[#2]}

\newcommand{\reach}{\Diamond}
\NewDocumentCommand{\boundedreach}{r<>}{\Diamond^{\leq #1}}

\DeclareMathOperator{\val}{\mathcal{V}}
\newcommand{\learningalgo}{\mathsf{A}}
\newcommand{\upperbound}{\mathsf{Up}}
\newcommand{\lowerbound}{\mathsf{Lo}}
\newcommand{\bounddifference}{\mathrm{Diff}}

\newcommand{\targetstate}{s_+}
\newcommand{\sinkstate}{s_-}
\newcommand{\targetset}{T}

\newcommand{\samplepath}{\textsc{SamplePairs}}
\DeclareMathOperator{\collapse}{collapse}
\newcommand{\updateecs}{\textsc{UpdateECs}}

\newcommand{\ecremain}{\mathrm{rem}}
\newcommand{\ecequivalent}{\mathsf{equiv}}
\newcommand{\eccollapsed}{\mathsf{collapsed}}
\newcommand{\ecstates}{\mathsf{states}}

\newcommand{\successor}{\mathsf{succ}}

\newcommand{\updatestep}{{\overline{\varepsilon}}}
\newcommand{\delay}{{\overline{m}}}
\newcommand{\updatecount}{{\overline{\xi}}}

\newcommand{\acc}{\texttt{acc}}
\newcommand{\visitcount}{\texttt{count}}
\newcommand{\learn}{\texttt{learn}}
\newcommand{\learnyes}{\texttt{yes}}
\newcommand{\learnonce}{\texttt{once}}
\newcommand{\learnno}{\texttt{no}}
\newcommand{\learndecrease}{\textsc{Decrease}}
\newcommand{\algostep}{\mathsf{t}}
\newcommand{\algoepisode}{\mathsf{e}}

\newcommand{\algoecs}{\mathsf{EC}}
\newcommand{\umaxactions}{\mathsf{MaxA}}

\newcommand{\algorithmspace}{\mathfrak{A}}
\newcommand{\algorithmsigma}{\mathcal{A}}
\NewDocumentCommand{\ProbabilityAlgo}{s d[]}{\IfNoValueTF{#2}{\mathbb{P}_{\learningalgo}}{\mathbb{P}_{\learningalgo}\IfBooleanTF{#1}{\left[#2\right]}{[#2]}}}

\newcommand{\algorithmexecution}{\mathfrak{a}}

\NewDocumentCommand{\ConvergedUpperBounds}{r<>}{\mathcal{K}_{#1}^{\upperbound}}
\NewDocumentCommand{\ConvergedLowerBounds}{r<>}{\mathcal{K}_{#1}^{\lowerbound}}

\makeatletter
\def\alabel#1#2{\begingroup
	\textbf{(#2)}\def\@currentlabel{\textbf{[#2]}}\phantomsection\label{#1}\endgroup
}
\makeatother

\newcounter{proofref}

\newcommand{\plabel}[1]{\refstepcounter{proofref}\label{#1}\textbf{[Fact \Roman{proofref}]}}
\AtBeginEnvironment{proof}{\setcounter{proofref}{0}}

\newcounter{proofcount}
\AtBeginEnvironment{proof}{\stepcounter{proofcount}}
\makeatletter
\@addtoreset{proofref}{proofcount}
\makeatother

\newtheorem{assumption}{Assumption}
\crefname{assumption}{Assumption}{Assumptions}
\crefname{casedistinction}{Case}{Cases}

\newcommand{\oldthankssym}{\textasteriskcentered}
\newcommand{\oldthanks}{\textsuperscript{\oldthankssym}}
	\title{Learning Algorithms for Verification of Markov Decision Processes}
\ThCSauthor[Brno]{Tom{\'a}{\v{s}} Br{\'a}zdil}{}[0000-0002-4547-3261]
\ThCSauthor[Klosterneuburg]{Krishnendu Chatterjee}{}[0000-0002-4561-241X]
\ThCSauthor[Google]{Martin Chmelik}{}[]
\ThCSauthor[Oxford]{Vojt{\v{e}}ch Forejt}{}[]
\ThCSauthor[Brno]{Jan K{\v{r}}et{\'\i}nsk{\'y}}{}[0000-0002-8122-2881]
\ThCSauthor[Oxford]{Marta Kwiatkowska}{}[0000-0001-9022-7599]
\ThCSauthor[Leipzig]{Tobias Meggendorfer}{tobias@meggendorfer.de}[0000-0002-1712-2165]
\ThCSauthor[Oxford]{David Parker}{}[0000-0003-4137-8862]
\ThCSauthor[Toronto]{Mateusz Ujma}{}[]

\ThCSaffil[Brno]{Masaryk University, Brno, Czech Republic}
\ThCSaffil[Klosterneuburg]{IST Austria, Klosterneuburg, Austria}
\ThCSaffil[Google]{Google LLC, Zurich, Switzerland}
\ThCSaffil[Oxford]{University of Oxford, Oxford, UK}
\ThCSaffil[Leipzig]{Lancaster University Leipzig, Leipzig, Germany}
\ThCSaffil[Toronto]{Rogers Communications, Toronto, Canada}

\ThCSthanks{This research was funded in part by the European Research Council (ERC) under grant agreement AdG-267989 (QUAREM)\oldthanks{}, AdG-246967 (VERIWARE)\oldthanks{}, StG-279307 (Graph Games)\oldthanks{}, CoG-863818 (ForM-SMArt), and AdG-834115 (FUN2MODEL), by the EU FP7 project HIERATIC\oldthanks{}, by the German Research Foundation (DFG) project 427755713 (GOPro), by the Austrian Science Fund (FWF) projects S11402-N23 (RiSE)\oldthanks{}, S11407-N23 (RiSE)\oldthanks{}, and P23499-N23\oldthanks{}, by the Czech Science Foundation grant No P202/12/P612\oldthanks{} and GA23-06963S, by the MUNI Award in Science and Humanities (MUNI/I/1757/2021) of the Grant Agency of Masaryk University, by EPSRC project EP/K038575/1\oldthanks{}, and by the Microsoft faculty fellows award\oldthanks{}. A preliminary version of this article appeared at ATVA 2014 \cite{DBLP:conf/atva/BrazdilCCFKKPU14}.
		The \oldthankssym{} indicates funding that supported that version.}
\ThCSshortnames{T.\ Br{\'a}zdil et.\ al}  
\ThCSshorttitle{Learning Algorithms for Verification of Markov Decision Processes}
\ThCSyear{2025}
\ThCSarticlenum{10}
\ThCSreceived{Mar 22, 2024}
\ThCSrevised{Nov 6, 2024}
\ThCSaccepted{Jan 12, 2025}
\ThCSpublished{April 1, 2025}
\ThCSdoicreatedtrue
\ThCSkeywords{Markov decision processes and Learning}

\begin{document}
	\maketitle
	
	\begin{abstract}
		We present a general framework for applying learning algorithms and heuristical guidance to the verification of Markov decision processes (MDPs).
The primary goal of our techniques is to improve performance by avoiding an exhaustive exploration of the state space, instead focussing on particularly relevant areas of the system, guided by heuristics.
Our work builds on the previous results of Br{\'{a}}zdil~et~al., significantly extending it as well as refining several details and fixing errors.

The presented framework focuses on probabilistic reachability, which is a core problem in verification, and is instantiated in two distinct scenarios.
The first assumes that full knowledge of the MDP is available, in particular precise transition probabilities.
It performs a heuristic-driven partial exploration of the model, yielding precise lower and upper bounds on the required probability.
The second tackles the case where we may only sample the MDP without knowing the exact transition dynamics.
Here, we obtain probabilistic guarantees, again in terms of both the lower and upper bounds, which provides efficient stopping criteria for the approximation.
In particular, the latter is an extension of statistical model checking (SMC) for unbounded properties in MDPs.
In contrast to other related approaches, we do not restrict our attention to time-bounded (finite-horizon) or discounted properties, nor assume any particular structural properties of the MDP.
	\end{abstract}

	\section{Introduction} \label{sec:introduction}

Markov decision processes (MDP) \cite{howard1960dynamic,filar1996competitive,DBLP:books/wi/Puterman94} are a well established formalism for modelling, analysis and optimization of probabilistic systems with non-determinism, with a large range of application domains \cite{DBLP:books/daglib/0020348,DBLP:conf/cav/KwiatkowskaNP11}.
For example, MDPs are used as models for concurrent probabilistic systems~\cite{DBLP:journals/jacm/CourcoubetisY95} or probabilistic systems operating in open environments~\cite{segala1996modeling}.
See \cite{white1985real,white1988further,white1993survey} for further applications.

In essence, MDP comprise three major parts, namely states, actions, and probabilities.
Intuitively, the system evolves as follows:
In any state, there is a set of actions to choose from.
This corresponds to the \emph{non-determinism} of the system.
After choosing an action, the system then transitions into the next state according to the probability distribution associated with that action.
For example, we may use MDP to represent a robot moving around in a 2D world (sometimes called \enquote{gridworld}).
The states then are (bounded, integer) coordinates, representing the current position of the robot.
In each state the robot can choose to move in one of the four cardinal directions or carry out some task depending on the current location.
To illustrate the randomness, consider a \enquote{move east} action.
Choosing this action may move the robot to the next position east of the current one, but it might also be the case that, with some probability, a navigation component of the robot fails and we instead end up in a state north of our current position.
Given such a system, the general goal is to optimize a given \emph{objective} by choosing optimal actions.
For example, we may want to control the robot such that it reaches an interesting research site with maximal probability.
We additionally may be interested in minimizing time or power consumption and avoiding dangerous terrain on our way to the site.

This example hints at one of the simplest, yet important objectives, namely \emph{reachability}.
A reachability problem is specified by an MDP together with a set of designated target states.
The task is to compute the maximal probability with which the system can reach this set of states.
Reachability is of particular interest since in the infinite horizon setting many other objectives, e.g., LTL or long-run average reward, can be reduced to variants of reachability.
A variety of approaches has been established to solve this problem.
In theory, linear programming \cite{10.1007/BFb0032043,DBLP:conf/sfm/ForejtKNP11} is the most suitable approach, as it provides exact answers (rational numbers with no representation imprecision) in polynomial time.
See \cite{DBLP:conf/qest/BarnatBCCT08} for an application.
Unfortunately, LP turns out to be quite inefficient in practice for classical reachability.
For systems with more than a few thousand states, linear programming often falls behind other approaches, see, e.g., \cite{DBLP:conf/sfm/ForejtKNP11,DBLP:conf/cav/AshokCDKM17,DBLP:conf/tacas/HartmannsJQW23}.
As an alternative, one can apply iterative methods.
Here, value iteration (VI) \cite{howard1960dynamic} is the most prominent variant.
See \cite{DBLP:conf/spin/ChatterjeeH08} for a detailed survey of VI.
Notably, variations of VI are the default method in the state-of-the-art probabilistic model checkers PRISM \cite{DBLP:conf/cav/KwiatkowskaNP11} and Storm \cite{DBLP:conf/cav/DehnertJK017}, even though it only provides an approximate solution, converging in the limit.
In contrast, strategy iteration (SI) (also known as policy iteration, PI) \cite{howard1960dynamic,DBLP:books/wi/Puterman94,DBLP:conf/atva/KretinskyM17} yields precise answers, but is also used to a lesser extent due to scalability issues.
See for example \cite{DBLP:journals/tnn/Bertsekas17} for an overview of both methods, \cite{DBLP:conf/tacas/HartmannsJQW23,practitionersJournalPreprint} for recent comparisons of practical implementations of LP, VI, and SI for MDP, \cite{DBLP:journals/iandc/KretinskyRSW22} for a similar comparison on \emph{stochastic games} (MDP with two antagonistic players), and \cite{qcomp2023} for in-depth practical comparison of modern probabilistic model checkers.

\paragraph{Interval Iteration}
Surprisingly, until about a decade ago, standard value iteration as applied in popular model checkers only yielded \emph{lower bounds} on the true value, without any \emph{sound stopping criterion}.
Concretely, this meant that the model checker might conclude that the computation is finished and stop it, despite still being far off from the true result.
We note that there exists a tight, exponential a-priori bound on the number of steps VI requires until convergence, see e.g.\ \cite{DBLP:conf/spin/ChatterjeeH08}.
This could be used as \enquote{stopping criterion}, by simply iterating for this number of steps.
However, this is far too pessimistic on most models.

In \cite{DBLP:conf/rp/HaddadM14,DBLP:conf/atva/BrazdilCCFKKPU14}, a correct and \emph{adaptive} stopping criterion was discovered independently.
This bound follows from under- and (newly obtained) over-approximations converging to the true value, yielding a straightforward stopping criterion: iterate until upper and lower bound are close enough.
This criterion is adaptive in the sense that if the iteration should converge faster than the naive a-priori bound, we can detect this case and stop early.
Subsequent works included this stopping criterion in model checkers \cite{DBLP:conf/cav/Baier0L0W17} and developed further sound approaches \cite{DBLP:conf/cav/QuatmannK18,DBLP:conf/cav/HartmannsK20}.
(Some more developments are discussed in the related work.)

However, despite value iteration scaling much better than linear programming, systems with more than a few million states remain out of reach, not only because of time-outs, but also memory-outs.
Several approaches have been devised to deal with such large state spaces, which we extensively survey in the related work section.
Now, we outline a variant of VI, called \emph{asynchronous VI}.
The central idea is to perform the iterative computations in an asynchronous manner, i.e.\ apply the iteration operation to some states more often than to others, or even not at all to some states.
This allows to obtain speed-ups of several orders of magnitude.
However, since states are evaluated at different paces and, potentially, a set of states is omitted completely, convergence is unclear and even its rate is unknown and hard to analyse.
Yet, by exploiting the discussed lower and upper bounds, we obtain a correct and efficient algorithm, inspired by \emph{bounded real-time dynamic programming} (BRTDP) \cite{DBLP:conf/icml/McMahanLG05}.
This algorithm interleaves construction of the model, analysis, and bound approximation.
For example, we can sample a path through the system (constructing states that we have not seen so far on the fly) and apply the bound update mechanism only on these paths.
For some models, this allows to obtain tight bounds on the true value while only constructing a small fraction of the complete state space.

\paragraph{Limited Information}
The methods discussed above (and most which are introduced in the related work) rely on an exact formalization of the system being available.
In particular they require that the transition probabilities are known precisely.
We call this situation the \emph{white box} or \emph{complete information} setting.
This is a common, valid assumption when verifying, e.g., formally defined protocols, but not so much when working with real-world systems comprising difficult dynamics, where the effects of an action can be approximated at most.
As such, these systems can be treated as a \emph{black box}, which accept a next action to take as input and output the subsequent state, sampled from the associated underlying, unknown distribution.

Here \emph{statistical model checking} (SMC) \cite{DBLP:conf/cav/YounesS02,DBLP:conf/vmcai/HeraultLMP04} is applicable.
The general idea of SMC is to repeatedly sample the system in order to obtain strong statistical guarantees.
Thus, SMC approaches can (at most) be \emph{probably approximately correct} (PAC), i.e.\ yield an answer close to the true value with high probability, but there always is a small chance for a significant error.
By itself, SMC algorithms are restricted to systems without non-determinism, e.g., Markov chains \cite{DBLP:conf/cav/Younes05a,DBLP:conf/qest/SenVA05}.
A number of approaches tackling the issue of non-determinism have been presented (see related work).
However, these methods deal with non-determinism by either resolving it uniformly at random or sample several schedulers, both of which can lead to surprising results in certain scenarios \cite{DBLP:conf/isola/BohlenderBJKNN14}.
Additionally, note that both approaches can only give a statistical estimate of a lower bound of the true achievable maximal reachability.
In particular, they do not give any guarantees on the \emph{maximal} achievable performance (i.e.\ an upper bound). Based on the ideas of \emph{delayed Q-learning (DQL)} \cite{DBLP:conf/icml/StrehlLWLL06} (which also only yields lower bounds) we present a PAC \emph{model-free} algorithm, yielding statistical \emph{upper and lower} bounds on the \emph{maximal} reachability.
(Model-free intuitively means that our algorithm only stores a fixed number of values per state-action pair, independent of how many transitions are associated with that action; further discussion can be found in \cref{rem:model_free}.)
This approach is similar in spirit to the BRTDP approach discussed above, however much more involved due to the underlying statistical arguments.
The main contribution of this algorithm is to prove the possibility of obtaining such a result, exploring the boundaries of what exactly is necessary to obtain guarantees.


\begin{algorithm}[t]
	\caption{High-level overview of the structure of our algorithms.}
	\label{fig:intro:overview}
	
	\DontPrintSemicolon
	\KwIn{MDP $\MDP$, target states $\targetset$, precision $\varepsilon$.}
	\KwOut{Values $(l, u)$ which are $\varepsilon$-optimal.}
	\While{difference between upper and lower bound in initial state is larger than $\varepsilon$}{
		Obtain a set of states to update by, e.g., sampling a path. \label{alg:overview:start}
		\ForEach{state and action in this set}{
			\uIf{this state is a target state}{
				Set its bounds to $1$. \;
			}
			\Else{
				Update action bounds based on the weighted average of its successors. \label{alg:overview:end} \;
			}
		}
		Detect end components in the relevant area of the system. \label{alg:overview:ec} \;
	}
	\Return lower and upper bound of the initial state. \;
\end{algorithm}

\paragraph{Algorithm Outline}
To provide the reader with a preliminary overview of our approach, we present a high-level pseudo-code in \cref{fig:intro:overview}.
As already mentioned, the fundamental idea is to compute lower and upper bounds on the true probability of reaching the target in each state (\cref{alg:overview:start} to \cref{alg:overview:end}).
Essentially, we want to iteratively update these bounds in a converging and correct manner.
In the complete information setting, this can be achieved by directly computing the weighted average of the successor bounds.
For the limited information setting, we instead aggregate many successor samples.
This yields a good approximation of this weighted average with high probability.

The details of how the set of states to be updated is obtained in \cref{alg:overview:start} are abstracted in the complete information setting and we only require some basic properties.
One possibility is a sampling-based approach, which is guided by the currently computed bounds.
We discuss several alternatives later on.
In contrast, the limited information setting requires a particular kind of sampling approach in order to ensure correctness.
We highlight these differences in the respective sections.

Now, while it is rather simple to prove correctness of the computed bounds, the tricky part is to obtain convergence.
In particular, for general MDP, this approach would not converge.
To solve this, in the past many algorithms working with MDP often made assumptions about the structure of the model.
For example, it was sometimes required that the model is \enquote{strongly connected} or free of \emph{end components} \cite{de1997formal} (except trivial ones).
Instead, one of the main contributions of \cite{DBLP:conf/rp/HaddadM14,DBLP:conf/atva/BrazdilCCFKKPU14} is to identify end components as the sole \enquote{culprits} and devising methods to deal with them in a general manner, obtaining convergence.
While \cite{DBLP:conf/rp/HaddadM14} tackles the problem in a \enquote{global} manner (assuming to have access to the complete MDP at once), we present an asynchronous way of treating end components.
This treatment is \enquote{on-the-fly} and can be interleaved with the iterative construction of the system. 

In the white box setting, we solve this problem by adapting exiting graph analysis algorithms and incorporating them with our main procedure.
However, with limited information we again need to employ statistical methods.
In essence, if we remain inside a particular region of the system for a long enough time, there is a high probability that this region is an end component.
This overall process then is repeated until the computed bounds in the initial state are close enough.

\subsection{Related Work}
We present a number of related ideas, all attempting, in one way or another, to make the analysis of (large or black box) probabilistic system tractable.

\emph{Compositional} techniques aim to first analyse parts of the system separately and combine the sub-results to obtain an overall result, e.g.\  \cite{DBLP:conf/qest/CaillaudDLLPW10,DBLP:journals/scp/DengH13,DBLP:conf/concur/HermannsKK13,DBLP:conf/concur/BassetKW14,DBLP:journals/fmsd/ChatterjeeCD15,DBLP:journals/iandc/BassetKW18}.
Then, there are \emph{abstraction} approaches which try to merge states with equivalent or sufficiently similar behaviour w.r.t.\ the objective in question, e.g.\  \cite{DBLP:conf/papm/DArgenioJJL02,DBLP:conf/cav/HermannsWZ08,DBLP:conf/tacas/HahnHWZ10,DBLP:journals/fmsd/KattenbeltKNP10,DBLP:conf/tacas/HahnHWZ10}.
\emph{Reduction} approaches try to eliminate states from the system and restrict computation to a sub-system through structural properties, e.g.\  \cite{DBLP:conf/qest/BaierGC04,DBLP:journals/entcs/BaierDG06,DBLP:conf/qest/CiesinskiBGK08,DBLP:conf/qest/DiazBEF12,DBLP:conf/fm/FanHM18,DBLP:conf/concur/BonnelandJLMS19}.
\emph{Guessing} \cite{DBLP:conf/soda/ChatterjeeMSS23} tries to guess and verify the value of certain states, which can lead to theoretical speed-ups when the guesses decompose the system into independent parts.
Another approach is \emph{symbolic computation}, where the model and value functions are compactly represented using \emph{BDD} \cite{DBLP:journals/tc/Bryant86} and \emph{MTBDD} / \emph{ADD} \cite{DBLP:journals/fmsd/BaharFGHMPS97,DBLP:journals/fmsd/FujitaMY97}.
See \cite{DBLP:conf/concur/BaierKH99,DBLP:journals/sttt/KwiatkowskaNP04,DBLP:conf/aaai/ZamaniSF12,DBLP:conf/qest/WimmerBBHCHDT10,DBLP:journals/corr/BohyBR14a,DBLP:conf/tacas/0001BCDDKM016} for further details and applications of symbolic methods.

In related fields such as planning and artificial intelligence, many learning-based and heuristic-driven approaches for MDP have been proposed.
In the complete information setting, RTDP \cite{DBLP:journals/ai/BartoBS95} and BRTDP \cite{DBLP:conf/icml/McMahanLG05} use very similar approaches, but have no stopping criterion or do not converge in general, respectively.
\cite{DBLP:conf/ijcai/PineauGT03} uses upper and lower bounds in the setting of \emph{partially observable MDP} (POMDP).
Many other algorithms rely on certain assumptions to ensure convergence, for example by including a \emph{discount factor} \cite{DBLP:journals/ml/KearnsMN02} or restricting to the \emph{Stochastic Shortest Path} (SSP) problems, whereas we deal with arbitrary MDP without discounting.
This is addressed by an approach called FRET \cite{DBLP:conf/aips/KolobovMWG11}, but this only yields a lower bound.
Others similarly only provide convergence in the limit \cite{DBLP:journals/corr/abs-1909-07299,DBLP:journals/corr/JonesAKSB15}, which is usually satisfactory for applications to planning or robotics, where systems have intractably large or even uncountable state spaces.
We are not aware of any attempts at generally applicable methods in the context of probabilistic verification prior to \cite{DBLP:conf/atva/BrazdilCCFKKPU14}.
An earlier, related paper is \cite{DBLP:conf/qest/AljazzarL09}, where heuristic methods are applied to MDP, but for generating counterexamples.

As mentioned, \cite{DBLP:conf/rp/HaddadM14} independently discovered a stopping criterion for value iteration on general MDP.
The idea behind this criterion is very similar to \cite{DBLP:conf/atva/BrazdilCCFKKPU14}, but they construct and analyse the whole system at once.
The underlying idea of \enquote{interval iteration}, spawned by these two papers, is further developed in, e.g., \cite{DBLP:conf/cav/Baier0L0W17,DBLP:journals/tcs/HaddadM18,DBLP:conf/lics/AshokCKWW20}.

Additionally, the idea of \emph{optimistic} value iteration (OVI) \cite{DBLP:conf/cav/HartmannsK20,DBLP:conf/atva/AzeemEKSW22} emerged.
Here, instead of \emph{always} updating both lower and upper bound, only the lower bound is iterated (as in classical value iteration).
Then, based on heuristics, the algorithm \emph{optimistically} conjectures that the values actually converged.
To verify that conjecture, a (potential) upper bound is guessed based on the current lower bound (e.g.\ by incrementing all bounds by $\varepsilon$) and then checked for consistency by applying a few steps of VI.
This approach turns out to be quite efficient in practice when dealing with MDP in a \enquote{global} manner, however is incompatible with our guided sampling approach, since we continuously use upper bounds for guidance.
Similarly, \emph{sound} value iteration (SVI) \cite{DBLP:conf/cav/QuatmannK18} also works with lower and upper bounds, however they derive bounds based on $k$-step reachability probabilities.
These fundamentally require a global and synchronous value iteration, which is precisely what we aim to avoid.

\paragraph{Statistical Methods}
There are two primary motivations to use statistical approaches.
Firstly, the model might be large, even too large to fit into memory, and analysing it by standard approaches becomes infeasible, yet generating samples may be quick and easy.
In this case, one can decide to \enquote{only} aim for a statistical guarantee, which often comes with tremendous speed-ups and space savings.
Secondly, as explained above, the model might be an unknown black box -- we do not know how it works internally, only that it is some Markov process.
If we can observe and control the system, we can gather samples and from that derive statistical guarantees for the considered value.

As mentioned, our approach focuses on the latter, however most statistical methods focus on the former.
Indeed, many of the following methods are \emph{only} applicable to the \enquote{full knowledge} setting, i.e.\ knowing the internals of the system.
Here, significant improvements can be observed:
Several SMC algorithms have sub-linear or even constant space requirements (often called \emph{model-free} algorithms).
Appropriately, SMC is an active area of research with extensive tool support \cite{DBLP:conf/tacas/JegourelLS12,DBLP:conf/mmb/BogdollHH12,DBLP:conf/qest/BoyerCLS13,DBLP:journals/corr/abs-1207-1272,DBLP:conf/cav/DavidLLMW11,DBLP:conf/cav/Younes05a,DBLP:conf/qest/SenVA05,wats} but also a lot of subtle pitfalls \cite{DBLP:journals/sigmobile/KurkowskiCC05}.
See also \cite{DBLP:conf/isola/Kretinsky16,DBLP:journals/tomacs/AghaP18,DBLP:series/lncs/LegayLTYSG19} for extensive surveys and \cite{DBLP:conf/hybrid/RoohiW0D017} for an application of SMC to a complex real world problem.
In contrast to our work, most algorithms focus on \emph{time-bounded} or discounted properties, e.g., step-bounded reachability, rather than truly unbounded properties.
Several approaches try to bridge this gap by transforming unbounded properties into testing of bounded properties, for example \cite{DBLP:conf/sbmf/YounesCZ10,DBLP:conf/kbse/HeJBGW10,DBLP:conf/atva/RabihP09,DBLP:conf/cav/SenVA05}.
However, these approaches target models without nondeterminism and as such are not applicable to MDP.
As a slight extension, \cite{DBLP:conf/forte/BogdollFHH11} considers MDP with \emph{spurious nondeterminism}, where the resolution of nondeterminism does not influence the value of interest.

Adapting SMC techniques to models with (true) nondeterminism such as MDP is an important topic, with several recent papers.
See \cite{chang2013simulation,DBLP:journals/jmlr/StrehlLL09} and \cite[Chapter~4.1.5]{DBLP:phd/dnb/Weininger22} for a survey on simulation-based algorithms in this context.
One approach is to give nondeterminism a probabilistic interpretation, e.g., resolving it uniformly, as is done in PRISM for MDP \cite{DBLP:conf/cav/KwiatkowskaNP11} and Uppaal SMC for timed automata \cite{DBLP:conf/cav/DavidLLMW11,DBLP:conf/formats/DavidLLMPVW11,DBLP:conf/ifm/Larsen13}.
A second approach, taken for example by recent versions of the \texttt{modes} tool \cite{DBLP:conf/tacas/HartmannsH14,DBLP:conf/isola/DArgenioHS18,DBLP:conf/tacas/BuddeDHS18}, is to repeatedly sample schedulers, using for example \emph{lightweight scheduler sampling} (LSS) \cite{DBLP:conf/sefm/LegayST14,DBLP:journals/sttt/DArgenioLST15}, and then estimate the performance of these controllers using existing SMC methods.
Uppaal Stratego \cite{DBLP:conf/tacas/DavidJLMT15} synthesizes a \enquote{good} scheduler and uses it for subsequent SMC analysis.
All of the above methods only yield a lower bound on the true reachability and the quality of this bound is highly dependent on the model.
Others aim to indeed quantify over all strategies and approximate the true maximal value, for example \cite{DBLP:journals/sttt/LassaigneP15,DBLP:conf/qest/HenriquesMZPC12}.
The work in those papers deals with the setting of discounted or bounded properties, respectively.
In \cite{DBLP:conf/qest/HenriquesMZPC12}, candidates for optimal schedulers are generated and gradually improved, which does not give upper bounds on the convergence.
The nearly simultaneously published \cite{DBLP:conf/rss/FuT14} essentially tackles the same problem.
In contrast to our work, their approach is model-based, i.e.\ the transition probabilities are learned, and is not guided by a heuristic, requiring to learn the whole transition matrix.

In summary, most approaches are only applicable to the first case, or, if they can work in the \enquote{limited information} setting, they require a purely probabilistic system, finite or discounted properties, or only give lower bounds on the optimal value.
Our focus explicitly lies on the limited information case, and, similar to many approaches from statistical model checking \cite{DBLP:conf/cav/YounesS02,DBLP:conf/vmcai/HeraultLMP04,DBLP:conf/cav/SenVA04,DBLP:conf/rss/FuT14}, we aim to provide PAC guarantees, however on the optimal value of an infinite horizon objective in models with nondeterminism.

Another issue of statistical methods is the analysis of \emph{rare events}.
This is, of course, very relevant for SMC approaches in general.
They can be addressed using for example importance sampling \cite{DBLP:conf/tacas/JegourelLS12,DBLP:conf/kbse/HeJBGW10} or importance splitting \cite{DBLP:conf/cav/JegourelLS13,DBLP:conf/setta/BuddeDH17}.
We take a rather conservative approach towards rare events and delegate more sophisticated handling of this issue to future work.

\subsection{Differences to the Published Article} \label{sec:intro:differences}

This work is a significant extension of \cite{DBLP:conf/atva/BrazdilCCFKKPU14}.
Numerous details are refined and errors discovered and fixed.
We discuss major changes in the following.
Notably, in the process of resolving some of the issues of \cite{DBLP:conf/atva/BrazdilCCFKKPU14}, we also discovered several problems in \cite{DBLP:conf/icml/StrehlLWLL06}, on which the DQL method of \cite{DBLP:conf/atva/BrazdilCCFKKPU14} is based, both conceptually and in terms of proof structure.
\begin{itemize}
	\item
	A complete rewrite, only retaining parts of the proof strategies.

	\item
	The related work is updated with recent advances and work based on \cite{DBLP:conf/atva/BrazdilCCFKKPU14}.

	\item
	The BRTDP approach and related proofs are extended significantly to a generic template, allowing for a variety of implementations of the sampling methods.

	\item
	Both variants of the DQL algorithm have been restructured and simplified.

	\item
	The proofs, especially those related to DQL, are more modular and easier to adapt / re-use for similar endeavours in these directions.

	\item
	Several technical issues of the original paper are fixed.
	Firstly, the proofs in the appendix proved properties of slightly different algorithms, only to conclude with a brief, imprecise argument that the presented algorithms are not too different from the algorithms proven correct.
	Some proofs were only given implicitly or assumed to be common knowledge, in particular treatment of collapsed end components and similar.
	Moreover, several small mistakes have been corrected.

	\item
	Lemma~16 of the original paper both has a flawed proof and an erroneous statement, which is now fixed:
	Firstly, the Algorithm as presented potentially never follows an $\varepsilon$ optimal strategy, as exemplified in \cref{exa:upper_strat_is_bad}.
	Secondly, the proof applies the multiplicative Chernoff bound to variables $X_i$, which indicate whether the algorithm performed a particular action during a time interval.
	To apply this bound, the variables would need to be independent, but the $X_i$ are dependent.
	This is elaborated in detail later on.

	Interestingly, a similar, yet slightly different error already is present in \cite{DBLP:conf/icml/StrehlLWLL06}.
	Firstly, their Theorem~1 claims that the algorithm eventually follows an $\varepsilon$ optimal strategy, which does not hold due to the same reasons.
	Secondly, in the corresponding proof the authors apply the Hoeffding bound to similar dependent variables.
	This happens at same location in the overall proof layout as in \cite{DBLP:conf/atva/BrazdilCCFKKPU14}, however the applied bound is different.
	Our alternative approach to proving the statements is also applicable to the proof of \cite{DBLP:conf/icml/StrehlLWLL06}.
\end{itemize}

\subsection{Impact of the Presented Work} \label{sec:intro:influence}

Since its publications about a decade ago, the two approaches introduced by \cite{DBLP:conf/atva/BrazdilCCFKKPU14}, i.e.\ BRTDP for complete information and DQL for limited information, have directly inspired a number of subsequent works, of which we provide a (non-exhaustive) list.
Firstly, the BRTDP approach has been extended to settings with long-run average reward \cite{DBLP:conf/cav/AshokCDKM17}, continuous time Markov chains \cite{DBLP:conf/atva/AshokBHK18}, continuous space MDP \cite{DBLP:conf/concur/GroverKMW22}, and stochastic games \cite{DBLP:journals/iandc/EisentrautKKW22}.
Notably, taking inspiration from \cite{DBLP:conf/atva/BrazdilCCFKKPU14} and subsequent works, \cite{DBLP:conf/lics/KretinskyMW23} recently provided a unified approach to value iteration for stochastic games.
Concretely, this work extends the central ideas required to obtain convergence guarantees in MDP to stochastic games in a unified way, subsuming and extending, among others, the ideas and algorithms of \cite{DBLP:conf/atva/BrazdilCCFKKPU14,DBLP:conf/rp/HaddadM14,DBLP:conf/cav/AshokCDKM17,DBLP:journals/iandc/EisentrautKKW22}.
In particular, this explains how to extend the BRTDP approach to further objectives, such as safety, expected total reward, or mean payoff.
In an orthogonal direction, \cite{DBLP:journals/lmcs/KretinskyM20} modifies the approach of \cite{DBLP:conf/atva/BrazdilCCFKKPU14} to determine \emph{cores} of probabilistic systems, which intuitively describe \enquote{most} possible behaviours of the given system.
(This can also be viewed as a probabilistic generalization of the set of reachable states.)

Secondly, the DQL approach (and its proof strategy) inspired a \emph{model-based} variant \cite{DBLP:conf/cav/AshokKW19}, which improved scalability.
(Note that, as remarked in \cite[Appendix~D]{DBLP:conf/cav/AshokKW19}, the convergence of their \enquote{fast} variant is not proven.)
Subsequently, this lead to a surge of papers considering model-based SMC, for example adapting to MDP with reachability \cite{DBLP:conf/isola/AshokDKW20} or mean payoff objective \cite{DBLP:conf/cav/AgarwalGKM22}, continuous state-spaces \cite{DBLP:journals/jair/BadingsRAPPSJ23}, dynamic information flow tracking games \cite{DBLP:conf/cdc/WeiningerGMK21}, or changing environments \cite{DBLP:conf/nips/SuilenS0022}.

Thirdly, for practical impact, we highlight the tool PET \cite{DBLP:conf/atva/Meggendorfer22,DBLP:conf/cav/MeggendorferW24}, which directly implements and extends the BRTDP approach in a highly efficient manner.
As seen in several evaluations, the relevance of partial exploration in practice highly depends on the structure of the model (as with many other approaches).
In some cases, effectively the entire model has to be explored and there is no improvement possible.
However, for several families of models orders-of-magnitude or even \emph{arbitrary} speed-ups can be observed.
This tool has also participated in several iterations of the \emph{Comparison of Tools for the Analysis of Quantitative Formal Models} (QComp), a friendly competition of quantitative model checking tools, namely in 2019 \cite{DBLP:conf/tacas/HahnHHKKKPQRS19} (as \texttt{PRISM-TUMheuristic}), 2020 \cite{DBLP:conf/isola/BuddeHKKPQTZ20}, and 2023 \cite{qcomp2023}.

\subsection{Contributions and Structure}

In \cref{sec:preliminaries} we set up notation and introduce some known results.
We then present our contributions as follows.
\begin{itemize}
	\item We introduce an extensible framework for efficient reachability on \enquote{complete information} MDP without end components in \cref{sec:brtdp_no_ec} and extend it to arbitrary MDP in \cref{sec:brtdp}.
	\item We introduce a model-free PAC learning algorithm for reachability on \enquote{limited information} MDP without end components in \cref{sec:dql_no_ec} and extend it to arbitrary MDP in \cref{sec:dql}.
\end{itemize}
We conclude in \cref{sec:conclusion}.
We intentionally omit an experimental evaluation and instead refer to tools based on these ideas, see e.g.\ the works in \cref{sec:intro:influence}.
 	\section{Preliminaries} \label{sec:preliminaries}

\begin{figure}
  \centering
     \includegraphics[scale=1.4]{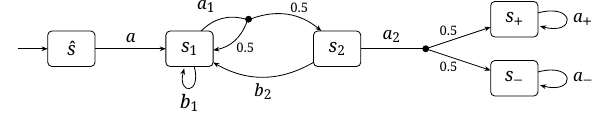}





	\caption{
		An example Markov decision process.
		Boxes represent states, dots represent actions, and arrows correspond to transitions (with the respective probabilities as labels).
		For simplicity, actions with a single successor are depicted as a single, direct arrow and the probability $1$ is omitted.
		We use this notation throughout the paper.
	} \label{fig:example}
\end{figure}

As usual, $\Naturals$ and $\Reals$ refers to the (positive) natural numbers and real numbers, respectively.
Given two real numbers $a, b \in \Reals$ with $a \leq b$, $[a, b] \subseteq \Reals$ denotes the set of all real numbers between $a$ and $b$ inclusively.
For a set $S$, $\setcomplement{S}$ denotes its complement, while $S^\star$ and $S^\omega$ refers to the set of finite and infinite sequences comprising elements of $S$, respectively.
We often explicitly name sub-claims in the form of \plabel{proof:pseudo}, and reference them by \ref{proof:pseudo}. In the digital version, the references are clickable.

We assume familiarity with basic notions of probability theory, e.g., \emph{probability spaces} and \emph{probability measures}.
A \emph{probability distribution} over a countable set $X$ is a mapping $\distribution : X \to [0,1]$, such that $\sum_{x \in X} \distribution(x) = 1$.
Its \emph{support} is denoted by $\support(\distribution) = \{x \in X \mid \distribution(x) > 0\}$.
$\Distributions(X)$ denotes the set of all probability distributions on $X$.
Some event happens \emph{almost surely} (a.s.) if it happens with probability $1$.
For readability, we omit detailed treatment of probability measures on uncountable sets and instead direct the reader to appropriate literature, e.g.\ \cite{billingsley2008probability}.
\subsection{Markov Systems}
Markov decision processes (MDPs) are a widely used formalism to capture both non-determinism (for, e.g., control, concurrency) and probability.
For a \enquote{complete} introduction to Markov systems, we direct the interested reader to \cite{DBLP:books/wi/Puterman94,kallenberg2011markov}.
A lighter, more recent introduction can be found in \cite[Chapter~2]{DBLP:phd/dnb/Meggendorfer21}.

First, we introduce Markov chains (MCs), which are purely stochastic.
\begin{definition}
	A \emph{Markov chain} (MC) is a tuple $\MC = (\States, \mctransitions)$, where
		$\States$ is a (countable) set of \emph{states}, and
		$\mctransitions : \States \to \Distributions(\States)$ is a \emph{transition function} that for each state $s$ yields a probability distribution over successor states.
\end{definition}
Note that we do not require the set of states of a Markov chain to be finite.
This is mainly due to technical reasons, which become apparent later.

Next, we define MDP, which extend Markov chains with non-determinism.
\begin{definition}
	A \emph{Markov decision process} (MDP) is a tuple $\MDP = (\States, \Actions, \stateactions, \mdptransitions)$, where
		$\States$ is a finite set of \emph{states},
		$\Actions$ is a finite set of \emph{actions},
		$\stateactions: \States \to 2^{\Actions} \setminus \{\emptyset\}$ assigns to every state a non-empty set of \emph{available actions}, and
		$\mdptransitions: \States \times \Actions \to \Distributions(\States)$ is a \emph{transition function} that for each state $s$ and (available) action $a \in \stateactions(s)$ yields a probability distribution over successor states.

	A state $s \in \States$ is called \emph{terminal}, if $\mdptransitions(s, a)(s) = 1$ for all enabled actions $a \in \stateactions(s)$.
\end{definition}
\begin{remark}
	We assume w.l.o.g.\ that actions are unique for each state, i.e.\ $\stateactions(s) \intersection \stateactions(s') = \emptyset$ for $s \neq s'$ and denote the unique state associated with action $a$ in $\MDP$ by $\actionstate<\MDP>(a)$.
	This can be achieved in general by replacing $\Actions$ with $\States \times \Actions$ and adapting $\stateactions$ and $\mdptransitions$.
\end{remark}
Note that we assume the set of available actions to be non-empty in all states.
This means that a run can never get \enquote{stuck} in a degenerate state without successors.
See \cref{fig:example} for an example of an MDP.

For ease of notation, we overload functions mapping to distributions $f: Y \to \Distributions(X)$ by $f: Y \times X \to [0, 1]$, where $f(y, x) \coloneqq f(y)(x)$.
For example, instead of $\mctransitions(s)(s')$ and $\mdptransitions(s, a)(s')$ we write $\mctransitions(s, s')$ and $\mdptransitions(s, a, s')$, respectively.
Furthermore, given a distribution $\distribution \in \Distributions(X)$ and a function $f : X \to \Reals$ mapping elements of a set $X$ to real numbers, we write $\ExpectedSum{\distribution}{f} \coloneqq \sum_{x \in X} \distribution(x) f(x)$ to denote the weighted sum of $f$ with respect to $\distribution$.
For example, $\ExpectedSumMC{\mctransitions}{s}{f}$ and $\ExpectedSumMDP{\mdptransitions}{s}{a}{f}$ denote the weighted sum of $f$ over the successors of $s$ in MC and $s$ with action $a$ in MDP, respectively.

\subsubsection*{State-Action Pairs}
Throughout this work, we often speak about \emph{state-action pairs}.
This refers to tuples of the form $(s, a)$ where $s \in \States$ and $a \in \stateactions(s)$ or equivalently $a \in \Actions$ and $s = \actionstate<\MDP>(a)$.
Due to our restriction that each action is associated with exactly one state, denoting both the state and action is superfluous, strictly speaking.
We keep the terminology for consistency with other works.
In \cref{sec:dql} this notation would however introduce significant overhead and we only speak about actions there.

Given a set of states $\States' \subseteq \States$ and an available-action function $\stateactions' : \States' \to \powerset(\Actions) \setminus \emptyset$ we write, slightly abusing notation, $\States' \times \stateactions' = \{(s, a) \mid s \in \States', a \in \stateactions'(s)\}$ to denote the set of state-action pairs obtained in $\States'$ using $\stateactions'$.
In particular, $\States \times \stateactions$ denotes the set of all state-action pairs in an MDP.
Moreover, for a set of state-action pairs $K$ we also write $s \in K$ if there exists an action $a$ such that $(s, a) \in K$.
Dually, we also write $a \in K$ if an appropriate state $s$ exists.

Note that there are two isomorphic representations of sets of state-action pairs, namely as a set of pairs $X \subseteq \States \times \stateactions$ or as a pair of sets $(R, B) \in 2^\States \times 2^{\Actions}$.
We make use of both views and note explicitly when switching from one to another.

\subsubsection*{Paths \& Strategies}
An \emph{infinite path} $\infinitepath$ in a Markov chain is an infinite sequence $\infinitepath = s_1 s_2 \cdots \in \States^\omega$, such that for every $i \in \Naturals$ we have that $\mctransitions(s_i, s_{i+1}) > 0$.
A \emph{finite path} (or \emph{history}) $\finitepath = s_1 s_2 \dots s_n \in \States^\star$ is a non-empty, finite prefix of an infinite path of length $\cardinality{\finitepath} = n$, ending in some state $s_n$, denoted by $\last{\finitepath}$.
For simplicity, we define $\cardinality{\infinitepath} = \infty$ for infinite paths $\infinitepath$.
We use $\infinitepath(i)$ and $\finitepath(i)$ to refer to the $i$-th state $s_i$ in a given (in)finite path.
A state $s$ \emph{occurs} in an (in)finite path $\infinitepath$, denoted by $s \in \infinitepath$, if there exists an $i \leq \cardinality{\infinitepath}$ such that $s = \infinitepath(i)$.
We denote the set of all finite (infinite) paths of a Markov chain $\MC$ by $\Finitepaths<\MC>$ ($\Infinitepaths<\MC>$).
Further, we use $\Finitepaths<\MC, s>$ ($\Infinitepaths<\MC, s>$) to refer to all (in)finite paths starting in state $s \in \States$.
Observe that in general $\Finitepaths<\MC>$ and $\Infinitepaths<\MC>$ are proper subsets of $\States^\star$ and $\States^\omega$, respectively, as we imposed additional constraints.

An \emph{infinite path} in an MDP is an infinite sequence $\infinitepath = (s_1, a_1) (s_2, a_2) \cdots \in (\States \times \stateactions)^\omega$, such that for every $i \in \Naturals$, $a_i \in \stateactions(s_i)$ and $s_{i+1} \in \support(\mdptransitions(s_i, a_i))$, setting the length $\cardinality{\infinitepath} = \infty$.
\emph{Finite path}s $\finitepath$ and $\last{\finitepath}$ are defined analogously as elements of $(\States \times \stateactions)^\star \times \States$ and the respective last state.
Again, $\infinitepath(i)$ and $\finitepath(i)$ refer to the $i$-th state in an (in)finite path with an analogous definition of a state occurring, $\cardinality{\finitepath}$ denotes the length of a finite path, we refer to the set of (in)finite paths of an MDP $\MDP$ by $\Finitepaths<\MDP>$ ($\Infinitepaths<\MDP>$), and write $\Finitepaths<\MDP, s>$ ($\Infinitepaths<\MDP, s>$) for all such paths starting in a state $s \in \States$.
Further, we use $\infinitepath^a(i)$ and $\finitepath^a(i)$ to denote the $i$-th action in the respective path.
We say that a state-action pair $(s,a)$ is in an (in)finite path $\finitepath$ if there exists an $i < \cardinality{\finitepath}$ with $s = \finitepath(i)$ and $a = \finitepath^a(i)$.

A Markov chain together with a state $s \in \States$ naturally induces a unique probability measure $\ProbabilityMC<\MC, s>$ over infinite paths \cite[Chapter~10]{DBLP:books/daglib/0020348}.
For MDP, we first need to eliminate the non-determinism in order to obtain such a probability measure.
This is achieved by \emph{strategies} (also called \emph{policy}, \emph{controller}, or \emph{scheduler}).

\begin{definition}
	A strategy on an MDP $\MDP = (\States, \Actions, \stateactions, \mdptransitions)$ is a function mapping finite paths to distributions over available actions, i.e.\ $\strategy : \Finitepaths<\MDP> \to \Distributions(\Actions)$ where $\support(\strategy(\finitepath)) \subseteq \stateactions(\last{\finitepath})$ for all $\finitepath \in \Finitepaths<\MDP>$.
\end{definition}
Intuitively, a strategy is a \enquote{recipe} describing which step to take in the current state, given the evolution of the system so far.
Note that the strategy may yield a distribution on the actions to be taken next.

A strategy $\strategy$ is called \emph{memoryless} (or \emph{stationary}) if it only depends on $\last{\finitepath}$ for all finite paths $\finitepath$ and we identify it with $\strategy : \States \to \Distributions(\Actions)$.
Similarly, it is called \emph{deterministic}, if it always yields a Dirac distribution, i.e.\ picks a single action to be played next, and we identify it with $\strategy : \Finitepaths<\MDP> \to \Actions$.
Together, \emph{memoryless deterministic} strategies can be treated as functions $\strategy : \States \to \Actions$ mapping each state to an action.
We write $\Strategies<\MDP>$ to denote the set of all strategies of an MDP $\MDP$, $\StrategiesM<\MDP>$ for memoryless strategies, and $\StrategiesMD<\MDP>$ for all memoryless deterministic strategies.

Fixing a strategy $\strategy$ induces a Markov chain $\MDP^\strategy = (\Finitepaths<\MDP>, \mctransitions^\strategy)$, where for a state $\finitepath = s_1 a_1 \dots s_n \in \Finitepaths<\MDP>$, action $a_{n+1} \in \stateactions(s_n)$, and successor state $s_{n+1} \in \support(\mdptransitions(s_n, a_{n+1}))$, the successor distribution is given by $\mctransitions^\strategy(\finitepath, \finitepath a_{n+1} s_{n+1}) = \strategy(\finitepath, a_{n+1}) \cdot \mdptransitions(s, a_{n+1}, s_{n+1})$.
In particular, for any MDP $\MDP$, strategy $\strategy \in \Strategies<\MDP>$, and state $s$, we obtain a measure over paths\footnote{Technically, this measure operates on infinite sequences of finite paths, as each state of $\MDP^\strategy$ is a finite path.
But this measure can easily be projected directly on finite paths.}
$\ProbabilityMC<\MDP^\strategy, s>$, which we refer to as $\ProbabilityMDP<\MDP, s><\strategy>$.
Observe that all these measures operate on the same probability space, namely the set of all infinite paths $\Infinitepaths<\MDP>$.
(See e.g.\ \cite[Section~2.1.6]{DBLP:books/wi/Puterman94} for further details.)
Consequently, given a measurable event $A$, we can define the maximal probability of this event starting from state $\initialstate$ under any strategy by
\begin{equation*}
	\ProbabilityMDPsup<\MDP, \initialstate>[A] \coloneqq {\sup}_{\strategy \in \Strategies<\MDP>} \ProbabilityMDP<\MDP, \initialstate><\strategy>[A].
\end{equation*}
Note that depending on the structure of $A$ it may be the case that no optimal witness exists, thus we have to resort to the supremum instead of the maximum.
We lift this restriction for our particular use case later on.
For a memoryless strategy $\strategy \in \StrategiesM<\MDP>$, we can identify $\MDP^\strategy$ with a Markov chain over the states of $\MDP$.

Given an MDP $\MDP$, memoryless strategy $\strategy \in \StrategiesM<\MDP>$, and a function assigning a value to each state-action pair $f : \States \times \stateactions \to \Reals$, we define $\ExpectedSumStrat{\strategy}{f} : \States \to \Reals$ as the expected value of taking one step in state $s$ following the strategy $\strategy$, i.e.
\begin{equation*}
	\ExpectedSumStrat{\strategy}{f}(s) \coloneqq {\sum}_{a \in \stateactions(s)} \strategy(s, a) \cdot f(s, a).
\end{equation*}

\subsubsection*{Strongly Connected Components and End Components}
A non-empty set of states $C \subseteq \States$ in a Markov chain is \emph{strongly connected} if for every pair $s, s' \in C$ there is a non-trivial path from $s$ to $s'$.
Such a set $C$ is a \emph{strongly connected component} (SCC) if it is inclusion maximal, i.e.\ there exists no strongly connected $C'$ with $C \subsetneq C'$.
Thus, each state belongs to at most one SCC.
An SCC is called \emph{bottom strongly connected component} (BSCC) if additionally no path leads out of it, i.e.\ for all $s \in C, s' \in S \setminus C$ we have $\mctransitions(s, s') = 0$.
The set of SCCs and BSCCs in an MC $\MC$ is denoted by $\Sccs(\MC)$ and $\Bsccs(\MC)$, respectively.

The concept of SCCs is generalized to MDPs by so called \emph{(maximal) end components} \cite{de1997formal}.
Intuitively, an end component describes a set of states in which the system can remain forever.
\begin{definition} \label{def:ec}
	Let $\MDP = (\States, \Actions, \stateactions, \mdptransitions)$ be an MDP.
	A pair $(R, B)$, where $\emptyset \neq R \subseteq \States$ and $\emptyset \neq B \subseteq \Union_{s \in R} \stateactions(s)$, is an \emph{end component} of an MDP $\MDP$ if
	\begin{enumerate}[(i)]
		\item for all $s \in R, a \in B \intersection \stateactions(s)$ we have $\support(\mdptransitions(s, a)) \subseteq R$, and \label{def:ec:support}
		\item for all $s, s' \in R$ there is a finite path $\finitepath = s a_0 \dots a_n s' \in \Finitepaths<\MDP> \intersection (R \times B)^\star \times R$, i.e.\ the path stays inside $R$ and only uses actions in $B$. \label{def:ec:reach}
	\end{enumerate}
	An end component $(R, B)$ is a \emph{maximal end component} (MEC) if there is no other end component $(R', B')$ such that $R \subseteq R'$ and $B \subseteq B'$.
\end{definition}
We identify an end component with the respective set of states, e.g.\ $s \in E = (R, B)$ means $s \in R$.
Observe that given two overlapping ECs $(R_1, B_1)$ and $(R_2, B_2)$ with $R_1 \intersection R_2 \neq \emptyset$, their union $(R_1 \union R_2, B_1 \union B_2)$ also is an EC.
Consequently, each state belongs to at most one MEC.
Again, a MEC is \emph{bottom} if there are no outgoing transitions.
The set of ECs of an MDP $\MDP$ is denoted by $\Ecs(\MDP)$, the set of MECs by $\Mecs(\MDP)$.
For the MDP in \cref{fig:example}, the set of MECs is given by $(\{s_1, s_2\}, \{a_1, b_1, b_2\})$, $(\{\targetstate\}, \{a_+\})$, and $(\{\sinkstate\}, \{a_-\})$.

\begin{remark} \label{rem:scc_and_mec_decomposition}
	For a Markov chain $\MC$, the computation of $\Sccs(\MC)$, $\Bsccs(\MC)$ and a topological ordering of the SCCs can be achieved in linear time w.r.t.\ the number of states and transitions by, e.g., Tarjan's algorithm~\cite{DBLP:journals/siamcomp/Tarjan72}.
	Similarly, the MEC decomposition of an MDP can be computed in polynomial time \cite{DBLP:journals/jacm/CourcoubetisY95}.
	For improved algorithms on general MDP and various special cases see~\cite{DBLP:conf/soda/ChatterjeeH11,DBLP:conf/soda/ChatterjeeH12,DBLP:journals/jacm/ChatterjeeH14}.
\end{remark}
These components fully capture the limit behaviour of any Markov chain and decision process, respectively.
Intuitively, both of the following statements say that a run of such systems eventually remains inside one BSCC or MEC forever, respectively.
The measurability of the sets in the following two lemmas is well known, see, e.g.\ \cite[Chapter~10]{DBLP:books/daglib/0020348}.
\begin{lemma}[MC almost-sure absorption] \label{stm:mc_almost_sure_absorption}
	For any MC $\MC$ and state $s$, we have that $\ProbabilityMC<\MC, s>[\{\infinitepath \mid \exists R_i \in \Bsccs(\MC). \exists n_0 \in \Naturals. \forall n > n_0. \infinitepath(n) \in R_i\}] = 1$.
\end{lemma}
\begin{proof}
	Follows from \cite[Theorem~10.27]{DBLP:books/daglib/0020348}. 
\end{proof}
\begin{lemma}[MDP almost-sure absorption] \label{stm:mdp_almost_sure_absorption}
	For any MDP $\MDP$, state $s$, and strategy $\strategy$, we have that 
	\begin{equation*}
		\ProbabilityMDP<\MDP, s><\strategy>[\{\infinitepath \mid \exists (R_i, B_i) \in \Mecs(\MDP). \exists n_0 \in \Naturals. \forall n > n_0. \infinitepath(n) \in R_i\}] = 1.
	\end{equation*}
\end{lemma}
\begin{proof}
	Follows from \cite[Theorem~3.2]{de1997formal}. 
\end{proof}

\subsection{Reachability}
For an MDP $\MDP = (\States, \Actions, \stateactions, \mdptransitions)$ and a set of \emph{target states} $T \subseteq \States$, \emph{bounded reachability} for step $k$, denoted by $\boundedreach<k> T = \{\infinitepath \in \Infinitepaths<\MDP> \mid \exists i \in \{1, \dots, k + 1\}.~\infinitepath(i) \in T\}$, is the set of all infinite paths that reach a state in $T$ within $k$ steps.
Analogously, \emph{(unbounded) reachability} $\reach T = \{\infinitepath \in \Infinitepaths<\MDP> \mid \exists i \in \Naturals.~\infinitepath(i) \in T\}$ are all paths which eventually reach the target set $T$.
We overload the $\reach$ operator to also accept sets of state-action pairs and sets of actions, with analogous semantics.
The sets of paths produced by $\reach$ are measurable for any MDP, target set, and step bound \cite[Section~10.1.1]{DBLP:books/daglib/0020348}.\footnote{Recall that we defined MDP to have finite state and action sets.}
Note that for a set $\targetset$, both $\reach \setcomplement{\targetset}$ and $\setcomplement{\reach \targetset}$ are well-defined, however they refer to two different concepts.
The former denotes the set of all paths reaching a state not in $\targetset$, whereas the latter is the set of all paths which never reach $\targetset$ (also called \emph{co-reachability} or \emph{safety}).

Now, it is straightforward to define the \emph{maximal reachability problem} of a given set of states.
Given an MDP $\MDP$, target set $\targetset$, and state $s$, we are interested in computing the maximal probability of eventually reaching $\targetset$, starting in state $s$.
Formally, we want to compute the \emph{value} of state $s$, defined as
\begin{align*}
	\val(s) \coloneqq \ProbabilityMDPsup<\MDP, s>[\reach \targetset] = {\sup}_{\strategy \in \Strategies<\MDP>} \ProbabilityMDP<\MDP, s><\strategy>[\reach \targetset].
\end{align*}
For an example, suppose we have $\targetset = \{\targetstate\}$ in \cref{fig:example}.
This can be reached from $\initialstate$ with probability $0.5$ by always choosing action $a_1$ in $s_1$ and $a_2$ in $s_2$, and this value is optimal.
In general, an optimal strategy always exists and memoryless deterministic strategies are sufficient to achieve the optimal value \cite[Theorem~3.10]{de1997formal}, i.e.
\begin{equation*}
	\val(s) = \ProbabilityMDPmax<\MDP, s>[\reach \targetset] = {\max}_{\strategy \in \Strategies<\MDP>} \ProbabilityMDP<\MDP, s><\strategy>[\reach \targetset] = {\max}_{\strategy \in \StrategiesMD<\MDP>} \ProbabilityMDP<\MDP, s><\strategy>[\reach \targetset].
\end{equation*}
This state value function satisfies a straightforward fixed point equation, namely
\begin{equation} \label{eq:value_fixpoint}
	\val(s) = \begin{dcases*}
		1 & if $s \in \targetset$, \\
		{\max}_{a \in \stateactions(s)} \ExpectedSumMDP{\mdptransitions}{s}{a}{\val} & otherwise.
	\end{dcases*}
\end{equation}
Moreover, $\val$ is the \emph{smallest} fixed point of this equation \cite{DBLP:books/wi/Puterman94}.
In our approach, we also deal with values of state-action pairs $(s, a) \in \States \times \stateactions$, where
\begin{equation*}
	\val(s, a) \coloneqq \ExpectedSumMDP{\mdptransitions}{s}{a}{\val} = {\sum}_{s' \in \States} \mdptransitions(s, a, s') \cdot \val(s').
\end{equation*}
Intuitively, $\val(s, a)$ is the value in state $s$ when playing action $a$ and then acting optimally (note that $a$ might be a suboptimal action).
The overall value of $s$, $\val(s)$, is obtained by choosing an optimal action, i.e.\ $\val(s) = \max_{a \in \stateactions(s)} \val(s, a)$.
\begin{remark}
	Our algorithms primarily work by approximating these state-action values and derive state-values by the above equation.
	This may seem counter-intuitive at first, since we could as well directly work with state values and derive state-action values as described above, saving memory.
	However, our approaches are inspired by \emph{reinforcement learning} \cite{DBLP:books/lib/SuttonB98}, explained later, which traditionally assigns values to actions.
	Thus, we stick with this convention in our algorithms as well.
	Finally, in the limited information setting of \cref{sec:dql_no_ec,sec:dql}, the algorithms do not have access to the exact transition probabilities and hence cannot exploit the above equation.
\end{remark}
See \cite[Section~4]{DBLP:conf/sfm/ForejtKNP11} for an in-depth discussion of reachability on finite MDP.

\subsubsection*{Approximate Solutions}

The value of a state $\val(s)$ can, for example, be determined using \emph{linear programming} \cite{10.1007/BFb0032043,DBLP:conf/sfm/ForejtKNP11}\footnote{See \cite{DBLP:books/daglib/0090562} for details on linear programming in general.} in polynomial time \cite{khachiyan1979polynomial,DBLP:journals/combinatorica/Karmarkar84}.
Unfortunately, this approach turns out to be inefficient in practice \cite{DBLP:conf/rp/HaddadM14,DBLP:conf/cav/AshokCDKM17}.
One way to potentially ease the task is by only considering \emph{approximate solutions}.
Concretely, on top of an MDP $\MDP$, starting state $\initialstate$, and target set $\targetset$, we assume that we are given a precision requirement $\varepsilon > 0$.
We say a strategy $\strategy$ is $\varepsilon$-optimal, if $\ProbabilityMDP<\MDP, \initialstate><\strategy>[\reach \targetset] + \varepsilon > \val(\initialstate)$.
Analogously, a tuple of values $(l, u)$ is \emph{$\varepsilon$-optimal} if $0 \leq u - l < \varepsilon$ and $\val(\initialstate) \in [l, u]$, i.e.\ $l$ and $u$ are lower and upper bounds on the value, respectively.
All algorithms in this work are designed to efficiently compute such $\varepsilon$-optimal values.
We omit computation of a witness strategy due to the technical difficulties this would entail in the general cases.
The general idea of obtaining the witness strategies moreover is not specific to our approach, as such the related discussion may in turn distract from the central results.

Note that requiring to find a single value $v$ such that $|v - \val(\initialstate)| < \varepsilon$ is similar, however slightly stricter.
In particular, if we find $(l, u)$ with $0 \leq u - l < 2 \varepsilon$ where $\val(\initialstate) \in [l, u]$, we know that $v = (u + l) / 2$ would satisfy this requirement (i.e.\ be at most $\varepsilon$ away from the true value).

\subsection{Probabilistic Learning Algorithms} \label{sec:preliminaries:learning}
In order to obtain such approximate solutions, we study a class of \emph{learning-based} algorithms that (stochastically) approximate the value function, inspired by approaches from the field of machine learning.
Let us fix an MDP $\MDP = (\States, \Actions, \stateactions, \mdptransitions)$, starting state $\initialstate$, and target set $\targetset \subseteq \States$.
Recall that by approximating the state-action values, we approximate the overall value of a state.
Inspired by \emph{BRTDP} (bounded real-time dynamic programming) \cite{DBLP:conf/icml/McMahanLG05}\footnote{See \cite{DBLP:journals/ai/BartoBS95} for the \enquote{non-bounded} case \emph{RTDP}.}, we consider algorithms which maintain and update $\upperbound$per bounds $\upperbound : \States \times \stateactions \to [0, 1]$ and $\lowerbound$wer bounds $\lowerbound : \States \times \stateactions \to [0, 1]$ of these sate-action values $\val(s, a)$.
The functions $\upperbound$ and $\lowerbound$ are initialised to appropriate values such that $\lowerbound(s, a) \leq \val(s, a) \leq \upperbound(s, a)$ for all $s \in \States$ and $a \in \stateactions(s)$.
This is clearly satisfied by $\lowerbound(\cdot, \cdot) = 0$ and $\upperbound(\cdot, \cdot) = 1$, but non-trivial bounds obtained by previous computations or domain knowledge can be incorporated.
We define the state-bounds by
\begin{equation*}
	\upperbound(s) \coloneqq {\max}_{a \in \stateactions(s)} \upperbound(s, a), \qquad \text{and} \qquad \lowerbound(s) \coloneqq {\max}_{a \in \stateactions(s)} \lowerbound(s, a).
\end{equation*}
It may seem counter-intuitive at first that both sides are maximized.
One can think of $\upperbound(s)$ as \enquote{an upper bound on the best this state can offer} (maximization) and $\lowerbound(s)$ as \enquote{at least this value can be obtained in this state} (also maximization).

Now, we clearly have  $\lowerbound(s) \leq \val(s) \leq \upperbound(s)$, thus we can determine the value of a state $\varepsilon$-precise when these respective bounds are sufficiently close.
In particular, if we have that
\begin{equation*}
	\upperbound(\initialstate) - \lowerbound(\initialstate) = {\max}_{a \in \stateactions(\initialstate)} \upperbound(\initialstate, a) - {\max}_{a \in \stateactions(\initialstate)} \lowerbound(\initialstate, a) < \varepsilon,
\end{equation*}
the values $(\lowerbound(\initialstate), \upperbound(\initialstate))$ are $\varepsilon$-optimal.

Our \emph{learning algorithms} update the upper and lower bounds by repeatedly selecting \enquote{interesting} / promising state-action pairs of the system $\MDP$, usually by sampling the system beginning in the starting state $\initialstate$.
As such, they are similar to \emph{Q-learning} \cite{watkins1992q} approaches, a commonly used reinforcement learning technique.
By following appropriate sampling heuristics the algorithm learns \enquote{important} areas of the system and focuses computation there, potentially omitting irrelevant parts of the state space without sacrificing correctness.
For example, given a state $s$ we propose to select an action $a$ with maximal upper bound $\upperbound(s, a)$, as such an action is the most \enquote{promising} one.
Then, either this action keeps up to its promise, which will eventually be reflected by an increasing lower bound, or the algorithm finds that the upper bound is too high and lowers it.
As such, this idea is very similar to \emph{optimism in the face of uncertainty} \cite[Section~4.2]{DBLP:series/synthesis/2010Szepesvari}, \cite{lai1985asymptotically}:
We only know that the exact value lies between the upper and lower bound, thus we are optimistic and assume the best value (= the upper bound) during sampling.
As it turns out, this will lead us to either (i)~proving that the upper bound is indeed correct (so following it was the \enquote{correct} move all along) or (ii)~proving that the bound is too optimistic, i.e.\ leading us to lower it (so following it was \enquote{required} to realize this fact).

The algorithms repeatedly experience (learning) \emph{episodes}, where each episode consists of several \emph{steps}.
One episode corresponds to sampling a path of some length in the system, while one step corresponds to sampling the successor state, i.e.\ each episode comprises several steps.
Throughout this paper, we use $\algoepisode \in \Naturals$ exclusively to refer to the $\algoepisode$-th episode of some algorithm execution.
Later we also refer to distinct steps within episodes by $\algostep \in \Naturals$.
In particular, $\algostep$ denotes the $\algostep$-th overall step.
Finally, $\algostep_\algoepisode$ denotes the first step of the $\algoepisode$-th episode, i.e.\ its starting step.
These variables also appear in the algorithms.

The considered algorithms make heavy use of randomness during their execution.
Thus, in order to reason about them, we model them as a stochastic process over an appropriate measure space $(\algorithmspace, \algorithmsigma, \ProbabilityAlgo)$.
The entire state of our algorithms at the beginning of episode $\algoepisode$ only depends on the sequences of state-action pairs considered until episode $\algoepisode$.\footnote{Due to their \enquote{template}-structure, \cref{alg:brtdp_no_ec,alg:brtdp} are allowed to introduce some further side effects.
For example, they may keep a round-robin counter on actions or other heuristics that are used to sample paths in the system.
We assume w.l.o.g.\ that these side effects are either deterministic or can be properly incorporated into the above measure space.}
Hence, we use episodes as our primitive objects.
We need to consider both finite and infinite episodes, since (i)~a single episode might in theory comprise infinitely many state-action pairs and (ii)~we could see infinitely many episodes, each of finite length.
(In both cases, the algorithm does not terminate.)
Thus, we set $\algorithmspace = ((\States \times \stateactions \times \States)^\times)^\times$, where $S^\times = S^\star \union S^\omega$.
(Note that this can be encoded into a single sequence space by introducing a fresh symbol to separate the individual episodes.)
The tuples $\States \times \stateactions \times \States$ correspond to the current state, chosen action, and sampled successor state, respectively.
The $\sigma$-field $\algorithmsigma$ is obtained analogously to the $\sigma$-field for Markov chains by considering cylinder sets induced by finite prefixes, see \cite[Section~2.1.6]{DBLP:books/wi/Puterman94}.
For a given prefix, its probability can be obtained by computing the probability of each episode occurring in the MDP given the current state of the algorithm.

Now that we defined the probability space these algorithms operate in, we can define notions like almost sure convergence.
\begin{definition}
	Denote by $\learningalgo(\varepsilon)$ the instance of learning algorithm $\learningalgo$ with precision $\varepsilon$.
	We say that $\learningalgo$ \emph{converges (almost) surely} if, for every MDP $\MDP$, starting state $\initialstate$, target set $T$, and precision $\varepsilon > 0$, the computation of $\learningalgo(\varepsilon)$ terminates (almost) surely (w.r.t.\ $\ProbabilityAlgo$) and yields $\varepsilon$-optimal values $l$ and $u$.
\end{definition}
We consider a symbolic input encoding, where the MDP's properties are specified implicitly.
In particular, we design our algorithms such that they are applicable when the available actions $\stateactions$ and transition function $\mdptransitions$ are given as oracles.
This means that given a state $s$ we can compute $\stateactions(s)$, and given a state-action pair $(s, a)$ we obtain the successor distribution $\mdptransitions(s, a)$.
This allows us to achieve sub-linear runtime for some classes of MDP w.r.t.\ their number of states and transitions.
Note that most practical modelling languages such as the PRISM language \cite{DBLP:conf/cav/KwiatkowskaNP11} or JANI \cite{DBLP:conf/tacas/BuddeDHHJT17} describe models in such a way.

Since our learning algorithms in essence only rely on being able to repeatedly sample the system, we can drastically reduce the knowledge needed about the system.
In particular, we consider the setting of \emph{limited information}, where the algorithm only has very restricted access to the system in question.
There, we are only provided with bounds on some properties of the MDP, e.g., the number of states, together with a minimal interaction mechanism.
Concretely, we only get an oracle revealing the currently available actions and a \enquote{sampling} oracle, which upon choosing one of the available actions moves the system into a successor state, sampled according to the underlying, hidden distributions.
The algorithm thus can only simulate an execution of the MDP starting from the initial state $\initialstate$, repeatedly choosing an action from the set of available actions and querying the sampling oracle for a successor.
This corresponds to a \enquote{black-box} setting, where we can easily interact with a system and observe the current state, but have very limited knowledge about its internal transition structure, as might be the case with complex physical systems.

Here, we cannot directly apply the ideas of Q-learning, since the value of the sampled successor might not correspond to the actual value of the action.
Instead, the algorithm remembers the result of recent visits, \emph{delaying} the learning update.
Intuitively, by seeing many sampling results, we can get a stochastic estimate of the distribution of successor values.
In particular, the average of these observations corresponds to the true value with high confidence.
This idea is exploited by \emph{delayed Q-learning} \cite{DBLP:conf/icml/StrehlLWLL06}.
In this setting, we inherently cannot guarantee almost sure convergence, instead we demand that the algorithm terminates correctly with sufficiently high probability, specified by the \emph{confidence} $\delta > 0$.
\begin{definition}
	Denote by $\learningalgo(\varepsilon, \delta)$ the instance of learning algorithm $\learningalgo$ with precision $\varepsilon$ and confidence $\delta$.
	We say that $\learningalgo$ is \emph{probably approximately correct} (PAC) if for every MDP $\MDP$, starting state $\initialstate$, target set $T$, precision $\varepsilon > 0$, and confidence $\delta > 0$, with probability at least $1 - \delta$ the computation of $\learningalgo(\varepsilon, \delta)$ terminates and yields $\varepsilon$-optimal values $l$ and $u$.
	In other words, we require that the set of correct and terminating executions has a measure of at least $1 - \delta$ under $\ProbabilityAlgo$.
\end{definition}
Note that the \enquote{confidence} parameter $\delta$ sometimes is used to refer to the probability of error and sometimes for the probability of correct results.
We deliberately use $\delta$ for the probability of error to slightly simplify notation.
See \cite{DBLP:journals/cacm/Valiant84,DBLP:conf/colt/Angluin88,DBLP:conf/icml/StrehlLWLL06,DBLP:conf/isaim/Strehl08} for several, slightly different variants of PAC.
Some (but not all) definitions also require that the result is obtained within a particular time-bound (called \emph{efficient PAC-MDP} in \cite{DBLP:conf/isaim/Strehl08}).
We prove appropriate bounds for both variants of our PAC approach.
\begin{remark}
	We assume the system to be \enquote{observable} in both settings, i.e.\ the algorithm can access the \emph{precise} current state of the system and the set of available actions.
	Extending our methods to \emph{partially observable} systems, e.g.\ POMDP, is left for future work.
	Moreover, we also assume that the system can be repeatedly \enquote{reset} into the initial configuration $\initialstate$.
\end{remark}
 	\section{Complete Information -- MDP without End Components} \label{sec:brtdp_no_ec}

In this section, we treat the case of complete information, i.e.\ the algorithm has full access to the system, in particular its transition function $\mdptransitions$.
Additionally, we assume that the system has no MECs except two distinguished terminal states.
This greatly simplifies the reachability problem and allows us to gradually introduce our approach.
In \cref{sec:brtdp}, we explain the issue of MECs (see \cref{example:mec_no_convergence}) and extend our approach to general MDP.

\subsection{The Ideas of Value Iteration}

Our approach is based on ideas related to \emph{value iteration} (VI) \cite{howard1960dynamic}.
Thus, we first explain the basic principles of VI.
Value iteration is a technique to solve, among others, reachability queries on MDP.
It essentially amounts to applying \emph{Bellman iteration} \cite{bellman1966dynamic} corresponding to the fixed point equation in Equation~\eqref{eq:value_fixpoint} \cite[Section~4.2]{DBLP:conf/sfm/ForejtKNP11}.
In particular, starting from an initial value vector $v_0$ with $v_0(s) = 1$ if $s \in \targetset$ and $0$ otherwise, we apply the iteration
\begin{equation*}
	v_{n+1}(s) = \begin{dcases*}
		1 & if $s \in \targetset$, \\
		{\max}_{a \in \stateactions(s)} \ExpectedSumMDP{\mdptransitions}{s}{a}{v_n} & otherwise.
	\end{dcases*}
\end{equation*}
It is known that this iteration converges to the true value $\val$ in the limit from below, i.e.\ for all states $s$ we have (i)~$\lim_{n \to \infty} v_n(s) = \val(s)$ and (ii)~$v_n(s) \leq v_{n+1}(s) \leq \val(s)$ for all iterations $n$ \cite[Theorem~7.2.12]{DBLP:books/wi/Puterman94}\footnote{Note that reachability is a special case of \emph{expected total reward}, obtained by assigning a one-time reward of $1$ to each goal state.}.
It is not difficult to construct a system where convergence up to a given precision takes exponential time \cite{DBLP:conf/rp/HaddadM14}, but in practice VI often is much faster than methods based on \emph{linear programming} (LP)\footnote{See \cite[Theorem~10.105]{DBLP:books/daglib/0020348} for an LP-based solution of reachability.} \cite{DBLP:conf/tacas/HartmannsJQW23}, which in theory has worst-case polynomial runtime and yields precise answers \cite{DBLP:journals/combinatorica/Karmarkar84}.
An important practical issue of VI is the absence of a \emph{stopping criterion}, i.e.\ a straightforward way of determining in general whether the current values $v_n(s)$ are close to the true value function $\val(s)$, as discussed in, e.g., \cite[Section~4.2]{DBLP:conf/sfm/ForejtKNP11}.
As already hinted at, we solve this problem by additionally computing upper bounds, converging to the true value from above.

While the classical value iteration approach updates all states synchronously, the iteration can also be executed \emph{asynchronously}.
This means that we do not have to update the values of all states (or state-action pairs) simultaneously.
Instead, the update order may be chosen by heuristics, as long as fairness constraints are satisfied, i.e.\ eventually all states get updated.
This observation is essential for our approach, since we want to focus our computation on \enquote{important} areas.

\subsection{The No-EC BRTDP Algorithm} \label{sec:brtdp_no_ec:algo}

With these ideas in mind, we are ready to present our first algorithm.
Throughout this section, fix a required precision $\varepsilon > 0$, an MDP $\MDP = (\States, \Actions, \stateactions, \mdptransitions)$ with two distinguished states $\targetstate, \sinkstate \in \States$, target set $\targetset = \{\targetstate\}$, and a starting state $\initialstate$.
We assume that $\MDP$ has no MECs except the two terminal states $\targetstate$ and $\sinkstate$.
\begin{assumption} \label{asm:mec_free}
	MDP $\MDP$ has no MECs, except two trivial ones comprising the target state $\targetstate$ and sink state $\sinkstate$, respectively.
	Formally, we require that $\Mecs(\MDP) = \{(\{\targetstate\}, \stateactions(\targetstate)), (\{\sinkstate\}, \stateactions(\sinkstate))\}$.
\end{assumption}

\begin{algorithm}[t]
	\caption{The BRTDP learning algorithm for MDPs without ECs.}
	\label{alg:brtdp_no_ec}
	\setcounter{AlgoLine}{0}
	\DontPrintSemicolon
	\KwIn{MDP $\MDP$, state $\initialstate$, precision $\varepsilon$, and initial bounds $\upperbound_1$ and $\lowerbound_1$.}
	\KwOut{$\varepsilon$-optimal values $(l, u)$, i.e., $\val(\initialstate) \in [l, u]$ and $0 \leq u - l < \varepsilon$.}
	
	$\algoepisode \gets 1$ \tcp*{Initialize}
	
	\While{$\upperbound_\algoepisode(\initialstate) - \lowerbound_\algoepisode(\initialstate) \geq \varepsilon$}{ \label{alg:brtdp_no_ec:while}
		$\finitepath_\algoepisode \gets \samplepath(\MDP, \initialstate, \upperbound_\algoepisode, \lowerbound_\algoepisode, \varepsilon)$ \tcp*{Sample pairs to update} \label{alg:brtdp_no_ec:sample}
		
		$\upperbound_{\algoepisode + 1} \gets \upperbound_\algoepisode$, $\lowerbound_{\algoepisode + 1} \gets \lowerbound_\algoepisode$ \;
		
	\ForAll(\tcp*[f]{Update the upper and lower bounds}){$(s,a) \in \finitepath_\algoepisode$ }{$\upperbound_{\algoepisode+1}(s, a) \gets \ExpectedSumMDP{\mdptransitions}{s}{a}{\upperbound_\algoepisode}$ \label{alg:brtdp_no_ec:update_u} \;
			$\lowerbound_{\algoepisode+1}(s, a) \gets \ExpectedSumMDP{\mdptransitions}{s}{a}{\lowerbound_\algoepisode}$ \label{alg:brtdp_no_ec:update_l} \;
		}
		
		$\algoepisode \gets \algoepisode + 1$ \;
	}
	
	\Return $(\lowerbound_\algoepisode(\initialstate), \upperbound_\algoepisode(\initialstate))$ \;
\end{algorithm}
Observe that with \cref{asm:mec_free} and $\targetset = \{\targetstate\}$, we have $\val(\targetstate) = 1$ and $\val(\sinkstate) = 0$.

We present our \emph{BRTDP} approach in \cref{alg:brtdp_no_ec}.
As already mentioned in the introduction, the algorithm repeatedly samples sets of state-action pairs from the system.
Based on these experiences, it updates the upper and lower bounds using \emph{Bellman updates} (or \emph{Bellman backups}), corresponding to Equation~\eqref{eq:value_fixpoint}, until convergence.
(Recall that $\upperbound(s) = \max_{a \in \stateactions(s)} \upperbound(s, a)$ and $\lowerbound(s)$ analogously.)

To allow for practical optimization, we leave the sampling method $\samplepath$ undefined and instead only require some generic properties.
A simple implementation is given by sampling a path starting in the initial state and following random actions.
However, $\samplepath$ may use randomization and sophisticated guidance heuristics, as long as it satisfies certain conditions in the limit (formally defined in \cref{asm:brtdp_no_ec:fair}).

\begin{remark} \label{rem:brtdp_no_ec:sample_pair}
	We highlight that $\samplepath$ is not even required to return paths.
	Instead, it can yield any set of state-action pairs.
	However, when dealing with the limited information setting, we require sampling paths.
	Thus, it may be instructive to already think of $\samplepath$ as a procedure returning paths.
\end{remark}

\subsection{Proof of Correctness}

In this section, we prove correctness of the algorithm, i.e.\ that the returned result is correct and that the algorithm terminates.
We now first establish correctness of the result, assuming that the received input is sane.
\begin{assumption} \label{asm:brtdp_no_ec:input_correct}
	We have that (i) the given initial bounds $\upperbound_1$ and $\lowerbound_1$ are correct, i.e.\ $\lowerbound_1(s, a) \leq V(s, a) \leq \upperbound_1(s, a)$ for all $(s, a) \in \States \times \stateactions$, and (ii)~$\lowerbound_1(\targetstate) = 1$ and $\upperbound_1(\sinkstate) = 0$.
\end{assumption}
\begin{lemma} \label{stm:brtdp:no_ec:bounds_correct}
	Assume that \cref{asm:brtdp_no_ec:input_correct} holds.
	Then, during any execution of \cref{alg:brtdp_no_ec} we have for every episode $\algoepisode$ and all state-action pairs $(s, a)$ that
	\begin{equation*}
		\lowerbound_\algoepisode(s, a) \leq \lowerbound_{\algoepisode + 1}(s, a) \leq \val(s, a) \leq \upperbound_{\algoepisode + 1}(s, a) \leq \upperbound_\algoepisode(s, a).
	\end{equation*}
\end{lemma}
\begin{proof}
	Initially, we have that $\lowerbound_1(s, a) \leq \val(s, a) \leq \upperbound_1(s, a)$ by \cref{asm:brtdp_no_ec:input_correct}.
	The updates in \cref{alg:brtdp_no_ec:update_u,alg:brtdp_no_ec:update_l} clearly preserve these inequalities by Equation~\eqref{eq:value_fixpoint}.
	A simple inductive argument concludes the proof. 
\end{proof}

\begin{lemma} \label{stm:brtdp:no_ec:correct}
	Assume that \cref{asm:brtdp_no_ec:input_correct} holds.
	Then, the result $(l, u)$ of \cref{alg:brtdp_no_ec} is correct, i.e.\ (i)~$0 \leq u - l < \varepsilon$, and (ii)~$\val(\initialstate) \in [l, u]$.
\end{lemma}

\begin{proof}
	Clearly, (i)~immediately follows from \cref{stm:brtdp:no_ec:bounds_correct} and the main loop condition in \cref{alg:brtdp_no_ec:while}.
	Similarly, (ii)~also follows from \cref{stm:brtdp:no_ec:bounds_correct}. 
\end{proof}
In order to prove (almost sure) convergence of \cref{alg:brtdp_no_ec}, we need some assumptions on $\samplepath$.
Intuitively, $\samplepath$ may not neglect actions which might be the optimal ones.
In order to allow for a wide range of implementations for $\samplepath$, we present the rather liberal but technical condition of \emph{fairness} in \cref{asm:brtdp_no_ec:fair}.
We further explain each part of this assumption in the following proof of convergence.

Before we continue to the assumption, we introduce a concept, namely the set of $\upperbound$-optimal actions, which is also used in the proof.
We define the set of actions optimal w.r.t.\ $\upperbound_\algoepisode$ in state $s$ during episode $\algoepisode$ as $\umaxactions_\algoepisode(s) \coloneqq {\argmax}_{a \in \stateactions(s)} \upperbound_\algoepisode(s, a)$.
If the algorithm does not converge, the set $ \umaxactions_{\algoepisode}(s)$ may change infinitely often.
For example, two equivalent actions may get updated in an alternating fashion.
Thus, for each state $s$, we also define the set of actions that are optimal infinitely often as $\umaxactions_\infty(s) \coloneqq {\Intersection}_{k = 1}^\infty {\Union}_{\algoepisode = k}^\infty \umaxactions_{\algoepisode}(s)$.
This set is non-empty, since there are only finitely many actions and $\umaxactions_\algoepisode(s)$ is non-empty for any episode $\algoepisode$.
\begin{assumption} \label{asm:brtdp_no_ec:fair}
	Let $\{\upperbound_\algoepisode\}_{\algoepisode=1}^\infty$ and $\{\lowerbound_\algoepisode\}_{\algoepisode=1}^\infty$ be consistent sequences of upper and lower bounds, i.e.\ $\lowerbound_1(s, a) \leq \lowerbound_2(s, a) \leq \dots \leq \val(s, a) \leq \dots \leq \upperbound_2(s, a) \leq \upperbound_2(s, a)$ for all state-action pairs $(s, a)$.
	Assume that each call $\samplepath(\MDP, \initialstate, \upperbound_\algoepisode, \lowerbound_\algoepisode, \varepsilon)$ terminates in finite time and let $\finitepath_1, \finitepath_2, \dots \in \powerset(\States \times \Actions) \setminus \emptyset$ the infinite sequence of non-empty state-action sets obtained from it.

	Set $\States_\infty = {\Intersection}_{k = 1}^\infty {\Union}_{\algoepisode = k}^\infty \{s \in \States \mid s \in \finitepath_{\algoepisode}\}$ the set of all states which occur infinitely often, analogous for the set of actions occurring infinitely often, denoted $\Actions_\infty$.
	Then
	\begin{enumerate}
		\item \label{asm:brtdp_no_ec:fair:initial}
		the initial state is sampled infinitely often, i.e.\ $\initialstate \in \States_\infty$,

		\item \label{asm:brtdp_no_ec:fair:actions}
		all actions which are optimal infinitely often are also sampled infinitely often, i.e.\ $\umaxactions_\infty(s) \subseteq \Actions_\infty$ for every $s \in \States_\infty$, and

		\item \label{asm:brtdp_no_ec:fair:successors}
		all successors of optimal actions are sampled infinitely often, i.e.\ for every $s \in \States_\infty$ and $a \in \umaxactions_\infty(s)$ we have that $\support(\mdptransitions(s, a)) \subseteq \States_\infty$.
	\end{enumerate}
\end{assumption}
We say $\samplepath$ almost surely satisfies \cref{asm:brtdp_no_ec:fair}, if all of its conditions hold with probability 1.

In essence, the assumption requires that all states which are reachable by following optimal actions are indeed reached infinitely often in the limit:
Starting from the initial state (\cref{asm:brtdp_no_ec:fair:initial}), we select each optimal action infinitely often (\cref{asm:brtdp_no_ec:fair:actions}) and explore all successors of these actions (\cref{asm:brtdp_no_ec:fair:successors}).
For each of these successors, we again select all optimal actions, etc.
This insight directly yields an implementation for $\samplepath$, namely to repeatedly sample a path, starting in the initial state and in each state selecting any optimal action from $\umaxactions_\algoepisode(s)$ uniformly at random, until $\targetstate$ or $\sinkstate$ are reached.
Variants of this implementation can, for example, select actions in a round-robin fashion or sample from the optimal actions in a weighted manner.
Similarly, naively selecting all state-action pairs in every iteration (effectively classical value iteration) or selecting a single pair at random would also satisfy the assumption.
\begin{lemma} \label{stm:brtdp:no_ec:termination}
	\cref{alg:brtdp_no_ec} terminates under \cref{asm:mec_free,asm:brtdp_no_ec:input_correct,asm:brtdp_no_ec:fair}.
	It terminates almost surely if \cref{asm:brtdp_no_ec:fair} is satisfied almost surely.
\end{lemma}
\begin{proof}
	We prove the second case, i.e.\ almost sure termination, by contradiction.
	Assume that \cref{asm:mec_free,asm:brtdp_no_ec:input_correct} hold, and that \cref{asm:brtdp_no_ec:fair} holds a.s.
	Further, assume for contradiction that the set of non-terminating executions of \cref{alg:brtdp_no_ec} has non-zero measure.
	Since we assume that each call to \samplepath{} terminates in finite time (\cref{asm:brtdp_no_ec:fair}) a.s., the only way \cref{alg:brtdp_no_ec} does not terminate is when the central while-loop is executed infinitely often, i.e.\ the bounds never converge.

	Given some execution of \cref{alg:brtdp}, define $\bounddifference_\algoepisode(s, a) \coloneqq \upperbound_\algoepisode(s, a) - \lowerbound_\algoepisode(s, a)$.
	Fix an arbitrary action $ a_\algoepisode^{\max}(s) \in \umaxactions_\algoepisode(s)$ for each episode $\algoepisode$.
	Clearly, for any such action $ a_\algoepisode^{\max}(s)$ we have $\bounddifference_\algoepisode(s,  a_\algoepisode^{\max}(s)) = \upperbound_\algoepisode(s) - \lowerbound_\algoepisode(s,  a_\algoepisode^{\max}) \geq \upperbound_\algoepisode(s) - \lowerbound_\algoepisode(s)$.
	By \cref{stm:brtdp:no_ec:bounds_correct}, the limits $\upperbound_\infty(s, a) \coloneqq \lim_{\algoepisode \to \infty} \upperbound_\algoepisode(s, a)$ and $\lowerbound_\infty(s, a) \coloneqq \lim_{\algoepisode \to \infty} \lowerbound_\algoepisode(s, a)$ are well-defined and finite for any state-action pair $(s, a)$.
	Thus, $\bounddifference(s, a) \coloneqq \lim_{\algoepisode \to \infty} \bounddifference_\algoepisode(s, a)$ and $\bounddifference(s) \coloneqq \limsup_{\algoepisode \to \infty} \bounddifference_\algoepisode(s,  a_\algoepisode^{\max}(s))$ is also well-defined and finite.
	We prove that $\bounddifference(\initialstate) = 0$ for almost all executions, contradicting the assumption, as then necessarily $\upperbound_\algoepisode(\initialstate) - \lowerbound_\algoepisode(\initialstate) \leq \bounddifference_\algoepisode(\initialstate) < \varepsilon$ for some $\algoepisode$ a.s.

	Observe that the preconditions of \cref{asm:brtdp_no_ec:fair} are satisfied through \cref{stm:brtdp:no_ec:bounds_correct} and \cref{asm:brtdp_no_ec:input_correct}, hence we have $\initialstate \in \States_\infty$ a.s.\ \plabel{proof:brtdp_no_ec:termination:initial_state_infinitely_visited}.
	Let $\States_\infty$ the set of states seen infinitely often as defined in \cref{asm:brtdp_no_ec:fair}.
	By the assumption, we also have that $\support(\mdptransitions(s, a)) \subseteq \States_\infty$ for all $s \in \States_\infty, a \in \umaxactions_\infty(s)$ a.s.\ \plabel{proof:brtdp_no_ec:termination:infinite_states_actions}.

	Now, we identify a witness action $a_{\bounddifference}(s)$ for the $\limsup$ of $\bounddifference(s)$, i.e.\ an action $a_{\bounddifference}(s)$ such that $ \bounddifference_\infty(s) = \lim_{\algoepisode \to \infty} \bounddifference_\algoepisode(s, a_{\bounddifference}(s))$ and then derive a fixed-point equation.
	We have $\upperbound_\infty(s, a) = \upperbound_\infty(s, a')$ for all $s \in \States_\infty$ and $a, a' \in \umaxactions_\infty(s)$, as otherwise one of the two actions would not be optimal eventually.
	Consequently, $\lim_{\algoepisode \to \infty} \upperbound_\algoepisode(s, a_\algoepisode^{\max})$ is well-defined and equals $\upperbound_\infty(s, a)$ for any $a \in \umaxactions_\infty(s)$.
	Equally, $\limsup_{\algoepisode \to \infty} \lowerbound_\algoepisode(s, a_\algoepisode^{\max})$ also is well-defined, since $\lowerbound_\algoepisode$ is bounded.
	Hence the $\limsup$ of $\bounddifference(s)$ distributes over the minus.
	Recall that for each state-action pair, the limit of $\lowerbound_\infty(s, a)$ is well-defined.
	As there are only finitely many actions, the sequence $\lowerbound_\algoepisode(s, a_\algoepisode^{\max})$ only has finitely many accumulation points and there necessarily exists an action $a_{\bounddifference}(s) \in \umaxactions_\infty(s)$ such that $\limsup_{\algoepisode \to \infty} \lowerbound_\algoepisode(s, a_\algoepisode^{\max})  = \lowerbound_\infty(s, a_{\bounddifference}(s))$.
	Together, we have that $\bounddifference(s) = \upperbound_\infty(s, a_{\bounddifference}(s)) - \lowerbound_\infty(s, a_{\bounddifference}(s))$.
	Since all states $\States_\infty$ and all optimal actions $\umaxactions_\infty$ are visited infinitely often, we have that $\upperbound_\infty(s, a) = \ExpectedSumMDP{\mdptransitions}{s}{a}{\upperbound_\infty}$ and $\lowerbound_\infty(s, a) = \ExpectedSumMDP{\mdptransitions}{s}{a}{\lowerbound_\infty}$ for all $s \in \States_\infty$ and $a \in \umaxactions_\infty(s)$ by the back-propagation in \cref{alg:brtdp_no_ec:update_u,alg:brtdp_no_ec:update_l}---if not, they would get updated.
	Consequently, $\bounddifference(s) = \ExpectedSumMDP{\mdptransitions}{s}{a_{\bounddifference}(s)}{\bounddifference}$ for all $s \in \States_\infty$, since $a_{\bounddifference}(s) \in \umaxactions_\infty(s)$ \plabel{proof:brtdp_no_ec:termination:witness_action}.

	Finally, we use \cref{asm:mec_free} together with the above equation to show that $\bounddifference(\initialstate) = 0$.
	Let the maximal difference $\bounddifference_{\max} = \max_{s \in \States_\infty} \bounddifference(s)$ and define the witness states $\States_\bounddifference = \{s \in \States_\infty \mid \bounddifference(s) = \bounddifference_{\max}\}$.
	Assume for contradiction that $\bounddifference > 0$ (a.s.).
	Then, clearly $\targetstate, \sinkstate \notin \States_\bounddifference$, as $\bounddifference(\targetstate) = \bounddifference(\sinkstate) = 0$ by \cref{stm:brtdp:no_ec:bounds_correct} (the bounds of the special states are both set to $1$ or $0$ initially, respectively) and \cref{asm:brtdp_no_ec:input_correct} (bounds are monotone).
	Consequently, $\States_\bounddifference$ cannot contain any EC by \cref{asm:mec_free} (the MDP is MEC-free).
	Since $\States_\bounddifference$ does not contain an EC, there exists some state $s \in \States_\bounddifference$ such that for all $a \in \stateactions(s)$ we have $\support(\mdptransitions(s, a)) \not\subseteq \States_\bounddifference$.
	In other words, for each action $a \in \stateactions(s)$, there exists a state $s_a$ with both $s_a \notin \States_\bounddifference$ and $\mdptransitions(s, a, s_a) > 0$.
	By definition of $\States_\bounddifference$ (all states with maximal difference), we have that $\bounddifference(s_a) < \bounddifference_{\max}$.
	In particular, $\bounddifference(s, {a_{\bounddifference}(s)}) < \bounddifference(s)$ \plabel{proof:brtdp_no_ec:termination:difference_state}.
	We abbreviate the witness action from \ref{proof:brtdp_no_ec:termination:witness_action} by $\overline{a} \coloneqq a_{\bounddifference}(s)$.
	Then
	\begin{align*}
		\bounddifference(s) & ~\overset{\mathclap{\ref{proof:brtdp_no_ec:termination:witness_action}}}{=}~ \ExpectedSumMDP{\mdptransitions}{s}{\overline{a}}{\bounddifference_{\max}} = {\sum}_{s' \in \States} \mdptransitions(s, \overline{a}, s') \cdot \bounddifference(s') \\
			& ~\overset{\mathclap{\ref{proof:brtdp_no_ec:termination:infinite_states_actions}}}{=}~ {\sum}_{s' \in \States_\infty} \mdptransitions(s, \overline{a}, s') \cdot \bounddifference(s') \\
			& ~=~ {\sum}_{s' \in \States_\infty \setminus \{s_{\overline{a}}\}} \mdptransitions(s, \overline{a}, s') \cdot \bounddifference(s') + \mdptransitions(s, \overline{a}, s_{\overline{a}}) \cdot \bounddifference(s_{\overline{a}}) \\
			& ~\leq~ {\sum}_{s' \in \States_\infty \setminus \{s_{\overline{a}}\}} \mdptransitions(s, \overline{a}, s') \cdot \bounddifference_{\max} + \mdptransitions(s, \overline{a}, s_{\overline{a}}) \cdot \bounddifference(s_{\overline{a}}) \\
			& ~\overset{\mathclap{\ref{proof:brtdp_no_ec:termination:difference_state}}}{<}~ {\sum}_{s' \in \States_\infty \setminus \{s_{\overline{a}}\}} \mdptransitions(s, \overline{a}, s') \cdot \bounddifference_{\max} + \mdptransitions(s, \overline{a}, s_{\overline{a}}) \cdot \bounddifference_{\max} \\
			& ~=~ \bounddifference_{\max},
	\end{align*}
	contradicting $s \in \States_\bounddifference$, i.e.\ $\bounddifference(s) = \bounddifference_{\max}$, and we have that $\bounddifference_{\max} = 0$.
	To conclude the proof, observe that $\States_\bounddifference = \States_\infty$ a.s., as $0 \leq \bounddifference(s) \leq \bounddifference_{\max} = 0$ for all $s \in \States_\infty$, and $\bounddifference(\initialstate) = 0$ a.s., since $\initialstate \in \States_\infty$ a.s.\ by \ref{proof:brtdp_no_ec:termination:initial_state_infinitely_visited}.

	Guaranteed convergence (instead of \enquote{only} almost sure) follows analogously. 
\end{proof}
As an immediate consequence of \cref{stm:brtdp:no_ec:correct} (correctness) and \cref{stm:brtdp:no_ec:termination} (termination), we get the desired result.
\begin{theorem}
	Assume that \cref{asm:mec_free,asm:brtdp_no_ec:input_correct}, and (almost surely) \cref{asm:brtdp_no_ec:fair} hold.
	Then \cref{alg:brtdp_no_ec} is correct and converges (almost surely).
\end{theorem}
\begin{remark}
	If an implementation of $\samplepath$ satisfies \cref{asm:brtdp_no_ec:fair} only almost surely, we can easily obtain a surely terminating variant by interleaving it with a deterministic sampling procedure, e.g., a round-robin method.
\end{remark}
\begin{figure}
  \centering
     \includegraphics[scale=1.4]{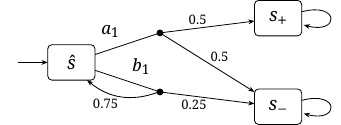}



	\caption{Example MDP where following the upper bounds is wrong.} \label{fig:upper_strat_is_bad}
\end{figure}
\begin{example} \label{exa:upper_strat_is_bad}
	Interestingly, following the optimal upper bound does not necessarily yield an $\varepsilon$-optimal strategy, as shown by the MDP in \cref{fig:upper_strat_is_bad}.
	Assume that initially we take action $a_1$, setting $\upperbound_2(\initialstate, a_1) = \lowerbound_2(\initialstate, a_1) = \frac{1}{2}$.
	Then, $\upperbound_2(\initialstate, b_1) = 1 > \upperbound_2(\initialstate, a_1)$ and we sample $b_1$, updating $\upperbound_3(\initialstate, b_1) = \frac{3}{4}$,  $\upperbound_4(\initialstate, b_1) = \frac{3}{4} \cdot \frac{3}{4}$, etc.
	This continues until the upper bound of $b_1$ is $\varepsilon$-close to $\frac{1}{2}$, when the algorithm terminates.
	Now, suppose that instead of $\mdptransitions(\initialstate, b_1, \sinkstate) = \frac{1}{4}$ exactly, we have $\mdptransitions(\initialstate, b_1, \sinkstate) = p$.
	Then, $\upperbound_i(\initialstate, b_1) = (1 - p)^{i - 1}$.
	For a fixed $\varepsilon$, choose $p$ such that $\frac{1}{2} < (1-p)^k < \frac{1}{2} + \varepsilon$ for some $k$.
	This means that in episode $\algoepisode = k + 1$ (where the algorithm terminates) we have $\upperbound_{\algoepisode}(\initialstate, b_1) > \upperbound_{\algoepisode}(\initialstate, a_1)$.
	Yet, following $b_1$ yields a (highly) suboptimal value, namely $0$ instead of $\frac{1}{2}$.

	It is straightforward to also apply this example to our DQL approach and as a counterexample to \cite[Lemma~16]{DBLP:conf/atva/BrazdilCCFKKPU14}.  
\end{example}
Following the maximal \emph{lower} bound yields a strategy achieving at least this value, using results on asynchronous VI \cite{DBLP:books/wi/Puterman94}.
We omit formal treatment of this claim, since we are not concerned with extracting a witness strategy to avoid distraction from the main result.
(Note that it is in general not correct to choose an arbitrary \emph{value-optimal} action, i.e.\ any action $\argmax_{a \in \stateactions(s)} \val(s, a)$.)
 	\section{Complete Information -- General Case} \label{sec:brtdp}

\begin{figure}[t]
  \centering
     \includegraphics[scale=1.4]{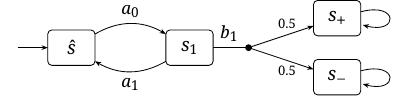}



	\caption{Example MDP with an EC where \cref{alg:brtdp_no_ec} does not converge.} \label{example:mec_no_convergence:figure}
\end{figure}

In this section, we deal with the case of general MDP, in particular, we allow for arbitrary ECs.
We first illustrate with an example the additional difficulties arising when considering general MDPs with non-trivial ECs.
In particular, \cref{alg:brtdp_no_ec} does not converge, even on a small example.

\begin{example} \label{example:mec_no_convergence}
	Consider the MDP depicted in \cref{example:mec_no_convergence:figure}.
	Clearly, we can reach the goal $\targetset = \{\targetstate\}$ with probability $\frac{1}{2}$ by playing $a_0$ in $\initialstate$ and then $b_1$ in $s_1$.
	But the EC $(\{\initialstate, s_1\}, \{a_0, a_1\})$ causes issues for \cref{alg:brtdp_no_ec}.
	When running the algorithm on this example MDP, we eventually have that $\upperbound(s_1, b_1) = \lowerbound(s_1, b_1) = \frac{1}{2}$, but $\upperbound(s_1, a_1) = 1$, since $\upperbound(\initialstate) = 1$.
	Similarly, we keep $\upperbound(\initialstate, a_0) = 1$, as $\upperbound(s_1) = 1$.
	Informally, $\initialstate$ and $s_1$ \enquote{promise} each other that the target state might still be reachable with probability 1, but these promises depend on each other cyclically.
	Removing the internal behaviour of this EC and \enquote{merging} $\initialstate$ and $s_1$ into a single state (with only action $b_1$) solves this issue.
\end{example}
In general, by definition of ECs, every state inside an EC can be reached from any other state with probability 1.
Since we are interested in (unbounded) reachability, this means that for an EC there can only be two cases.
Either, the EC contains a target state.
Then, reaching any state of the EC is (a.s.) equivalent to reaching the target already and we do not need to treat the internal transitions of the EC further.
Otherwise, i.e.\ when the EC does not contain a target state, we can also omit treatment of its internal behaviour and only consider its interaction with outside states.
For the remainder of the section, fix an arbitrary MDP $\MDP = (\States, \Actions, \stateactions, \mdptransitions)$, starting state $\initialstate$, target set $\targetset$, and precision $\varepsilon > 0$.
\begin{lemma} \label{stm:ec_same_value}
	Let $(R, B) \in \Ecs(\MDP)$ be an EC of $\MDP$.
	Then, $\ProbabilityMDPmax<\MDP,s>[\reach \{s'\}] = 1$ for any states $s, s' \in R$ and consequently $\ProbabilityMDPmax<\MDP,s>[\reach \targetset] = \ProbabilityMDPmax<\MDP,s'>[\reach \targetset]$ for any target set $\targetset \subseteq \States$.
\end{lemma}

\begin{proof}
	Follows directly from \cite[Lemma~1]{DBLP:conf/qest/CiesinskiBGK08} (observe that the first claim is a special case of the second claim with $\targetset = \{s'\}$). 
\end{proof}

In other words, states in the same EC are equivalent for reachability and we can apply a quotienting construction w.r.t.\ to ECs.
This idea has been exploited by the \emph{MEC quotient} construction \cite{de1997formal,DBLP:conf/qest/CiesinskiBGK08,DBLP:conf/rp/HaddadM14}, a preprocessing step where first all MECs are identified and then \enquote{collapsed} into a representative state.
However, this approach requires that the whole graph structure of the MDP is known.
Constructing the whole graph of the system may be prohibitively expensive or even impossible, as, e.g., in our limited knowledge setting (see \cref{def:limited_information}).
Hence, we propose a modification to the BRTDP algorithm, which detects and handles ECs \enquote{on-the-fly}.
The algorithm will repeatedly identify ECs and maintain a separate, simplified MDP, which is similar to a MEC quotient.

\subsection{Collapsing End Components} \label{sec:brtdp:collapsing}
As already explained, collapsing an EC can be viewed as replacing it with a single representative state, omitting the internal behaviour of the EC.
In the following definition, we introduce the \emph{collapsed MDP}, where end components are merged into representative states.
Moreover, we again introduce the special states $\targetstate$ and $\sinkstate$, acting as a target and sink respectively, to avoid corner cases.
Many statements in this section are similar to \cite[Section~6.4]{de1997formal} but adapted to our particular use case.
Note that our definition of collapsed MDP in particular depends on the target set $\targetset$.
\begin{definition}
	Let $\algoecs = \{(R_1, B_1), \dots, (R_n, B_n)\} \subseteq \Ecs(\MDP)$ be a (possibly empty) set of ECs in $\MDP$ with $R_i, B_i \neq \emptyset$ and pairwise disjoint.
	Define $R_{\algoecs} = \Union_i R_i$ and $B_{\algoecs} = \Union_i B_i$ the set of all states and actions in $\algoecs$, respectively.
	
	The \emph{collapsed MDP} is defined as $\MDP^c = (\States^c, \Actions^c, \stateactions^c, \mdptransitions^c) = \collapse(\MDP, \algoecs, \initialstate, \targetset)$,
	\begin{itemize}
		\item $\States^c = \States \setminus R_{\algoecs} \union \{s_{(R_i,B_i)}\} \union \{\targetstate, \sinkstate\}$, where $s_{(R_i, B_i)} \notin \States$ are new \emph{representative} states, $\targetstate$ is the new target state, and $\sinkstate$ is a new sink state,
		\item $\Actions^c = \Actions \setminus B_{\algoecs} \union \{\ecremain_i\} \union \{a_+, a_-\}$, where $\ecremain_i \notin \Actions$ are new \emph{remain} actions (one per state, as we assume actions to be uniquely associated with one state),
		\item $\stateactions^c(s)$ is defined by
		\begin{itemize}
			\item $\stateactions^c(s) = \stateactions(s)$ for $s \in \States \setminus R_{\algoecs}$,\footnote{Recall that actions in $B_{\algoecs}$ are only available for states in $R_{\algoecs}$, hence $\stateactions(s) \subseteq \Actions^c$ for other states.}
			\item $\stateactions^c(s_{(R_i,B_i)}) = \Union_{s \in R_i} \stateactions(s) \setminus B_i \union \{\ecremain_i\}$,
			\item $\stateactions^c(\targetstate) = \{a_+\}$, $\stateactions'(\sinkstate) = \{a_-\}$, and
		\end{itemize}
		\item $\mdptransitions^c$ is defined by ($\ecstates$ is an auxiliary function defined below)
		\begin{itemize}
			\item $\mdptransitions^c(s^c, a^c, s'^c) = \sum_{s' \in \ecstates(s'^c)} \mdptransitions(\actionstate<\MDP>(a^c), a^c, s')$ for $s^c, s'^c \in \States^c \setminus \{\targetstate, \sinkstate\}$ and $a^c \in \stateactions^c(s^c) \intersection B$,
			\item $\mdptransitions^c(s_{(R_i,B_i)}, \ecremain_i) = \{\targetstate \mapsto 1\}$ if $\targetset \intersection R_i \neq \emptyset$ and $\{\sinkstate \mapsto 1\}$ otherwise, and
			\item $\mdptransitions^c(\targetstate, a_+, \targetstate) = 1$, $\mdptransitions'(\sinkstate, a_-, \sinkstate) = 1$,
		\end{itemize}
	\end{itemize}
	with the following auxiliary functions
	\begin{itemize}
		\item $\eccollapsed : \States \to \States^c$ maps states of $\MDP$ to their corresponding state in the collapsed MDP, i.e.\ $\eccollapsed(s) = s_{(R_i, B_i)}$ if $s \in R_i$ for some $i$ and $\eccollapsed(s) = s$ otherwise,
		\item $\ecstates : \States^c \setminus \{\targetstate, \sinkstate\} \to 2^\States$ maps states in the collapsed MDP to the set of states they represent, i.e.\ $\ecstates(s^c) = R_i$ if $s^c = s_{(R_i,B_i)}$ for some $i$ and $\ecstates(s^c) = \{s^c\} \subseteq \States$ otherwise,
		\item $\ecequivalent : \States \to 2^\States$ maps states of $\MDP$ to all states in their EC, i.e.\ $\ecequivalent(s) = R_i$ if $s \in R_i$ for some $i$ and $\ecequivalent(s) = \{s\}$ otherwise.
	\end{itemize}
	Note that $\ecequivalent(s) = \ecstates(\eccollapsed(s))$.
	For ease of notation, we extend these auxiliary functions to sets of states in the obvious way, i.e.\ $\eccollapsed(R) = \{\eccollapsed(s) \mid s \in R\}$, $\ecstates(R^c) = \Union_{s^c \in R^c} \ecstates(s^c)$, and $\ecequivalent(R) = \Union_{s \in R} \ecequivalent(s)$.
	Finally, if $\initialstate \in R_i$ for some $i$, we identify $\initialstate$ with $s_{(R_i, B_i)}$ for ease of notation.
	This guarantees that we always have $\initialstate \in \States^c$.
\end{definition}
\begin{figure}[t]
  \centering
     \includegraphics[scale=1.4]{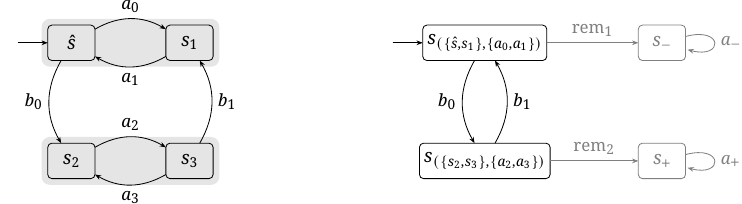}






	\caption[Example of an MDP and its collapsed version.]{Example of an MDP (left) and its collapsed version (right) with $\targetset = \{s_2\}$ and $\algoecs = \{(\{\initialstate, s_1\}, \{a_0, a_1\}), (\{s_2, s_3\}, \{a_2, a_3\})\}$.} \label{example:ec_collapse}
\end{figure}\noindent See \cref{example:ec_collapse} for an example of a collapsed MDP.
Observe that given a set $\algoecs$ explicitly, the collapsed MDP can be computed on-the-fly, i.e.\ without constructing the original MDP completely.
In particular, for a state $s$ in the MDP $\MDP$, we can compute the corresponding state $s^c = \eccollapsed(s)$ as well as $\stateactions^c(s^c)$ and $\mdptransitions^c(s^c, a^c)$ for all actions $a \in \stateactions^c(s^c)$, based on the given set $\algoecs$.

Now, we prove some useful properties about the collapsed MDP.
These properties are rather intuitive, however the corresponding proofs are surprisingly technical without revealing relevant insights.
Thus, the proofs may be skipped.
In essence, we prove that (i)~there is a correspondence of paths between the original and the collapsed MDP, (ii)~there is a correspondence of ECs between the two MDPs, and, most importantly, (iii)~the reachability probability is equal on the two MDPs.

Fix a collapsed MDP of $\MDP$ as $\MDP^c = (\States^c, \Actions^c, \stateactions^c, \mdptransitions^c) = \collapse(\MDP, \algoecs, \initialstate, \targetset)$ for the remainder of this section, where $\algoecs = \{(R_i,B_i)\}_{i=1}^n$ is any appropriate set of end components.

\begin{lemma} \label{stm:collapse:collapse_actionstate}
	We have that $\eccollapsed(\actionstate<\MDP>(a)) = \actionstate<\MDP^c>(a)$ for all $a \in \Actions \intersection \Actions^c$.
\end{lemma}

\begin{proof}
	First, observe that $\Actions \intersection \Actions^c = \Actions^c \setminus \{a_+, a_-, \ecremain_i\}$ by definition.
	The claim follows by a case distinction on $s^c = \actionstate<\MDP^c>(a)$.
	If $s^c \in \States$, then $\stateactions(s^c) = \stateactions^c(s^c)$ and $\eccollapsed(s^c) = s^c$.
	If instead $s^c = s_{(R_i, B_i)}$ for some $(R_i, B_i) \in \algoecs$, we have that $a \in \Union_{s \in R_i} \stateactions(s) \setminus B_i$.
	Thus, there exists a state $s \in R_i$ such that $s = \actionstate<\MDP>(a)$.
	But then by definition $\eccollapsed(s) = s^c$. 
\end{proof}
The following two lemmas show how we can relate paths in the two MDPs with each other.
See \cite[Section~6.4.1]{de1997formal} for an alternative view.
Intuitively, the collapsed MDP also gives us a \enquote{quotient} on the set of paths.
Essentially, a continuous sequence of state-action pairs belonging to the same EC from $\algoecs$ is \enquote{collapsed} to the corresponding representative.
Vice-versa, any path in the collapsed MDP corresponds to a set of paths in the original MDP.
\begin{lemma} \label{stm:collapse:path_normal_to_collapsed}
	Let $\finitepath = s_1 a_1 \dots a_{n-1} s_n \in \Finitepaths<\MDP>$ be a finite path in the MDP $\MDP$.
	There exists a number $m \leq n$ and indices $i_1, \dots, i_m$ with $1 \leq i_j < i_{j+1} \leq n$ such that
    \begin{equation*}
        \finitepath^c = \eccollapsed(s_{i_1}) a_{i_1} \dots a_{i_{m-1}} \eccollapsed(s_{i_m}) \in \Finitepaths<\MDP^c>    
    \end{equation*}
    is a finite path in the collapsed MDP $\MDP^c$ with $\eccollapsed(s_1) = \eccollapsed(s_{i_1})$ and $\eccollapsed(s_n) = \eccollapsed(s_{i_m})$.
\end{lemma}
\begin{proof}
	We construct the path $\finitepath^c$ inductively.
	We start with $i_1 = 1$ and $s_1^c = \eccollapsed(s_1)$.
	Now, either all actions of $\finitepath$ are in $B_\algoecs$, then by definition of ECs all states of $\finitepath$ are within the same EC and we are done.
	Otherwise, let $a$ be the first action along the path $\finitepath$ such that $a \in \Actions^c$ (i.e.\ $a \notin B_\algoecs$) and let its index equal $j$.
	Set $i_2 = j$, $a_1^c = a$ and $s_2^c = \eccollapsed(s_{i+1})$.
	Then $a \in \stateactions(s_1^c)$.
	Repeat the argument with the path $\finitepath'$ equal to the suffix of $\finitepath$ starting at $j + 1$. 
\end{proof}
\begin{lemma} \label{stm:collapse:path_collapsed_to_normal}
	Let $\finitepath^c = s^c_1 a^c_1 \dots a^c_{m-1} s^c_m \in \Finitepaths<\MDP^c>$ be a finite path in the collapsed MDP $\MDP^c$ not containing the special states $\targetstate$, $\sinkstate$.
	There exists a finite path $\finitepath = s_1 a_1 \dots a_{n-1} s_n \in \Finitepaths<\MDP>$ in the MDP $\MDP$ with $n \geq m$ and indices $i_1, \dots, i_m$ with $1 \leq i_j < i_{j+1} \leq n$ and
	\begin{itemize}
		\item $s_k \in \ecstates(s^c_j)$ for all $j$ and $k$ with $i_j \leq k < i_{j+1}$ (defining $i_{m+1} = n + 1$) and
		\item if $s^c_j = s_{(R_i, B_i)}$ then $a_k \in B_i$ for all $j$ and $k$ with $i_j \leq k < i_{j+1} - 1$.
	\end{itemize}
\end{lemma}
\begin{proof}
	Similar to the above proof, we construct the path $\finitepath$ inductively.
	Distinguish two cases for $s^c_1$.
	If $s^c_1 \in \States$, set $s_1 = s^c_1$ and $a_1 = a^c_1$ and repeat the argument with the next step of $\finitepath^c$.
	Otherwise, we have that $s^c_1 = s_{(R_i, B_i)}$ for some EC $(R_i, B_i) \in \algoecs$.
	Since $(R_i, B_i)$ is an EC in $\MDP$, there exists a finite path in $\Finitepaths<\MDP>$ only using actions of $B_i$ from any state in $R_i$ to $\actionstate<\MDP>(a^c_1)$.
	This path corresponds to the first state-action pair in $\finitepath^c$.
	By definition, there exists a state $s' \in \States$ such that $s' \in \support(\mdptransitions(\actionstate<\MDP>(a^c_1), a^c_1))$ and $\eccollapsed(s') = s^c_2$.
	Thus, we can extend the above path by $a^c_1 s'$ and repeat the argument. 
\end{proof}
Based on the previous lemmas, we can establish a correspondence of end components between the original MDP and its (partly) collapsed version.
In particular, for every EC in the original MDP there either exists a single state representing this EC or a new EC in the collapsed MDP.
\begin{lemma} \label{stm:collapse:ec_correspondence_normal_to_collapse}
	For any EC $(R, B) \in \Ecs(\MDP)$ in the MDP $\MDP$ we either have
	\begin{enumerate}
		\item \label[casedistinction]{item:stm:collapse:ec_correspondence_normal_to_collapse:true_ec}
		an EC $(R^c, B^c)$ in $\MDP^c$, where $R^c = \eccollapsed(R)$ and $B^c = B \intersection \Actions^c$, or
		\item \label[casedistinction]{item:stm:collapse:ec_correspondence_normal_to_collapse:state}
		a state $s_{(R',B')} \in \States^c$ with $R \subseteq R'$ and $B \subseteq B'$.
	\end{enumerate}
\end{lemma}

\begin{proof}
	Observe that \cref{item:stm:collapse:ec_correspondence_normal_to_collapse:state} is trivial by definition, in particular this case is equivalent to $B \subseteq B_i$ for some $i$.
	Moreover, \cref{item:stm:collapse:ec_correspondence_normal_to_collapse:true_ec} and \cref{item:stm:collapse:ec_correspondence_normal_to_collapse:state} are mutually exclusive since by construction for any EC $(R_i, B_i)$ the internal actions $B_i$ are removed, thus there is no $B \subseteq B_i$ such that $(\{s_{(R_i, B_i)}\}, B)$ is an EC in $\MDP^c$.

	Let thus $(R, B)$ be an EC in the MDP $\MDP$ with $B \not\subseteq B_i$ for all $i$.
	We show that $(R^c, B^c)$ with $R^c = \eccollapsed(R)$ and $B^c = B \intersection \Actions^c$ is an EC in $\MDP^c$.

	First, we show by contradiction that $B \not\subseteq B_{\algoecs}$ \plabel{proof:collapse:ec_correspondence:B_not_subset}, i.e.\ $B$ cannot comprise only internal actions of the ECs in $\algoecs$.
	Recall that by assumption on $\algoecs$ the EC states $R_i$ are disjoint and $B_i$ are subsets of the actions enabled in the respective states of $R_i$.
	Since we assume not to be in \cref{item:stm:collapse:ec_correspondence_normal_to_collapse:state}, $(R, B)$ is an EC with $B \not\subseteq B_i$ for all $i$.
	Assume for contradiction that $B \subseteq B_{\algoecs} = \Union B_i$.
	Then $(R, B)$ necessarily has to contain states of at least two ECs from $\algoecs$.
	Formally, there exist two states $s, s' \in R$ with $s \in R_i$, $s' \in R_j$, and $i \neq j$.
	Since $(R, B)$ is an EC, there exists a path from $s$ to $s'$ and vice versa, using only actions from $B$.
	As $B \subseteq B_{\algoecs}$, these actions were available in the ECs before.
	Since $s$ and $s'$ are in two ECs with disjoint state sets and a path using only actions from $B$ exists between them, there exists a state $s''$ and action $a \in B \subseteq B_{\algoecs}$ with $\support(\mdptransitions(s'', a)) \not\subseteq R_i$.
	Since the $a \in B_{\algoecs}$, we necessarily have $a \in B_i$, contradicting the assumption that $(R_i, B_i)$ is an EC, proving \ref{proof:collapse:ec_correspondence:B_not_subset}.

	Next, we prove that $R^c = \Union_{a \in B^c} \actionstate<\MDP^c>(a)$ \plabel{proof:collapse:ec_correspondence:R_is_action_union}.
	Observe that by assumption we have $R = \Union_{a \in B} \actionstate<\MDP>(a)$.
	By definition of $B^c = B \intersection \Actions^c$, we thus have that $\Union_{a^c \in B^c} \actionstate<\MDP>(a^c) \subseteq R$.
	Consequently 
	\begin{equation*}
		{\Union}_{a^c \in B^c} \eccollapsed(\actionstate<\MDP>(a^c)) \subseteq \eccollapsed(R) = R^c
	\end{equation*}
	Applying \cref{stm:collapse:collapse_actionstate} yields
    \begin{equation*}
        {\Union}_{a^c \in B^c} \eccollapsed(\actionstate<\MDP>(a^c)) = {\Union}_{a^c \in B^c} \actionstate<\MDP^c>(a^c),    
    \end{equation*}
    thus $\Union_{a^c \in B^c} \actionstate<\MDP^c>(a^c) \subseteq R^c$.

	Now, assume for contradiction that there exists a state $s^c \in R^c$ such that $s^c \neq \actionstate<\MDP^c>(a^c)$ for all $a^c \in B^c$.
	Due to the definition of $\MDP^c$, we either have $s^c \in \States$, $s^c = s_{(R', B')}$ for some EC $(R', B') \in \algoecs$, or $s^c \in \{\targetstate, \sinkstate\}$.
	The third case immediately leads to a contradiction, since $B^c \subseteq \Actions$ and thus $a_+, a_- \notin B^c$.
	In the first case, we have that $s^c \notin R_i$ for any $i$, thus $\stateactions(s^c) = \stateactions^c(s^c) \subseteq \Actions^c$.
	Hence, any action $a$ of this state contained in the EC $(R, B)$ is still available in the collapsed MDP and thus also contained in the EC $(R^c, B^c)$.
	The second case implies, by definition of $R^c = \eccollapsed(R)$, that there exists an EC $(R_i, B_i) \in \algoecs$ such that $R_i \intersection R \neq \emptyset$.
	Recall that $\stateactions^c(s_{(R_i, B_i)}) = \Union_{s \in R_i} \stateactions(s) \setminus B_i$.
	The case assumption is thus equivalent to $B^c \intersection (\Union_{s \in R_i} \stateactions(s) \setminus B_i) = \emptyset$.
	Inserting the definition of $B^c$ and $\Actions^c$ yields
	\begin{multline*}
		B \intersection (\Actions \setminus B_{\algoecs}) \intersection ({\Union}_{s \in R_i} \stateactions(s) \setminus B_i) = B \intersection ({\Union}_{s \in R_i} \stateactions(s) \setminus B_i) = \\
		{\Union}_{s \in R_i \intersection R} \stateactions(s) \intersection B \setminus B_i = \emptyset.
	\end{multline*}
	This implies that $\stateactions(s) \intersection B \subseteq B_i$ for all $s \in R_i \intersection R$, i.e.\ all such states only have \enquote{internal} actions of the EC $(R_i, B_i)$ available in $(R, B)$.
	But this implies $R \subseteq R_i$ and $B \subseteq B_i$, contradicting our assumptions.
	This concludes the proof of \ref{proof:collapse:ec_correspondence:R_is_action_union}.

	Now, we prove that $(R^c, B^c)$ is a proper EC in $\MDP^c$, i.e.\ that (i)~$R^c \neq \emptyset$, $\emptyset \neq B^c \subseteq \Union_{s^c \in R^c} \stateactions(s^c)$, (ii)~for all $s^c \in R^c$, $a \in B^c \intersection \stateactions^c(s^c)$ we have $\support(\mdptransitions^c(s^c, a^c)) \subseteq R^c$, and (iii)~for all states $s^c, s'^c \in R^c$ there exists a path from $s^c$ to $s'^c$ only using actions from $B^c$.

	For (i), we have $B^c \neq \emptyset$, otherwise $B^c = B \intersection \Actions^c = \emptyset$ implies $B \subseteq B_{\algoecs}$, contradicting \ref{proof:collapse:ec_correspondence:B_not_subset}.
	\ref{proof:collapse:ec_correspondence:R_is_action_union} yields the second part of the first condition.

	To prove (ii), assume a contradiction, i.e.\ let $s^c \in R^c$, $a \in B^c \intersection \stateactions^c(s^c)$ such that $s'^c \in \support(\mdptransitions^c(s^c, a^c)) \setminus R^c$.
	Let $s = \actionstate<\MDP>(a^c)$ (implying $s^c = \eccollapsed(s)$).
	Again, we proceed by a case distinction, this time on the successor $s'^c$.
	If $s'^c \in \States$, we have that $s'^c \in \support(\mdptransitions(s, a^c))$, since $s \in R$ and $a^c \in B$ and $(R, B)$ is an EC.
	Further, $\mdptransitions^c(s^c, a^c, s'^c) = \mdptransitions(s^c, a^c, s'^c)$, thus $s'^c \in \support(\mdptransitions^c(s^c, a^c))$, contradicting the assumption.
	If instead $s'^c = s_{(R_i, B_i)}$, then there exists a state $s' \in \support(\mdptransitions(s, a^c)) \intersection R_i$ by definition of $\mdptransitions^c$.
	But then $s_{(R_i, B_i)} \in R^c$ by definition of $R^c$, contradiction.

	Finally, to show (iii), we can directly apply \cref{stm:collapse:path_normal_to_collapsed} to obtain the required path as follows.
	Let $s^c, s'^c \in R^c$ two states and pick two arbitrary $s, s' \in R$ with $\eccollapsed(s) = s^c$ and $\eccollapsed(s') = s'^c$.
	Since $(R, B)$ is an EC, there exists a finite path $\finitepath$ from $s$ to $s'$, using only actions of $B$.
	By \cref{stm:collapse:path_normal_to_collapsed}, we get a path $\finitepath^c$ from $s^c$ to $s'^c$ using only actions from $B \intersection \Actions^c = B^c$, concluding the proof of \cref{item:stm:collapse:ec_correspondence_normal_to_collapse:true_ec}. 
\end{proof}
As expected, the corresponding reverse statement holds true, too, i.e.\ every EC in the collapsed MDP yields a corresponding EC in the original MDP.
\begin{lemma} \label{stm:collapse:ec_correspondence_collapse_to_normal}
	For all ECs $(R^c, B^c)$ in $\MDP^c$ with $\targetstate, \sinkstate \notin R^c$ we have that $(R, B)$ with $R = \ecstates(R^c)$ and $B = B^c \union \Union_{s_{(R_i, B_i)} \in R^c} B_i$ is an EC in $\MDP$.
\end{lemma}

\begin{proof}
	Fix an EC $(R^c, B^c)$ in $\MDP^c$ and set $R = \ecstates(R^c)$ and $B = B^c \union \Union_{s_{(R_i, B_i)} \in R^c} B_i$.
	We need to prove that $(R, B)$ is an EC in $\MDP$.
	Clearly, $R$ and $B$ are non-empty.
	We show that $R = \Union_{a \in B} \actionstate<\MDP>(a)$.
	For any $s \in R$, there exists a $s^c \in R^c$ such that $s \in \ecstates(s^c)$ by definition of $R$.
	If $s = s^c$ we have $s \in R^c$ and there exists an action $a^c \in B^c \subseteq B$ with $\actionstate<\MDP>(a^c) = s$.
	Otherwise, there is an EC $(R_i, B_i) \in \algoecs$ with $s \in R_i$, $s_{(R_i, B_i)} \in R^c$, and, since $(R_i, B_i)$ is in EC in $\MDP$, there is an action $a \in B_i \subseteq B$ with $\actionstate<\MDP>(a) = s$.
	Similarly, for any action $a \in B$ we have that $\actionstate<\MDP>(a) \in R$ by analogous reasoning.

	It remains to show that (i)~for all $s \in R$, $a \in B \intersection \stateactions(s)$ we have $\support(\mdptransitions(s, a)) \subseteq R$, and (ii)~for all $s, s' \in R$ there is a finite path from $s$ to $s'$ only using actions from $B$.
	For (i), we again assume contradiction, i.e.\ there are states $s \in R$, $s' \in \States$ and an action $a \in \stateactions(s) \intersection B$ such that $s' \in \support(\mdptransitions(s, a)) \setminus R$.
	We again proceed by case distinctions, but now first on $a$.
	If $a \in B^c$, then $\support(\mdptransitions^c(\eccollapsed(s), a)) \subseteq R^c$, as $(R^c, B^c)$ is an EC.
	By definition of $\mdptransitions^c$, we have $\eccollapsed(s') \in \support(\mdptransitions^c(\eccollapsed(s), a))$.
	Together, this implies $s' \in R$, yielding a contradiction.
	If instead $a \in B_i$ for some EC $(R_i, B_i) \in \algoecs$, then $s, s' \in R_i \subseteq R$, also leading to a contradiction.
	Finally, to prove (ii), we can directly apply \cref{stm:collapse:path_collapsed_to_normal} to a path from $\eccollapsed(s)$ to $\eccollapsed(s')$ in $(R^c, B^c)$, yielding a path from $s$ to $s'$ in $(R, B)$. 
\end{proof}
The previous statement implies that if we collapse a MEC of the original MDP, then there can be no EC in the collapsed MDP containing the MEC representative state.
\begin{lemma} \label{stm:collapse:no_ecs}
	Let $\{(R'_i, B'_i)\}_{i=1}^m \subseteq \algoecs \intersection \Mecs(\MDP)$ be some MECs of $\MDP$ in $\algoecs$.
	Then, we have that $s_{(R'_i, B'_i)} \notin R^c$ for any EC $(R^c, B^c)$ in $\MDP^c$.
\end{lemma}

\begin{proof}
	Assume there is such an EC $(R^c, B^c)$ with $s_{(R'_i, B'_i)} \in R^c$.
	\cref{stm:collapse:ec_correspondence_collapse_to_normal} yields an EC $(R, B)$ with $R'_i \subseteq R$, $B'_i \subsetneq B$, contradiction to $(R, B)$ being a MEC in $\MDP$. 
\end{proof}
The statement of \cref{stm:collapse:no_ecs} does not hold for any EC $(R'_i, B'_i) \in \algoecs$, since there might be a larger EC containing $s_{(R'_i, B'_i)}$.
For example, in \cref{example:ec_collapse}, the collapsed MDP has an EC containing representative states.
However, if all MECs are collapsed, the resulting collapsed MDP indeed has no ECs except two trivial ones.
\begin{corollary}
	Let $\MDP^c = \collapse(\MDP, \Mecs(\MDP), \initialstate, \targetset)$ be the collapsed MDP of $\MDP$ with $\algoecs = \Mecs(\MDP)$.
	Then, $\MDP^c$ satisfies \cref{asm:mec_free}.
\end{corollary}

\begin{proof}
	Follows directly from the above \cref{stm:collapse:no_ecs}. 
\end{proof}
Finally, we also get that the reachability probabilities are preserved.
\begin{lemma} \label{stm:collapse:reachability_equal}
	Let $\MDP^c = (\States^c, \Actions^c, \stateactions^c, \mdptransitions^c) = \collapse(\MDP, \algoecs, \initialstate, \targetset)$ be the collapsed MDP of $\MDP$, where $\algoecs = \{(R_i,B_i)\}_{i=1}^n$ is any appropriate set of end components.
	Then, for any state $s \in \States$ it holds that
	\begin{equation*}
		\ProbabilityMDPmax<\MDP, s>[\reach \targetset] = \ProbabilityMDPmax<\MDP^c, \eccollapsed(s)>[\reach \eccollapsed(\targetset)] = \ProbabilityMDPmax<\MDP^c, \eccollapsed(s)>[\reach (\{\targetstate\} \union (\targetset \intersection \States^c))].
	\end{equation*}
\end{lemma}

\begin{proof}
	First, observe that $\ProbabilityMDPmax<\MDP^c, s^c>[\reach \{\targetstate\}] = 1$ for any state $s^c = s_{(R_i, B_i)}$ with $R_i \intersection \targetset \neq \emptyset$ by definition.
	Moreover, $\targetset \intersection \States^c = \targetset \setminus R_\algoecs$, i.e.\ all target states which are not part of an EC in $\algoecs$.
	Every state $s^c \in \eccollapsed(\targetset)$ is of one of these two kinds.
	Hence, $\ProbabilityMDPmax<\MDP^c, \eccollapsed(s)>[\reach (\{\targetstate\} \union (\targetset \intersection \States^c))] = \ProbabilityMDPmax<\MDP^c, \eccollapsed(s)>[\reach \eccollapsed(\targetset)]$, proving the second equality.

	For the first equality, we argue how to transform the witness strategies, achieving the same overall reachability probability.
	Thus, let $\strategy \in \StrategiesMD<\MDP>$ be a (memoryless deterministic) strategy in $\MDP$ maximizing the probability of reaching $\targetset$.
	We define a strategy $\strategy^c$ on $\MDP^c$ simulating $\strategy$ as follows.
	Note that $\strategy^c$ does not have to be memoryless or deterministic.
	For all states $s^c \in \States$, i.e.\ $s^c$ is not a collapsed representative, $\strategy^c$ mimics $\strategy$, i.e.\ $\strategy^c(s) = \strategy(s)$.
	For the other case, namely $s^c = s_{(R_i, B_i)}$ for some EC $(R_i, B_i) \in \algoecs$, recall that $\strategy^c$ is allowed to have memory.
	In particular, it can remember the action $a$ leading to $s^c$.
	Clearly, for any such action $a$ and other action $a' \in \stateactions^c(s^c)$ we can compute the probability of $a'$ action being the first action not in $B_i$ under $\strategy$.
	Then, $\strategy^c$ simply selects $a'$ in $s^c$ after $a$ with this probability.
	Moreover, we also need to compute the probability of remaining inside $R_i$ forever, which corresponds to the probability of $\strategy^c$ choosing $\ecremain_i$.
	It is easy to see that $\strategy^c$ achieves the same reachability as $\strategy$.

	If we instead start with a strategy in the collapsed MDP $\strategy^c \in \StrategiesMD<\MDP^c>$, we construct the respective strategy $\strategy$ on $\MDP$ as follows.
	Again, on states $s \notin R_\algoecs$, we simply replicate the choice of $\strategy^c$.
	On states $s_{(R_i, B_i)}$ the strategy $\strategy^c$ chooses a single action $a^c \in \stateactions^c(s_{(R_i, B_i)})$, since it is deterministic.
	If that action is $\ecremain_i$, $\strategy$ simply picks any internal $a \in B_i$ in each state $R_i$.
	Otherwise, there exists a strategy $\strategy'$ on $R_i$ reaching state $\actionstate<\MDP>(a)$ with probability $1$.
	Thus, $\strategy$ mimics $\strategy'$ until that state is reached and then plays $a^c$, again achieving the same reachability. 
\end{proof}

\subsection{The General BRTDP Algorithm}

Now, we present our modification of \cref{alg:brtdp_no_ec}, using the idea of collapsing, to obtain the general approach as shown in \cref{alg:brtdp}.
On top of the previously presented ideas, the algorithm maintains a growing set of ECs and repeatedly collapses the input MDP.

 \begin{algorithm}[t]
 	\caption{The BRTDP learning algorithm for general MDPs.}
 	\label{alg:brtdp}
 	
 	\DontPrintSemicolon
 	\setcounter{AlgoLine}{0}  
 	
 	\KwIn{MDP $\MDP$, state $\initialstate$, target set $\targetset$, precision $\varepsilon$, initial bounds $\upperbound_1$ and $\lowerbound_1$, and initial set of ECs $\Ecs_1$.}
 	\KwOut{$\varepsilon$-optimal values $(l, u)$, i.e., $\val(\initialstate) \in [l, u]$ and $0 \leq u - l < \varepsilon$.}
 	
 	$\algoepisode \gets 1$, $\MDP^c_1 \gets \collapse(\MDP, \Ecs_1, \initialstate, \targetset)$ \label{alg:brtdp:line:initial_collapse} \;
 	$\upperbound_1(\targetstate, a_+) \gets 1$, $\lowerbound_1(\targetstate, a_+) \gets 1$, $\upperbound_1(\sinkstate, a_-) \gets 0$, $\lowerbound_1(\sinkstate, a_-) = 0$ \;
 	
 	\While{$\upperbound_\algoepisode(\initialstate) - \lowerbound_\algoepisode(\initialstate) \geq \varepsilon$}{
 		\ForAll( \tcp*[f]{Initialize bounds of representative states}){$(R_j, B_j) \in \Ecs_{\algoepisode}$}{ \label{alg:brtdp:line:collapse_bounds_update_start}
 			\ForAll(\tcp*[f]{Copy bounds for existing actions}){$a \in \stateactions(s_{(R_j,B_j)}) \setminus \{\ecremain_j\}$}{ 
 				$\upperbound_{\algoepisode}(s_{(R_j,B_j)}, a) \gets \upperbound_{\algoepisode}(\actionstate<\MDP>(a), a)$ \label{alg:brtdp:line:collapse_bounds_update_upper} \;
 				$\lowerbound_{\algoepisode}(s_{(R_j,B_j)}, a) \gets \lowerbound_{\algoepisode}(\actionstate<\MDP>(a), a)$ \label{alg:brtdp:line:collapse_bounds_update_lower} \;
 			}
 			\uIf(\tcp*[f]{Set bounds for remain action}){$R_j \intersection \targetset = \emptyset$}{ 
 				$\upperbound_{\algoepisode}(s_{(R_j,B_j)}, \ecremain_j) \gets 0$, $\lowerbound_{\algoepisode}(s_{(R_j,B_j)}, \ecremain_j) \gets 0$ \label{alg:brtdp:line:collapse_bounds_update_sink} \;
 			}
 			\Else{
 				$\upperbound_{\algoepisode}(s_{(R_j,B_j)}, \ecremain_j) \gets 1$, $\lowerbound_{\algoepisode}(s_{(R_j,B_j)}, \ecremain_j) \gets 1$ \label{alg:brtdp:line:collapse_bounds_update_target} \;
 			}
 		} \label{alg:brtdp:line:collapse_bounds_update_end}
 		
 		$\upperbound_{\algoepisode + 1} \gets \upperbound_\algoepisode$, $\lowerbound_{\algoepisode + 1} \gets \lowerbound_\algoepisode$ \;
 		$\finitepath \gets \samplepath(\MDP^c_\algoepisode, \initialstate, \upperbound_\algoepisode, \lowerbound_\algoepisode, \varepsilon)$ \tcp*{Sample a path in collapsed MDP}  		
 		\ForAll(\tcp*[f]{Update the upper and lower bounds}){$(s,a) \in \finitepath$}{ 
 			$\upperbound_{\algoepisode + 1}(s, a) \gets \ExpectedSumMDP{\mdptransitions}{s}{a}{\upperbound_\algoepisode}$ \label{alg:brtdp:line:update_u} \;
 			$\lowerbound_{\algoepisode + 1}(s, a) \gets \ExpectedSumMDP{\mdptransitions}{s}{a}{\lowerbound_\algoepisode}$ \label{alg:brtdp:line:update_l} \;
 		}
 		$\Ecs_{\algoepisode + 1} \gets \updateecs(\MDP, \Ecs_\algoepisode)$ \tcp*{Search for new ECs} 
 		$\MDP^c_{\algoepisode + 1} \gets \collapse(\MDP, \Ecs_{\algoepisode + 1}, \initialstate, \targetset)$ \tcp*{Update the collapsed MDP} \label{alg:brtdp:line:collapse} 		
 		$\algoepisode \gets \algoepisode + 1$ \;
 	}
 	
 	\Return $(\lowerbound_\algoepisode(\initialstate), \upperbound_\algoepisode(\initialstate))$ \;
 \end{algorithm}
 
The new auxiliary procedure $\updateecs$ is supposed to identify ECs in $\MDP$.
As with $\samplepath$, we only require some properties instead of giving a concrete implementation.
Essentially, $\updateecs$ should only grow its list of ECs and eventually identify all ECs which are repeatedly visited by $\samplepath$.

\begin{assumption} \label{asm:updateecs_growing}
	Let $\algoecs_1 \subseteq \Ecs(\MDP)$ be an initial set of state-disjoint ECs, $\algoecs_{\algoepisode+1} = \updateecs(\MDP, \algoecs_\algoepisode)$ the identified ECs, and $\MDP^c_\algoepisode = \collapse(\MDP, \algoecs_\algoepisode, \initialstate, \targetset)$ the corresponding collapsed MDPs.
	Then, for any episode $\algoepisode$ and EC $(R, B) \in \algoecs_\algoepisode$, $(R, B)$ is an EC of $\MDP$ and there exists $(R', B') \in \algoecs_{\algoepisode+1}$ with $R \subseteq R'$ and $B \subseteq B'$.
\end{assumption}
This is, for example, easily satisfied by searching for ECs in the set of visited states in every step.
However, an efficient implementation may want to choose the times when it actually searches heuristically.

Since there are only finitely many states, this assumption implies that eventually $\algoecs_\algoepisode$ and thus $\MDP^c_\algoepisode$ stabilizes, i.e.\ there exists some episode $\overline{\algoepisode}$ such that for all $\algoepisode \geq \overline{\algoepisode}$ we have that $\algoecs_\algoepisode = \algoecs_{\algoepisode+1}$ and thus $\MDP^c_\algoepisode = \MDP^c_{\algoepisode+1}$.
We call $\overline{\algoepisode}$ the \emph{EC-stable episode}.

\begin{assumption} \label{asm:updateecs_all_visited}
	Let $\algoecs_\algoepisode$ and $\MDP^c_\algoepisode$ as in \cref{asm:updateecs_growing} and assume that assumption holds.
	Further, let $\finitepath_\algoepisode \in \Finitepaths<\MDP^c_\algoepisode>$ be an infinite series of sets of state-action pairs in $\MDP^c_\algoepisode$ and define $\States^c_\infty = \Intersection_{k=1}^\infty \Union_{\algoepisode=k}^\infty \{s \in \States^c_\algoepisode \mid s \in \finitepath^c_\algoepisode\}$ the set of states occurring infinitely often.\footnote{As mentioned above, due to \cref{asm:updateecs_growing} we get a EC-stable episode $\overline{\algoepisode}$ and thus have $\States^c_\infty \subseteq \States_{\overline{\algoepisode}}^c$, i.e.\ the set of infinitely often seen states are all states of $\MDP_{\overline{\algoepisode}}^c$.}
	Then, there exists no EC $(R^c, B^c)$ in $\MDP^c_{\overline{\algoepisode}}$ with $R^c \subseteq \States_\infty^c$ except $R^c = \{\targetstate\}$ or $R^c = \{\sinkstate\}$.
\end{assumption}

\subsection{Proof of Correctness}

We now continue to prove correctness and termination of \cref{alg:brtdp}.
First, we argue that the algorithm indeed is well-defined, i.e.\ it never accesses undefined values.

\begin{lemma}
	\cref{alg:brtdp} is well-defined.
\end{lemma}

\begin{proof}
	We only need to show that the states introduced by the collapsing in \cref{alg:brtdp:line:initial_collapse,alg:brtdp:line:collapse} are assigned bounds before being accessed.
	By definition of the collapsed MDP, we add a state for each EC together with an additional action, and the special states $\{\targetstate, \sinkstate\}$.
	The initial collapse in \cref{alg:brtdp:line:initial_collapse} adds the special states together with their corresponding actions.
	Their values are initialised in the following line.
	Furthermore, the EC collapsing in \cref{alg:brtdp:line:initial_collapse,alg:brtdp:line:collapse} adds a state $s_{(R,B)}$ for any EC $(R,B) \in \Ecs_\algoepisode$ and a corresponding $\ecremain$ action.
	Their values are initialised in \cref{alg:brtdp:line:collapse_bounds_update_start,alg:brtdp:line:collapse_bounds_update_end} and not accessed prior to that. 
\end{proof}
As in \cref{asm:brtdp_no_ec:input_correct}, we again assume that the initial inputs are correct.
\begin{assumption} \label{asm:brtdp_input_correct}
	The given initial bounds $\upperbound_1$ and $\lowerbound_1$ are correct, i.e.\ $\lowerbound_1(s, a) \leq \val(s, a) \leq \upperbound_1(s, a)$ for all $s \in \States, a \in \stateactions(s)$.
	Furthermore, the given initial set of ECs is correct, i.e.\ $\algoecs_1 \subseteq \Ecs(\MDP)$ and pairwise disjoint.
\end{assumption}
\begin{lemma} \label{stm:brtdp:bounds_correct}
	Assume that \cref{asm:brtdp_input_correct} holds.
	Then, during any execution of \cref{alg:brtdp} we have for every episode $\algoepisode$, all states $s \in \States_\algoepisode$ and action $a \in \stateactions^c_\algoepisode(s)$ that
	\begin{equation*}
		\lowerbound_\algoepisode(s, a) \leq \lowerbound_{\algoepisode + 1}(s, a) \leq \val(s, a) \leq \upperbound_{\algoepisode + 1}(s, a) \leq \upperbound_\algoepisode(s, a).
	\end{equation*}
\end{lemma}

\begin{proof}
	We prove that the initialization of values for newly added states is correct.
	The remaining proof then is completely analogous to the proof of \cref{stm:brtdp:no_ec:bounds_correct}.

	Since $\targetstate$ is the target in $\MDP^c$, setting $\lowerbound_1(\targetstate, a_+) = 1$ is correct.
	Analogously, we see that $\sinkstate$ has no outgoing action and thus cannot reach $\targetstate$, justifying $\upperbound_1(\sinkstate, a_-) = 0$.

	The correctness of updates for the collapsed states follows from \cref{stm:collapse:reachability_equal}. 
\end{proof}

\begin{lemma} \label{stm:brtdp:correct}
	The result of \cref{alg:brtdp} is correct under \cref{asm:brtdp_input_correct}, i.e.\ (i)~$0 \leq u - l < \varepsilon$, and (ii)~$\val(\initialstate) \in [l, u]$.
\end{lemma}

\begin{proof}
	As in \cref{stm:brtdp:no_ec:correct}, the claims follows from the algorithm and \cref{stm:brtdp:bounds_correct}. 
\end{proof}

Finally, we can prove termination of our presented algorithm.
The proof is very similar to the proof of \cref{stm:brtdp:no_ec:termination} and we only need to incorporate the new assumptions about $\updateecs$.

\begin{lemma} \label{stm:brtdp:terminates}
	\cref{alg:brtdp} terminates under \cref{asm:brtdp_no_ec:fair,asm:updateecs_growing,asm:updateecs_all_visited,asm:brtdp_input_correct}.
	It terminates almost surely if \cref{asm:brtdp_no_ec:fair} is satisfied almost surely.
\end{lemma}

\begin{proof}
	We apply the same reasoning as in \cref{stm:brtdp:no_ec:termination} until \cref{asm:mec_free} is applied in the final part of the proof.
	Since we do not necessarily explore all of $\MDP$, $\MDP^c_\algoepisode$ may still contain MECs.
	In the proof, \cref{asm:mec_free} is used only to show that $\States_\bounddifference \subseteq \States_\infty$ does not contain MECs.
	Observe that any non-terminating execution eventually reaches an EC-stable episode $\overline{\algoepisode}$, thus the collapsed MDP considered by the algorithm does not change.
	Now, $\States_\infty$ in the previous proof exactly corresponds to $\States_\infty^c$ of \cref{asm:updateecs_all_visited}, which yields that again there is no EC in $\States_\infty^c$.
	Thus, we can continue to apply the previous proof's reasoning. 
\end{proof}
Again, we get the overall soundness as a direct consequence.
\begin{theorem}
	Assume that (almost surely) \cref{asm:brtdp_no_ec:fair}, as well as \cref{asm:updateecs_growing,asm:updateecs_all_visited,asm:brtdp_input_correct} hold.
	Then \cref{alg:brtdp} is correct and converges (almost surely).
\end{theorem}

\subsection{Relation to Interval Iteration}

We briefly outline how our BRTDP algorithm presented in \cref{alg:brtdp} generalizes both the original BRTDP algorithm of \cite{DBLP:conf/atva/BrazdilCCFKKPU14} and the interval iteration algorithm of \cite{DBLP:conf/rp/HaddadM14}.
To this end, we give a brief overview of interval iteration.
The algorithm first identifies all MECs and constructs a quotient similar to the one we presented in \cref{sec:brtdp:collapsing}.
Then, each state is initialised with straightforward upper and lower bounds.
These bounds then are iterated globally according to the Bellman operator.
We can emulate this behaviour by directly yielding the set of all MECs in $\updateecs$ and returning $\States^c \times \stateactions^c$ on each call to $\samplepath$.
All variants of \cite{DBLP:conf/atva/BrazdilCCFKKPU14} can be obtained by choosing the appropriate path sampling heuristics for $\samplepath$.
 	\section{Limited Information -- MDP without End Components} \label{sec:dql_no_ec}

We adapt our approach to the setting of limited information, where we can access the system only as a \enquote{black box} and we are given some bounds on the shape of the system (see \cref{sec:preliminaries:learning}).
Intuitively, since we are interested in an $\varepsilon$-precise solution, we can repeatedly sample the system to learn the transition probabilities with high confidence.
By adapting our previous ideas, we can enhance this approach to only learn \enquote{interesting} transitions.
Since we can never bound the transition probabilities with absolute certainty, we aim for a \emph{probably approximately correct} algorithm, which gives an $\varepsilon$-optimal solution with probability at least $1 - \delta$.

\paragraph{Relevance and Applicability}
Before we go into the details, we discuss the purpose and motivation for the subsequent algorithms.
As mentioned in the introduction, our primary aim is to provide a \emph{possibility result}, showing that it is possible to obtain PAC bounds on the \emph{maximal} value on \emph{infinite horizon} reachability values in a \emph{black box} setting, only using samples of \emph{finite length} and only starting in the \emph{initial state}, and all this for \emph{general MDP} (with ECs) \emph{even in a model-free setting} (see below for a brief comment on model-free).
Due to this focus, the bounds that the presented approaches obtain are rather impractical and of mostly theoretical value.
This can be alleviated in several ways.
For one, tighter statistical methods could be used, see \cite{DBLP:journals/corr/abs-2404-05424} for a recent discussion (we use the na\"ive Hoeffding's inequality to simplify proofs).
Additionally, our approach is generic in the sense that it assumes the worst of the system.
Specific knowledge about the model, e.g.\ (in-)dependence of states, could be incorporated to significantly improve practical scalability.
Yet, these points are orthogonal to our aim of proving the possibility of (model-free) PAC, for which we provide a complete proof in the following.
Moreover, an additional aim of this work is to provide a re-usable framework for proofs in this direction.
We believe that several statements in the proofs below might be useful for other endeavours of this kind, especially the auxiliary statements in \cref{sec:appendix:auxiliary}.

\begin{remark} \label{rem:model_free}
	Intuitively, the idea of \enquote{model-free} is that such approaches do not try to learn the concrete transition probabilities or the entire graph structure, but more \enquote{compressed} quantities such as state- or action-values.
	Indeed, our algorithm only stores a fixed number of values per state-action pair, not for each transition.
	In most literature, model-free is only loosely defined, as it is difficult to formalize precisely \cite{DBLP:conf/icml/StrehlLWLL06}.
	In \cite[Definition~1]{DBLP:conf/icml/StrehlLWLL06}, the authors try to capture model-free by requiring that the space complexity of an approach should be $o(|\States|^2 |\Actions|)$ (in other words, less than the explicit graph representation of the MDP).
	At the same time, the space complexity naturally also depends on parameters such as $\varepsilon$ and $\delta$ (e.g.\ suppose that $\varepsilon$ were of exponential size w.r.t.\ the entire system).
	As such, we are interested in the above complexity for \emph{fixed} parameters.
	The estimates our algorithm obtains are based on repeated updates to action values.
	Later, in \cref{stm:dql_no_ec:successful_update_count}, we show that (for fixed parameters) the number of executed updates is bounded by $\cardinality{\Actions}$, and thus one can prove that the updates only involves numbers that are of size $\cardinality{\Actions}$.
	In any case, proving that our approach formally satisfies (one of the many) definition of model-free is not our main goal, but rather observing that it captures the \enquote{spirit} of model-free by not learning probabilities but rather values directly.

	For a model-based approach to this problem, we direct the reader to \cite{DBLP:conf/cav/AshokKW19,DBLP:journals/corr/abs-2404-05424}.
	These approaches essentially obtain bounds on every single transition probability in the system and then solve the induced \emph{interval MDP} to obtain bounds on the value.
\end{remark}

\subsection{Definition of Limited Information}
We define the limited information setting.
\begin{definition} \label{def:limited_information}
	Let $\MDP = (\States, \Actions, \stateactions, \mdptransitions)$ be some MDP, $\initialstate \in \States$ a starting state, and $T \subseteq \States$ a target set.
	An algorithm has \emph{limited information} if it can access
	\begin{itemize}
		\item the starting state $\initialstate$,
		\item a target oracle for $\targetset$, i.e.\ given a state $s$ it can query whether $s \in T$,
		\item an upper bound $A$ of the number of actions, $A \geq \cardinality{\Actions}$,
		\item a lower bound $q$ on the transition probabilities under any uniform strategy, $0 < q \leq p_{\min} = \min\{ \cardinality{\stateactions(s)}^{-1} \cdot \mdptransitions(s, a, s') \mid s \in \States, a \in \stateactions(s), s' \in \support(\mdptransitions(s, a)) \}$,
		\item an oracle for the set of available actions $\stateactions$, and
		\item a successor oracle $\successor$, which given a state-action pair yields a successor state, sampled according to the underlying, hidden probability distribution $\mdptransitions$.
	\end{itemize}
\end{definition}

To tackle this problem, we combine the BRTDP approach with \emph{delayed Q-learning} (DQL) \cite{DBLP:conf/icml/StrehlLWLL06}.
In essence, DQL temporarily accumulates sampled values for each state-action pair and only attempts an update after a certain delay, i.e.\ after enough samples have been gathered for a particular pair.
Intuitively, with a large enough delay, the average of the sampled values is close to the true average with high confidence.
Moreover, the attempted update is only successful if the value is changed by at least some margin.
If instead the update fails, another update is only allowed if any other value in the system has changed significantly.
This way, we can bound the total number of attempted updates and thus control the overall probability of any \enquote{wrong} update occurring.
We explain all these ideas in more detail later on.

\subsection{The No-EC DQL Algorithm}

First, we again restrict ourselves to the case of no end components, as these pose an additional difficulty.
Thus, we assume the MDP $\MDP$ satisfies \cref{asm:mec_free} and instead of a target state oracle, the algorithm is explicitly given the special states $\targetstate$ and $\sinkstate$.
We present our DQL-based approach in \cref{alg:dql_no_ec}.
While it is similar in spirit to \cref{alg:brtdp_no_ec}, we give a concrete instantiation of $\samplepath$, since this setting needs a lot of additional guarantees.

\begin{algorithm}[!tp]
  		\caption{The DQL learning algorithm for MDPs without ECs.}\label{alg:dql_no_ec}%
        \setcounter{AlgoLine}{0}%
  		\KwIn{Inputs as given in \cref{def:limited_information} satisfying \cref{asm:mec_free}, special states $\targetstate, \sinkstate$, precision $\varepsilon$, and confidence $\delta$.}%
  		\KwOut{Values $(l, u)$ which are $\varepsilon$-optimal, i.e., $\val(\initialstate) \in [l, u]$ and $0 \leq u - l < \varepsilon$, with probability at least $1 - \delta$.}%
  		$\upperbound_1(\cdot, \cdot) \gets 1$, $\lowerbound_1(\cdot, \cdot) \gets 0$, $\upperbound_1(\sinkstate, \cdot) \gets 0$, $\lowerbound_1(\targetstate, \cdot) \gets 1$\;
  		\For{$\circ \in \{\upperbound, \lowerbound\}$}{
  			$\learn_1^\circ(\cdot, \cdot) \gets \learnyes$, $\acc_1^\circ(\cdot, \cdot) \gets 0$, $\visitcount_1^\circ(\cdot, \cdot) \gets 0$\;
  		}
  		$\algoepisode \gets 1$, $\algostep \gets 1$\;
  		
  		\While{$\upperbound_\algostep(\initialstate) - \lowerbound_\algostep(\initialstate) \geq \varepsilon$}{
  			\lFor{$s \in \States$}{$\umaxactions_\algoepisode(s) \gets \argmax_{a \in \stateactions(s)} \upperbound_\algostep(s, a)$}
  			$s_\algostep \gets \initialstate$
  			
  			\While(\tcp*[f]{Experience the current learning episode}){$s_\algostep \notin \{\targetstate, \sinkstate\}$}{
  				$a_\algostep \gets $ sampled uniformly from $\umaxactions_\algoepisode(s_\algostep)$\tcp*{Pick an action}  
  				$s_\algostep' \gets \successor(s_{\algostep}, a_\algostep)$  \tcp*{Query successor oracle}
  			
  				\tcc{Update bound estimates}
  				\For{$\circ \in \{\upperbound, \lowerbound\}$}{
  					\If{$\learn_\algostep^\circ(s_\algostep, a_\algostep) \neq \learnno$}{
  						$\visitcount_{\algostep + 1}^\circ(s_\algostep, a_\algostep) \gets \visitcount_\algostep^\circ(s_\algostep, a_\algostep) + 1$\;
  						$\acc_{\algostep + 1}^\circ(s_\algostep, a_\algostep) \gets  \acc_\algostep^\circ(s_\algostep, a_\algostep) + \bigcirc_\algostep(s_\algostep')$\;
  					}
  				}
  				
  				\tcc{Learn upper bounds}
  				\If(\tcp*[f]{Attempt update of $\upperbound$}){$\visitcount_{\algostep + 1}^\upperbound(s_\algostep, a_\algostep) = \delay$}{
  					\If{$\acc_{\algostep + 1}^\upperbound(s_\algostep, a_\algostep) / \delay < \upperbound_\algostep(s_\algostep, a_\algostep) - 2 \updatestep$}{
  						$\upperbound_{\algostep + 1}(s_\algostep, a_\algostep) \gets  \acc_{\algostep + 1}^\upperbound(s_\algostep, a_\algostep) / \delay + \updatestep$\label{alg:dql_no_ec:line:update_upperbound_value} \tcp*{Successful update}
  						$\learn_{\algostep + 1}^\upperbound(\cdot, \cdot) \gets \learnyes$\label{alg:dql_no_ec:line:reset_learn_upper}  \tcp*{Re-enable learning for all actions}
  					}
  					\Else{
  						$\learn_{\algostep + 1}^\upperbound(s_\algostep, a_\algostep) \gets \learndecrease(\learn_{\algostep}^\upperbound(s_\algostep, a_\algostep))$\tcp*{Failed update}
  					}
  					$\visitcount_{\algostep + 1}^\upperbound(s_\algostep, a_\algostep) \gets 0$, $\acc_{\algostep + 1}^\upperbound(s_\algostep, a_\algostep) \gets 0$\;
  				}
  				
  				\tcc{Learn lower bounds}
  				\If(\tcp*[f]{Attempt update of $\lowerbound$}){$\visitcount_{\algostep + 1}^\lowerbound(s_\algostep, a_\algostep) = \delay$}{
  					\If{$\acc^\lowerbound_{\algostep + 1}(s_\algostep, a_\algostep)  / \delay > \lowerbound_\algostep(s_\algostep, a_\algostep) + 2 \updatestep$}{
  						$\lowerbound_{\algostep + 1}(s_\algostep, a_\algostep) \gets \acc^\lowerbound_{\algostep + 1}(s_\algostep, a_\algostep) / \delay - \updatestep$ \tcp*{Successful update}
  						$\learn_{\algostep + 1}^\lowerbound(\cdot, \cdot) \gets \learnyes$ \label{alg:dql_no_ec:line:reset_learn_lower} \tcp*{Re-enable learning for all actions}
  					}
  					\Else{
  						$\learn_{\algostep + 1}^\lowerbound(s_\algostep, a_\algostep) \gets \learndecrease(\learn_{\algostep}^\lowerbound(s_\algostep, a_\algostep))$\tcp*{Failed update}
  					}
  					$\visitcount_{\algostep + 1}^\lowerbound(s_\algostep, a_\algostep) \gets 0$, $\acc_{\algostep + 1}^\lowerbound(s_\algostep, a_\algostep) \gets 0$\;
  				}
  				
  				$s_{\algostep + 1} \gets s_\algostep'$, $\algostep \gets \algostep + 1$\tcp*{Increase step counter}
  			}
  			$\algoepisode \gets \algoepisode + 1$\tcp*{Increase episode counter}
  		}
  		\Return $(\lowerbound_\algostep(\initialstate), \upperbound_\algostep(\initialstate))$\;
  	\end{algorithm}

The algorithm contains several auxiliary variables.
Most are values kept for each state-action pair, and separate for both the upper and lower bound.
We give a brief intuition for each variable, where $\circ \in \{\upperbound, \lowerbound\}$ and $(s, a)$ is a state-action pair in $\MDP$:
\begin{itemize}
	\item
	$\algostep$: The number of steps the algorithm took so far, increased by 1 after each iteration of the main loop, as already mentioned in the preliminaries.

	\item
	$s_\algostep, a_\algostep, s'_\algostep$:
	The state, action, and the sampled successor state in step $\algostep$, respectively.
	\item
	$\upperbound_\algostep(s, a)$ and $\lowerbound_\algostep(s, a)$:
	The (estimated) upper and lower bounds for the state-action pair $(s, a)$ at step $\algostep$.
	Note that in contrast to the previous algorithm, the upper and lower bounds are updated at each step instead of each episode.
	\item
	$\learn^\circ_\algostep(s, a)$:
	A three-valued flag ($\learnyes$, $\learnonce$, or $\learnno$) indicating whether the algorithm currently tries to learn and update the $\circ$-bounds for $(s, a)$.
	The meaning of $\learnonce$ is explained in the following.
	We additionally use the $\learndecrease$ function for convenience, which is defined by $\learnyes \mapsto \learnonce$, $\learnonce \mapsto \learnno$, and $\learnno \mapsto \learnno$.
	\item
	$\visitcount^\circ_\algostep(s, a)$:
	The number of times a value for $(s, a)$ was experienced.
	When $\visitcount^\circ_\algostep(s, a)$ is large enough, we can attempt an update with sufficient confidence.
	\item
	$\acc^\circ_\algostep(s, a)$:
	The accumulated sampled values of the last $\visitcount_\algostep^\circ(s, a)$ visits to $(s, a)$.
	We want $\acc^\circ_\algostep(s, a) / \visitcount_\algostep^\circ(s, a)$ to approximate the true $\circ$-bound.
\end{itemize}
Moreover, the algorithm contains the two constants $\updatestep$ and $\delay$.
We define their value (and the value of another constant, used for readability) as follows.
\begin{figure}[t]
  \centering
     \includegraphics[scale=1.4]{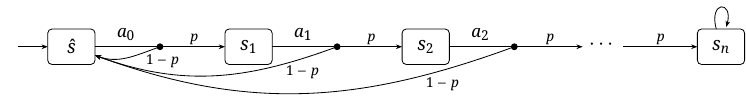}


	\caption{Example MDP to explain the choices and interpretations of some constants.} \label{fig:dql_constants_example}
\end{figure}
\begin{equation*}
	\updatestep = \frac{\varepsilon}{2} \cdot \frac{p_{\min}^{\cardinality{\States}}}{3 \cardinality{\States}} \qquad
	\updatecount = 2 \cardinality{\Actions} \left(1 + \frac{\cardinality{\Actions}}{\updatestep} \right) \qquad
	\delay = \left\lceil \frac{1}{2\updatestep^2} \ln \left(\frac{8}{\delta} \updatecount \right) \right\rceil
\end{equation*}
We call $\updatestep$ the \emph{update step} (the smallest update increment considered significant by the algorithm), $\updatecount$ the \emph{update count} (the maximal possible number of update attempts, mainly introduced for readability), and $\delay$ the \emph{update delay} (the number of samples we want to obtain for a state-action pair before we attempt an update).
These three constants are used throughout this and the following section.
Note that bounds on these constants can be obtained from \cref{def:limited_information} (recalling that $\cardinality{\Actions}$ is an upper bound on $\cardinality{\States}$).
Within the proofs, an even smaller value for $\updatestep$ or an even larger value for $\delay$ are also sufficient.
We define the constants with \enquote{tight} values to aid readability and intuition.

These constants are closely related to the worst-case \emph{mixing rate} (see e.g.\ \cite[Chapter~5]{levin2017markov} for a detailed discussion) of the MDP, which intuitively indicates how fast information \enquote{propagates} through the system.
For Markov chains, this is given by the difference between first and second eigenvalue of the transition matrix, which is also called \emph{spectral gap}.
This gap can be (quite conservatively) bounded by $p_{\min}^{\cardinality{\States}}$.
This also gives a bound on the convergence rate of the \emph{power iteration}, which in the context of Markov chains and MDP is closely related to value iteration.
(See, for example, \cite[Theorem~8.5.2]{DBLP:books/wi/Puterman94}, noting that $p_{\min}^{\cardinality{\States}}$ is a lower bound for $\eta$ with $J = \cardinality{\States}$.)

The concept of information propagation (and the tightness of the $p_{\min}^{\cardinality{\States}}$ bound) is illustrated in \cref{fig:dql_constants_example}.
In order to propagate any information about state $s_n$ to the initial state $\initialstate$, we need $\cardinality{\States}$ steps.
Moreover, after this many steps only a fraction $p_{\min}^{\cardinality{\States}}$ of the information is propagated, so, intuitively, to \enquote{observe} a difference of $\varepsilon$, we need to perform $\approx \cardinality{\States}p_{\min}^{-\cardinality{\States}} / \varepsilon$ steps.
Thus, we need to visit a state-action pair often enough, i.e.\ $\delay$ times, before an update to ensure that relevant information has propagated already with high confidence.
Dually, if a state-action pair was visited often enough and new information does not differ from the previous information by more than $\updatestep$, there likely is no new information to be propagated and we may assume that the values of this state-action pair have converged.

Inside the main loop, the algorithm repeats two steps to obtain a path.
First, an action maximizing the upper bounds (at the beginning of the episode) is randomly picked.
More precisely, we again consider the set $\umaxactions_\algoepisode(s) \coloneqq \argmax_{a \in \stateactions(s)} \upperbound_{\algostep_\algoepisode}(s, a)$ and uniformly select an action thereof.
To obtain the successor, we query the successor oracle with the given action to obtain the successor $s'$.
In other words, in episode $\algoepisode$ the algorithm samples a path in the MDP using a memoryless strategy randomizing uniformly over $\umaxactions_\algoepisode(s)$ in each state.
We call this strategy the \emph{sampling strategy} $\strategy_\algoepisode(s, a) \coloneqq \cardinality{\umaxactions_\algoepisode(s)}^{-1}$ if $a \in \umaxactions_\algoepisode(s)$ and $0$ otherwise.
We will later on introduce the upper bound maximizing strategy $\strategy_\algostep$, which selects among $\upperbound$-optimal actions at the current step $\algostep$.
Note that if the algorithm follows this strategy $\strategy_\algostep$ while sampling, the samples would not be obtained from a memoryless strategy in general, since an update might happen while sampling and thus change the strategy.
One might be tempted to solve this issue by first sampling a path until $\targetstate$ or $\sinkstate$ is reached and then propagating the values.
However this path might be of exponential size w.r.t.\ the number of states; this already occurs for the structurally simple example in \cref{fig:dql_constants_example}.

After sampling a tuple $(s, a, s')$, the algorithm learns from this \enquote{experience}.
It does so by learning upper and lower bounds separately, depending on the respective $\learn$ flags, which are explained later.
In case one of the bounds should be learned ($\learn^\circ_\algostep(s, a) \neq \learnno$), the accumulator is updated with the newly observed values, i.e.\ the respective bound of the successor $s'$.
Furthermore, if the algorithm has gathered enough information, i.e.\ this pair has been experienced $\delay$ times, an update of $(s, a)$'s estimate is attempted (if the respective $\learn$ is $\learnyes$ or $\learnonce$).
By choosing $\delay$ large enough, the information we gathered about the bounds of $(s, a)$ very likely is a faithful approximation of the true expected value over its successors.
If the newly learned estimate, i.e.\ the average over the last $\delay$ experiences of $(s, a)$, significantly differs from the current estimate stored in $\upperbound$ or $\lowerbound$, the current estimates are updated conservatively.
If instead this new estimate is close to the current estimate, the algorithm marks this state-action pair as (potentially) converged by \enquote{decreasing} its $\learn$ flag, as specified by $\learndecrease$.

The learned bounds of a pair depend on the bounds of other state-action pairs.
In particular, whenever any bound is changed, we may need to re-learn the values for all other state-action pairs.
This is taken care of by globally resetting the learn flags to $\learnyes$ in \cref{alg:dql_no_ec:line:reset_learn_upper,alg:dql_no_ec:line:reset_learn_lower}.
We highlight that this is one of the main differences to \cite{DBLP:conf/cav/AshokKW19}, where samples are instead used to learn bounds on the transition probabilities while the actual values are propagated according to these estimates, trading memory for speed of convergence.

The need for the intermediate value $\learnonce$ of $\learn$ arises from the asynchronicity of the updates.
Suppose an update of some pair $(s, a)$ succeeds and we reset all learn values to $\learnyes$.
However, for some other state-action pair $(s', a')$ we are very close to an update, too.
Then, the values which will be used for an attempted update of $(s', a')$ were mostly learned before the update of $(s, a)$.
Now, if for example $s$ is a successor of $(s', a')$, the values of $(s', a')$ may be influenced significantly by the update of $(s, a)$.
Hence, we need to learn the value of $(s', a')$ $\learnonce$ more in order to be on the safe side.
A different solution approach would be to simply reset all $\visitcount$ and $\acc$ values after every successful update, however this would be much less efficient:
If we again consider the above example, it might be the case that the values we gathered for $(s', a')$ before the update of $(s, a)$ already are sufficient for a successful update, discarding them would slow down convergence drastically.

In the algorithms of \cite{DBLP:conf/icml/StrehlLWLL06,DBLP:conf/atva/BrazdilCCFKKPU14}, this problem instead is taken care of by remembering the last globally successful update.
There, $\learn(s, a)$ is only set to $\learnno$ if the previous attempted update of $(s, a)$ happened after the last successful update.
This similarly implies that all values which are considered in the current update attempt are \enquote{up to date}.
We decided for this alternative approach since we have to track less variables.

\subsection{Proof of Correctness}

We now prove that \cref{alg:dql_no_ec} is probably approximately correct.
We first prove correctness of the result by showing that the computed bounds are faithful upper and lower bounds in \cref{stm:dql_no_ec:bounds_ordered}.
However, we cannot guarantee that this is always the case due to statistical outliers.
Thus we first obtain bounds on the probability of these outliers.
Then, in order to prove termination with high probability, we argue that by our choice of constants the propagation of values is probably correct.
This means that whenever we update the bounds of a state-action pair $(s, a)$, the updated value is close to the true average under $\mdptransitions(s, a)$.
Finally, we show that with high probability an update will occur as long as the bounds are not $\varepsilon$-close.

\begin{lemma} \label{stm:dql_no_ec:successful_update_count}
	The number of successful updates of $\upperbound$ and $\lowerbound$ is bounded by $\frac{\cardinality{\Actions}}{\updatestep}$ each.
\end{lemma}

\begin{proof}
	Let $a \in \Actions$ be some action and $s = \actionstate<\MDP>(a)$ the associated state.
	The upper bound of $(s, a)$ is initialised to $1$ or $0$, similar for the lower bound.
	Whenever $\upperbound_\algostep(s, a)$ is updated in \cref{alg:dql_no_ec:line:update_upperbound_value}, its value is decreased by at least $\updatestep$:
	We have that $\acc_\algostep^\upperbound(s, a) / m < \upperbound_\algostep(s, a) - 2 \updatestep$, hence $\acc_\algostep^\upperbound(s, a) / m + \updatestep < \upperbound_\algostep(s, a) - \updatestep$.
	Thus, $\upperbound_{\algostep+1}(s, a) < \upperbound_\algostep(s, a) - \updatestep$.
	Analogously, $\lowerbound_\algostep(s, a)$ is always increased by at least $\updatestep$ whenever updated.

	Moreover, $\acc_\algostep^\upperbound(s, a) \geq 0$ and $\acc_\algostep^\lowerbound(s, a) \leq \delay$ by initialization and update of these values, hence we never set $\upperbound_\algostep(s, a)$ to a negative value and $\lowerbound_\algostep(s, a)$ is always smaller or equal to $1$.
	Consequently, we change the value of $\upperbound_\algostep(s, a)$ and $\lowerbound_\algostep(s, a)$ at most $\frac{1}{\updatestep}$ times and there are at most $\frac{\cardinality{\Actions}}{\updatestep}$ successful updates to the upper and lower bounds, respectively.
	Note that we do not necessarily have $\upperbound_\algostep(s, a) \leq \lowerbound_\algostep(s, a)$ for all executions of the algorithm, hence there are at most $\frac{\cardinality{\Actions}}{\updatestep}$ updates for each of the bounds individually. 
\end{proof}
Observe that this implies that for every execution, eventually there will be no more successful updates of $\upperbound$ and the sampling strategy $\strategy_\algoepisode$ does not change.
This fact will be used in some of the subsequent proofs.
Moreover, we can use the above result to show that similarly, the number of \emph{attempted} updates is bounded.
\begin{lemma} \label{stm:dql_no_ec:attempted_update_count}
	The number of attempted updates of the upper bounds $\upperbound$ and lower bounds $\lowerbound$ is bounded by $\updatecount = 2 \cardinality{\Actions} (1 + \frac{\cardinality{\Actions}}{\updatestep})$, respectively.
\end{lemma}

\begin{proof}
	Let $(s, a) \in \States \times \stateactions$ be a state-action pair.
	Suppose an update of $\upperbound_\algostep(s, a)$ is attempted at step $\algostep$, i.e.\ $a_\algostep = a$, $\visitcount_\algostep(s, a) = \delay - 1$, and $\learn_\algostep^\upperbound(s, a) \neq \learnno$.
	Then, either the update is successful or $\learn^\upperbound_{\algostep + 1}(s, a)$ is updated with $\learndecrease$.
	The learn flag is only set to $\learnyes$ again if some other upper bound is successfully updated.
	Analogous reasoning applies to updates of the lower bounds.

	By \cref{stm:dql_no_ec:successful_update_count}, there are at most $\frac{\cardinality{\Actions}}{\updatestep}$ successful updates to either bounds in total.
	If an update of a particular state-action pair is attempted, it either succeeds or fails.
	In the latter case, at most one more update of this state-action pair will be attempted until an other update succeeds.
	Hence, for a particular state-action pair $(s, a)$ we have in the worst case two attempted $\upperbound$-updates after every successful $\upperbound$-update (of \emph{any} pair).
	Together, there are at most $2 + 2 \frac{\cardinality{\Actions}}{\updatestep}$ (two more attempts can occur after the last successful update).
	Since there are $\cardinality{\Actions}$ state-action pairs in total, the statement follows. 
\end{proof}
\begin{assumption} \label{asm:dql_no_ec:sampled_value_close_to_real_value}
	Suppose an $\upperbound$-update of the state-action pair $(s, a)$ is attempted at step $\algostep$.
	Let $k_1 < k_2 < \ldots < k_\delay = \algostep$ be the steps of the $\delay$ most recent visits to $(s, a)$.
	Then $\frac{1}{\delay} \sum_{i=1}^\delay \val(s_{k_i}') \geq \val(s, a) - \updatestep$.
	Analogously, for an attempted $\lowerbound$-update, we have $\frac{1}{\delay} \sum_{i=1}^\delay \val(s_{k_i}') \leq \val(s, a) + \updatestep$. 
\end{assumption}

\begin{lemma} \label{stm:dql_no_ec:sampled_value_close_to_real_value_probability}
	The probability that \cref{asm:dql_no_ec:sampled_value_close_to_real_value} is violated during the execution of \cref{alg:dql_no_ec} is bounded by $\frac{\delta}{4}$. 
\end{lemma}

\begin{proof}
	We show that the claim for the upper bound is violated with probability at most $\frac{\delta}{8}$.
	The lower bound part follows analogously and the overall claim via union bound.

	Let $(s, a)$ and $k_i$ as in \cref{asm:dql_no_ec:sampled_value_close_to_real_value}, i.e.\ an $\upperbound$-update of $(s, a)$ is attempted at step $k_m = \algostep$.
	First, observe that due to the Markov property, the successor state under $(s, a)$ does not depend on the algorithm's execution.
	Hence, the states $s_{k_i}'$, i.e.\ the successor states after each visit of $(s, a)$, are distributed i.i.d.\ according to the underlying probability distribution $\mdptransitions(s, a)$.
	Define $Y_i = \val(s_{k_i}')$.
	Clearly, $Y_i$ are i.i.d., since the actual value of a state $\val(s)$ is independent of the algorithm's execution.
	Moreover, $\Expectation[Y_i] = \val(s, a)$, since $\val$ satisfies the fixed point conditions $\val(s, a) = \ExpectedSumMDP{\mdptransitions}{s}{a}{\val}$.
	Define the empirical average $\underline{Y} = \frac{1}{\delay} \sum_{i = 1}^\delay Y_i$.
	Observe that $\Expectation[\underline{Y}] = \frac{1}{\delay} \sum_{i = 1}^\delay \Expectation[Y_i] = \val(s, a)$.
	By the Hoeffding bound \cite{hoeffding1994probability} we have that
	\begin{equation*}
		\ProbabilityAlgo*[\Expectation*[\underline{Y}] - \underline{Y} > \updatestep] \leq e^{-2 \delay \updatestep^2} = \frac{\delta}{8} \cdot \updatecount^{-1}
	\end{equation*}
	By reordering, we obtain that $\ProbabilityAlgo[\val(s, a) - \updatestep > \frac{1}{m} \sum_{i=1}^m \val(s_i)] \leq \frac{\delta}{8} \cdot \updatecount^{-1}$  \plabel{proof:dql_no_ec:sampled_value_close_to_real_value_probability:bounded_probability}.
	To conclude the proof, we extend the above argument to all steps $k_1$ satisfying the preconditions of the assumption.
	By \cref{stm:dql_no_ec:attempted_update_count}, the number of attempted updates to $\upperbound$ and $\lowerbound$ is bounded by $\updatecount$, respectively \plabel{proof:dql_no_ec:sampled_value_close_to_real_value_probability:bounded_number}.
	Consequently, by employing the union bound, we see that
	\begin{align*}
		& \ProbabilityAlgo*[\text{\enquote{$\frac{1}{\delay} {\sum}_{i = 1}^\delay \val(s_{k_i}) < \val(s, a) - \updatestep$ for some $k_1$}}] \\
		& \qquad \leq \ProbabilityAlgo*[{\Union}_{k_1} \text{\enquote{$\frac{1}{\delay} {\sum}_{i = 1}^\delay \val(s_{k_i}) < \val(s, a) - \updatestep$ for $k_1$}}] \\
		& \qquad \overset{\mathclap{\ref{proof:dql_no_ec:sampled_value_close_to_real_value_probability:bounded_probability}}}{\leq}~ {\sum}_{k_1} \frac{\delta}{8} \cdot \updatecount^{-1} ~\overset{\mathclap{\ref{proof:dql_no_ec:sampled_value_close_to_real_value_probability:bounded_number}}}{\leq}~ \frac{\delta}{8}.\qedhere
	\end{align*}
\end{proof}
\begin{lemma} \label{stm:dql_no_ec:bounds_ordered}
	Assume that \cref{asm:dql_no_ec:sampled_value_close_to_real_value} holds.
	Then, during any execution of \cref{alg:dql_no_ec} we have for every step $\algostep$, all states $s \in \States_\algoepisode$ and action $a \in \stateactions_\algoepisode(s)$ that
	\begin{equation*}
		\lowerbound_\algostep(s, a) \leq \lowerbound_{\algostep + 1}(s, a) \leq \val(s, a) \leq \upperbound_{\algostep + 1}(s, a) \leq \upperbound_\algostep(s, a).
	\end{equation*}
\end{lemma}

\begin{proof}
	First, by definition of the algorithm we clearly have that $\upperbound$ can only decrease and $\lowerbound$ can only increase.
	It remains to show that $\lowerbound_\algostep(s, a) \leq \val(s, a) \leq \upperbound_\algostep(s, a)$.
	We proceed by induction on the step $\algostep$.
	For $\algostep = 0$, the statement clearly holds, since $\upperbound_1(s, a) = 1$ for all states except the special state $\sinkstate$, which by assumption cannot reach the target $\targetstate$.
	Analogously, the statement holds for $\lowerbound_1(s, a)$.
	Now, fix an arbitrary step $\algostep$.
	We have that $\upperbound_{\algostep'}(s, a) \geq \val(s, a)$ for all steps $\algostep' \leq \algostep$ \alabel{proof:dql_no_ec:bounds_ordered:ih}{IH}.
	Assume that $(s, a)$ is the state-action pair sampled at step $\algostep$.
	If no successful update takes place there is nothing to prove, since the values of $\upperbound$ and $\lowerbound$ do not change.
	Otherwise, \cref{asm:dql_no_ec:sampled_value_close_to_real_value} is applicable and we get
	\begin{equation*}
		\upperbound_{\algostep + 1}(s, a) = \frac{1}{\delay} {\sum}_{i=1}^\delay \upperbound_{k_i}(s_{k_i}) + \updatestep \overset{\ref{proof:dql_no_ec:bounds_ordered:ih}}{\geq} \frac{1}{\delay} {\sum}_{i=1}^\delay \val(s_{k_i}) + \updatestep \geq \val(s, a).
	\end{equation*}
	Analogously, we have $\lowerbound_{\algostep + 1}(s, a) \leq \val(s, a)$. 
\end{proof}
This gives us correctness of the returned result with high confidence upon termination.
It remains to show that the algorithm also terminates with high probability.

To this end, we introduce the upper bound maximizing strategy $\strategy_\algostep$ which selects in each state $s$ uniformly among all actions maximal with respect to the \emph{current} upper bounds, i.e.\ $\upperbound_\algostep(s, \cdot)$.
This allows us to reason about the current value at step $\algostep$.
Note that this strategy differs from the sampling strategy $\strategy_\algoepisode$, since $\strategy_\algostep$ might change during an episode.
However, once there are no updates to upper bounds, we have that $\strategy_\algoepisode = \strategy_\algostep$.
We use this fact in the final convergence proof.
Once the two strategies align, we can transfer properties proven with respect to $\strategy_\algostep$ to the actual sampling behaviour of the algorithm.

Using this strategy, we define the set of converged state-action pairs.
\begin{definition} \label{def:dql_no_ec:converged_bounds}
	For every step $\algostep$, define $\ConvergedUpperBounds<\algostep>, \ConvergedLowerBounds<\algostep> \subseteq \States \times \stateactions$ by
	\begin{align*}
		\ConvergedUpperBounds<\algostep> & \coloneqq \{(s, a) \mid \upperbound_\algostep(s, a) - \ExpectedSumMDP{\mdptransitions}{s}{a}{\ExpectedSumStrat{\strategy_\algostep}{\upperbound_\algostep}} \leq 3 \updatestep\} \text{ and}\\
		\ConvergedLowerBounds<\algostep> & \coloneqq \{(s, a) \mid \ExpectedSumMDP{\mdptransitions}{s}{a}{\ExpectedSumStrat{\strategy_\algostep}{\lowerbound_\algostep}} - \lowerbound_\algostep(s, a) \leq 3 \updatestep\},
	\end{align*}
	i.e.\ all state-action pairs whose $\upperbound$- or $\lowerbound$-value is close to the respective value of its successors under $\strategy_\algostep$.
	If $(s, a) \in \ConvergedUpperBounds<\algostep>$, we say that $(s, a)$ is \emph{$\upperbound$-converged at step $\algostep$}, analogously $(s, a) \in \ConvergedLowerBounds<\algostep>$ is called \emph{$\lowerbound$-converged at step $\algostep$}.
\end{definition}
The approach for the convergence proof is to show that (with high probability) (i)~if an update of some bound fails, the current bound is consistent with its successors, i.e.\ the respective pair is converged, and (ii)~we visit non-converged pairs only finitely often.
Finally, we combine these two facts non-trivially to prove convergence.
\begin{lemma} \label{stm:dql_no_ec:expected_sum_strategy} \label{stm:dql_no_ec:nonconverged_bounds_monotonic}
	We have for every step $\algostep$ and state $s$ that
	\begin{equation*}
		\ExpectedSumStrat{\strategy_\algostep}{\upperbound_\algostep}(s) = \upperbound_\algostep(s) \text{\quad and \quad} \ExpectedSumStrat{\strategy_\algostep}{\lowerbound_\algostep}(s) \leq \lowerbound_\algostep(s).
	\end{equation*}
	Moreover, if $(s, a) \notin \ConvergedUpperBounds<\algostep>$, then $(s, a) \notin \ConvergedUpperBounds<\algostep'>$ for all $\algostep' > \algostep$ until an $\upperbound$-update of $(s, a)$ succeeds.
	If no more updates of upper bounds take place, the analogous statement holds for lower bounds, too.
\end{lemma}
\begin{proof}
	Since the strategy $\strategy_\algostep$ maximizes the upper bound we have
	\begin{equation*}
		\ExpectedSumStrat{\strategy_\algostep}{\upperbound_\algostep}(s) = {\sum}_{a \in \stateactions(s)} \strategy_\algostep(s, a) \cdot \upperbound_\algostep(s, a) = {\max}_{a \in \stateactions(s)} \upperbound_\algostep(s, a) = \upperbound_\algostep(s).
	\end{equation*}
	We also trivially have that $\ExpectedSumStrat{\strategy_\algostep}{\lowerbound_\algostep}(s) \leq \lowerbound_\algostep(s)$, as $\lowerbound_\algostep(s)$ is the maximum over all actions.

	For the second claim, recall that $\upperbound$-values can only decrease.
	If $(s, a) \notin \ConvergedUpperBounds<\algostep>$, we have $\upperbound_\algostep(s, a) > 3 \updatestep + \ExpectedSumMDP{\mdptransitions}{s}{a}{\ExpectedSumStrat{\strategy_\algostep}{\upperbound_\algostep}} = 3 \updatestep + \ExpectedSumMDP{\mdptransitions}{s}{a}{\upperbound_\algostep}$.
	Since (i)~$\upperbound_\algostep(s, a) = \upperbound_{\algostep + 1}(s, a)$ unless a successful $\upperbound$-update of $(s, a)$ occurs and (ii)~$\upperbound_\algostep(s) \geq \upperbound_{\algostep + 1}(s)$ for all states $s$, we obtain the claim.
	The lower bound statement is proven analogously, noting that once upper bounds remain fixed the only way to change $\ConvergedLowerBounds<\algostep>$ is a successful update of some lower bound. 
\end{proof}
\begin{assumption} \label{asm:dql_no_ec:converged_successful_update}
	Suppose an update of the upper bound (lower bound) of the state-action pair $(s, a)$ is attempted at step $\algostep$.
	Let $k_1 < k_2 < \ldots < k_\delay = \algostep$ be the steps of the $\delay$ most recent visits to $(s, a)$.
	If $(s, a)$ is not $\upperbound$-converged ($\lowerbound$-converged) at step $k_1$, the update at step $\algostep$ is successful.
\end{assumption}
Intuitively, this assumption says that whenever the bound for a state-action pair is significantly different from its successors and we visit that pair often enough, we obtain a significantly better estimate.
We cannot guarantee this surely due to outliers, but we bound the probability of this assumption being violated using our choice of the delay $\delay$.
\begin{lemma} \label{stm:dql_no_ec:converged_successful_update_probability}
	The probability that \cref{asm:dql_no_ec:converged_successful_update} is violated during the execution of \cref{alg:dql_no_ec} is bounded by $\frac{\delta}{4}$.
\end{lemma}
\begin{proof}
	As in \cref{stm:dql_no_ec:sampled_value_close_to_real_value_probability}, we prove that an attempted update of the upper bounds fails with probability at most $\frac{\delta}{8}$.
	The same bound then can be obtained for the lower bound variant with a mostly analogous proof.
	The overall result again follows using the union bound.

	Let $(s, a)$ and $k_i$ as in \cref{asm:dql_no_ec:converged_successful_update}, i.e.\ $(s, a) \notin \ConvergedUpperBounds<k_1>$ and an update of the upper bound is attempted at step $\algostep$ \plabel{proof:dql_no_ec:converged_successful_update_probability:attempt}.
	Define $X_i = \ExpectedSumStrat{\strategy_{k_1}}{\upperbound_{k_1}}(s_{k_i}')$.
	Note that all $X_i$ are defined using $\upperbound_{k_1}$ and $\strategy_{k_1}$ (instead of $\upperbound_{k_i}$ and $\strategy_{k_i}$).
	Consequently, the $X_i$ are i.i.d.\ and we can apply the Hoeffding bound to the empirical average $\underline{X} = \frac{1}{\delay} \sum_{i = 1}^\delay X_i$.
	This yields that
	\begin{equation*}
		\ProbabilityAlgo[\underline{X} - \Expectation[\underline{X}] \geq \updatestep] \leq e^{-2 \delay \updatestep^2} = \frac{\delta}{8} \cdot \updatecount^{-1}.
	\end{equation*}
	Since the $X_i$ are i.i.d., we have that $\Expectation[\underline{X}] = \Expectation[X_i]$ for all $1 \leq i \leq \delay$, in particular $\Expectation[\underline{X}] = \Expectation[X_1]$.
	Thus, the probability that $\underline{X} - \Expectation[X_1] \geq \updatestep$ is at most $\frac{\delta}{8} \cdot \updatecount^{-1}$ \plabel{proof:dql_no_ec:converged_successful_update_probability:difference_bound}.
	For the lower bound proof, we analogously define $X_i = \ExpectedSumStrat{\strategy_{k_1}}{\lowerbound_{k_1}}(s_{k_i}')$ and prove that $\Expectation[X_1] - \underline{X} \geq \updatestep$ with the same probability.

	Now, we show that if $\underline{X} - \Expectation[X_1] < \updatestep$ the update at step $\algostep$ will be successful \plabel{proof:dql_no_ec:converged_successful_update_probability:update_success}.
	Recall that an update is successful when the $\delay$ most recent samples significantly differ from the currently stored value, i.e.\ when the currently stored value $\upperbound_\algostep(s, a)$ is significantly larger than the newly learned value.
	We have that (reasoning below)
	\begin{align}
		& \upperbound_\algostep(s, a) - \frac{1}{m} {\sum}_{i=1}^\delay \upperbound_{k_i}(s_{k_i}') \geq \upperbound_\algostep(s, a) - \frac{1}{m} {\sum}_{i=1}^\delay \upperbound_{k_1}(s_{k_i}') \label{eq:stm:dql_no_ec:converged_successful_update_probability:first} \\
			& \qquad = \upperbound_\algostep(s, a) - \frac{1}{m} {\sum}_{i=1}^\delay \ExpectedSumStrat{\strategy_{k_1}}{\upperbound_{k_1}}(s_{k_i}') \label{eq:stm:dql_no_ec:converged_successful_update_probability:second} \\
			& \qquad > \upperbound_\algostep(s, a) - \Expectation[X_1] - \updatestep \label{eq:stm:dql_no_ec:converged_successful_update_probability:third} \\
			& \qquad = \upperbound_{k_1}(s, a) - \Expectation[X_1] - \updatestep \label{eq:stm:dql_no_ec:converged_successful_update_probability:fourth} \\
			& \qquad = \upperbound_{k_1}(s, a) - \ExpectedSumMDP{\mdptransitions}{s}{a}{\ExpectedSumStrat{\strategy_{k_1}}{\upperbound_{k_1}}} - \updatestep \label{eq:stm:dql_no_ec:converged_successful_update_probability:fifth} \\
			& \qquad > 2 \updatestep. \label{eq:stm:dql_no_ec:converged_successful_update_probability:sixth}
	\end{align}
	Inequality~\eqref{eq:stm:dql_no_ec:converged_successful_update_probability:first} follows from the fact that $\upperbound$-values can only decrease over time by definition of the algorithm.
	Equality~\eqref{eq:stm:dql_no_ec:converged_successful_update_probability:second} follows directly from \cref{stm:dql_no_ec:expected_sum_strategy}.
	Inequality~\eqref{eq:stm:dql_no_ec:converged_successful_update_probability:third} follows from the above derivation.
	Equality~\eqref{eq:stm:dql_no_ec:converged_successful_update_probability:fourth} follows from the fact that $\upperbound_{k_i}(s, a) = \upperbound_{k_1}(s, a)$ for all $1 \leq i \leq \delay$:
	Since an update is attempted at step $k_\delay = t$, there can be no update attempts in the previous $\delay - 1$ visits, consequently the value of $\upperbound_{k_i}(s, a)$ does not change between $k_1$ and $k_\delay$.
	Equality~\eqref{eq:stm:dql_no_ec:converged_successful_update_probability:fifth} follows directly from the definition of $X_1$.
	Finally, Inequality~\eqref{eq:stm:dql_no_ec:converged_successful_update_probability:sixth} follows from \ref{proof:dql_no_ec:converged_successful_update_probability:attempt}, i.e.\ that $(s, a)$ is not $\upperbound$-converged at step $k_1$, formally $\upperbound_{k_1}(s, a) - \ExpectedSumMDP{\mdptransitions}{s}{a}{\ExpectedSumStrat{\strategy_{k_1}}{\upperbound_{k_1}}} > 3 \updatestep$.

	For the lower bound, we prove a similar result:
	\begin{align*}
		& \frac{1}{m} {\sum}_{i=1}^\delay \lowerbound_{k_i}(s_{k_i}') - \lowerbound_\algostep(s, a) \geq \frac{1}{m} {\sum}_{i=1}^\delay \lowerbound_{k_1}(s_{k_i}') - \lowerbound_\algostep(s, a) \\
			& \qquad \geq \frac{1}{m} {\sum}_{i=1}^\delay \ExpectedSumStrat{\strategy_{k_1}}{\lowerbound_{k_1}}(s_{k_i}') - \lowerbound_\algostep(s, a) \\
			& \qquad > \Expectation[X_1] - \updatestep - \lowerbound_\algostep(s, a) \\
			& \qquad = \Expectation[X_1] - \updatestep - \lowerbound_{k_1}(s, a) \\
			& \qquad = \ExpectedSumMDP{\mdptransitions}{s}{a}{\ExpectedSumStrat{\strategy_{k_1}}{\lowerbound_{k_1}}} - \updatestep - \lowerbound_{k_1}(s, a) \\
			& \qquad > 2 \updatestep.
	\end{align*}
	The only major difference lies in the second inequality (corresponding to Equality~\eqref{eq:stm:dql_no_ec:converged_successful_update_probability:second}), where we instead use the fact that $\ExpectedSumStrat{\strategy_\algostep}{\lowerbound_\algostep}(s) \leq \lowerbound_\algostep(s)$.

	Finally, we again extend the argument to all steps $k_1$ similar to \cref{stm:dql_no_ec:sampled_value_close_to_real_value_probability}, i.e.\ that by \cref{stm:dql_no_ec:attempted_update_count} the number of attempted updates is bounded by $\updatecount$ \plabel{proof:dql_no_ec:converged_successful_update_probability:bounded_number}.
	Together with the union bound, we thus obtain
	\begin{align*}
		& \ProbabilityAlgo*[\text{\enquote{\cref{asm:dql_no_ec:converged_successful_update} is violated for $\upperbound$}}] \\
			& \qquad =~ \ProbabilityAlgo*[{\Union}_{k_1} \text{\enquote{$k_1$ satisfies condition \ref{proof:dql_no_ec:converged_successful_update_probability:attempt}, but the $\upperbound$-update fails}}] \\
			& \qquad \leq~ {\sum}_{k_1} \ProbabilityAlgo*[\text{\enquote{$k_1$ satisfies condition \ref{proof:dql_no_ec:converged_successful_update_probability:attempt}, but the $\upperbound$-update fails}}] \\
			& \qquad \overset{\mathclap{\ref{proof:dql_no_ec:converged_successful_update_probability:update_success}}}{\leq}~ {\sum}_{k_1} \ProbabilityAlgo*[\text{\enquote{$\underline{X} - \Expectation[X_1] \geq \updatestep$ for $k_1$}}] \\
			& \qquad \overset{\mathclap{\ref{proof:dql_no_ec:converged_successful_update_probability:difference_bound}}}{\leq}~ {\sum}_{k_1} \frac{\delta}{8} \cdot \updatecount^{-1} ~\overset{\mathclap{\ref{proof:dql_no_ec:converged_successful_update_probability:bounded_number}}}{\leq}~ \frac{\delta}{8}. \qedhere
	\end{align*}
\end{proof}

\begin{lemma} \label{stm:dql_no_ec:unsuccessful_update_means_converged}
	Assume that \cref{asm:dql_no_ec:converged_successful_update} holds.
	If an attempted $\upperbound$-update of $(s, a)$ at step $\algostep$ fails and $\learn^\upperbound_{\algostep+1}(s, a) = \learnno$, then $(s, a) \in \ConvergedUpperBounds<\algostep+1>$.
	If no more updates of upper bounds take place, the analogous statement holds for the lower bounds, too.
\end{lemma}

\begin{proof}
	We prove the statement for the upper bound, with the corresponding lower bound statement following analogously.
	Assume an unsuccessful $\upperbound$-update of $(s, a)$ occurs at step $\algostep$ and let $k_1 < k_2 < \ldots < k_\delay = \algostep$ be the $\delay$ most recent visits to $(s, a)$.
	We consider three cases:
	\begin{enumerate}
		\item
		If $(s, a) \notin \ConvergedUpperBounds<k_1>$, then by \cref{asm:dql_no_ec:converged_successful_update} the $\upperbound$-update of $(s, a)$ at step $\algostep$ will be successful and there is nothing to prove.

		\item
		We have $(s, a) \in \ConvergedUpperBounds<k_1>$ and there exists $i \in \{2, \ldots, \delay\}$ such that $(s, a)$ is not $\upperbound$-converged at step $k_i$.
		It follows that there must have been a successful update of some $\upperbound$-value between steps $k_1$ and $k_\delay$, say step $\algostep'$.
		By \cref{alg:dql_no_ec:line:reset_learn_upper}, $\learn_{\algostep' + 1}^\upperbound(s, a)$ is set to $\learnyes$ and there is nothing to prove.

		\item
		For the last case, we have that for all $i \in \{1, \ldots, \delay\}$ that $(s, a)$ is $\upperbound$-converged at step $k_i$, particularly $(s, a) \in \ConvergedUpperBounds<k_{\delay}> = \ConvergedUpperBounds<\algostep>$.
		As the attempt to update the $\upperbound$-value of $(s, a)$ at step $\algostep$ was unsuccessful, we have that $\ConvergedUpperBounds<\algostep> = \ConvergedUpperBounds<\algostep+1>$.
	\end{enumerate}
	For the lower bound statement, observe that $\ConvergedLowerBounds<\algostep>$ may be changed by a successful update of $\upperbound_\algostep$.
	Hence, the above reasoning can only be followed once upper bounds do not change. 
\end{proof}

\begin{lemma} \label{stm:dql_no_ec:non_converged_visit_count}
	Assume that \cref{asm:dql_no_ec:converged_successful_update} holds.
	Then, there are at most $2 \delay \cdot \frac{\cardinality{\Actions}}{\updatestep}$ visits to state-action pairs which are not $\upperbound$-converged.
	Moreover, once the upper bounds are not updated any more, there are at most $2 \delay \cdot \frac{\cardinality{\Actions}}{\updatestep}$ visits to state-action pairs which are not $\lowerbound$-converged.
\end{lemma}

\begin{proof}
	We show that whenever a state-action pair $(s, a)$ is not $\upperbound$-converged at step $\algostep$, then in at most $2 \delay$ more visits to $(s, a)$ a successful $\upperbound$-update will occur.
	Assume that $(s, a)$ is visited at step $\algostep$ and it is not $\upperbound$-converged, i.e.\ $(s, a) \notin \ConvergedUpperBounds<\algostep>$.
	We distinguish two cases.
	\begin{enumerate}
		\item $\learn^\upperbound_\algostep(s, a) = \learnno$:
		This implies that the last attempted $\upperbound$-update of $(s, a)$ was not successful.
		Let $\algostep'$ be the step of this attempt, $\algostep' < \algostep$.
		We have $\learn^\upperbound_{\algostep' + 1}(s, a) = \learnno$.
		By \cref{stm:dql_no_ec:unsuccessful_update_means_converged}, we have that $(s, a) \in \ConvergedUpperBounds<\algostep'+1>$.
		Since we assumed $(s, a) \notin \ConvergedUpperBounds<\algostep>$, there was a successful update of some $\upperbound$-value between $\algostep'$ and $\algostep$, otherwise we would have $\ConvergedUpperBounds<\algostep'+1> = \ConvergedUpperBounds<\algostep>$.
		Consequently, we have $\learn^\upperbound_{\algostep+1}(s, a) = \learnyes$.
		By \cref{asm:dql_no_ec:converged_successful_update} the next attempted $\upperbound$-update of $(s, a)$ will be successful.
		This attempt will occur after $\delay$ more visits to $(s, a)$.

		\item $\learn^\upperbound_\algostep(s, a) \neq \learnno$:
		By construction of the algorithm, we have that in at most $\delay - 1$ more visits to $(s, a)$, an $\upperbound$-update of $(s, a)$ will be attempted.
		Suppose this attempt takes place at step $\algostep'$, $\algostep' \geq \algostep$ and the most $\delay$ recent visits to $(s, a)$ prior to $\algostep'$ happened at steps $k_1 < k_2 < \ldots < k_\delay = \algostep'$.
		Note that we do not necessarily have that $\algostep = k_1$ or $\algostep = k_\delay$, but surely $\algostep \in \{k_1, \dots, k_\delay\}$.
		If the $\upperbound$-update at step $\algostep'$ succeeds, there is nothing to prove, hence assume that this update fails.
		There are two possibilities:
		\begin{enumerate}
			\item
			If $(s, a)$ is not $\upperbound$-converged at step $k_1$, then by \cref{asm:dql_no_ec:converged_successful_update} the $\upperbound$-update at step $\algostep'$ will be successful, contradicting the assumption.

			\item
			If instead $(s, a)$ is $\upperbound$-converged at step $k_1$, we have that $\ConvergedUpperBounds<k_1> \neq \ConvergedUpperBounds<\algostep>$, since we assumed that $(s, a) \notin \ConvergedUpperBounds<\algostep>$.
			Consequently, there was a successful $\upperbound$-update of some other state-action pair at some step $\algostep''$ with $k_1 < \algostep'' \leq \algostep$ and thus $\learn^\upperbound_{\algostep'' + 1}(s, a) = \learnyes$.
			Moreover, we necessarily have that no $\upperbound$-update of $(s, a)$ is attempted after $\algostep''$.
			Together, we have that $\learn^\upperbound_{\algostep' + 1}(s, a) = \learnonce$ even though the attempted $\upperbound$-update at step $\algostep'$ fails.
			By \cref{stm:dql_no_ec:nonconverged_bounds_monotonic}, we have that $(s, a) \notin \ConvergedUpperBounds<\algostep'+1>$, as $(s, a) \notin \ConvergedUpperBounds<\algostep>$ and no successful $\upperbound$-update of $(s, a)$ occurred between $\algostep$ and $\algostep'$.
			By \cref{asm:dql_no_ec:converged_successful_update} the next attempt to update $\upperbound$-value of $(s, a)$ will succeed.
		\end{enumerate}
	\end{enumerate}
	By \cref{stm:dql_no_ec:successful_update_count}, the number of successful $\upperbound$-updates is bounded by $\frac{\cardinality{\Actions}}{\updatestep}$, and by the previous arguments we have that if for some $\algostep$ the pair $(s, a)$ is not $\upperbound$-converged then in at most $2 \delay$ more visits to $(s, a)$, there will be a successful update to $\upperbound(s, a)$.
	Hence, there can be at most $2 \delay \cdot \frac{\cardinality{\Actions}}{\updatestep}$ steps $\algostep$ such that the current state-action pair is not $\upperbound$-converged.
	Once no more $\upperbound$-updates take place, $\strategy_\algostep$ remains fixed and $\ConvergedLowerBounds<\algostep>$ only changes due to successful updates of the lower bounds, yielding an analogous proof for $\lowerbound$. 
\end{proof}
As a last auxiliary lemma, we show that whenever the probability of reaching a non-converged pair is low, we necessarily are close to the optimal value.
\begin{lemma} \label{stm:dql_no_ec:bounds_close}
	Assume that \cref{asm:dql_no_ec:sampled_value_close_to_real_value} holds and fix a step $\algostep$.
	Then, for every state $s \in \States$ we have that
	\begin{gather*}
		\upperbound_\algostep(s) - 3 \updatestep \cdot \cardinality{\States} p_{\min}^{-\cardinality{\States}} - \ProbabilityMDP<\MDP, s><\strategy_\algostep>[\reach \overline{\ConvergedUpperBounds<\algostep>}] \leq \\
		\ProbabilityMDP<\MDP, s><\strategy_\algostep>[\reach \{\targetstate\}] \leq \\
		\lowerbound_\algostep(s) + 3 \updatestep \cdot \cardinality{\States} p_{\min}^{-\cardinality{\States}} + \ProbabilityMDP<\MDP, s><\strategy_\algostep>[\reach \overline{\ConvergedLowerBounds<\algostep>}].
	\end{gather*}
\end{lemma}

\begin{proof}
	The central idea of this proof is to apply \cref{stm:assignment_mdp_value} twice, with $X(s, a) = \upperbound_\algostep(s, a)$ and $X(s, a) = \lowerbound_\algostep(s, a)$, respectively.

	For the first application, set $\kappa_l = -1$, $\kappa_u = 3 \updatestep$, and $\strategy = \strategy_\algostep$.
	Then, $\mathcal{K} = \ConvergedUpperBounds<\algostep>$ and
	\begin{equation}
		\ProbabilityMDP<\MDP', s><\strategy_\algostep>[\reach \{\targetstate\}] - \ProbabilityMDP<\MDP, s><\strategy_\algostep>[\reach \overline{\ConvergedUpperBounds<\algostep>}] \leq \ProbabilityMDP<\MDP, s><\strategy_\algostep>[\reach \{\targetstate\}] \label{eq:stm:dql_no_ec:bounds_close:up_probs}
	\end{equation}
	since $\MDP'$ and $\MDP$ are equivalent on $\ConvergedUpperBounds<\algostep>$.
	The lemma then yields that
	\begin{equation}
		\ExpectedSumStrat{\strategy_\algostep}{\upperbound_\algostep}(s) - \ProbabilityMDP<\MDP'_t, s><\strategy_\algostep>[\reach \{\targetstate\}] \leq 3 \updatestep \cdot \cardinality{\States} p_{\min}^{-\cardinality{\States}}. \label{eq:stm:dql_no_ec:bounds_close:up_strat}
	\end{equation}
	Recall that $\strategy_t$ is a strategy randomizing uniformly over some of the available actions in each state, hence $\delta_{\min}(\strategy)$ is at least $p_{\min}$.
	For the second application, we dually set $\kappa_l = - 3 \updatestep$, $\kappa_u = 1$, and $\strategy = \strategy_\algostep$.
	Again, we have $\mathcal{K} = \ConvergedLowerBounds<\algostep>$ and
	\begin{equation}
		\ProbabilityMDP<\MDP, s><\strategy_\algostep>[\reach \{\targetstate\}] \leq \ProbabilityMDP<\MDP', s><\strategy_\algostep>[\reach \{\targetstate\}] + \ProbabilityMDP<\MDP, s><\strategy_\algostep>[\reach \overline{\ConvergedLowerBounds<\algostep>}]. \label{eq:stm:dql_no_ec:bounds_close:low_probs}
	\end{equation}
	The lemma gives us
	\begin{equation}
		\ProbabilityMDP<\MDP'_t, s><\strategy_\algostep>[\reach \{\targetstate\}] - \ExpectedSumStrat{\strategy_\algostep}{\lowerbound_\algostep}(s) \leq 3 \updatestep \cdot \cardinality{\States} p_{\min}^{-\cardinality{\States}}. \label{eq:stm:dql_no_ec:bounds_close:low_strat}
	\end{equation}

	Now, recall that $\ExpectedSumStrat{\strategy_\algostep}{\upperbound_\algostep}(s) = \upperbound_\algostep(s)$ and $\ExpectedSumStrat{\strategy_\algostep}{\lowerbound_\algostep}(s) \leq \lowerbound_\algostep(s)$ \plabel{proof:dql_no_ec:bounds_close:expected_strat} due to \cref{stm:dql_no_ec:expected_sum_strategy}.
	Together, we have
	\begin{gather*}
		\upperbound_\algostep(s) - 3 \updatestep \cdot \cardinality{\States} p_{\min}^{-\cardinality{\States}} \overset{\ref{proof:dql_no_ec:bounds_close:expected_strat}}{=} \ExpectedSumStrat{\strategy_\algostep}{\upperbound_\algostep}(s) - 3 \updatestep \cdot \cardinality{\States} p_{\min}^{-\cardinality{\States}} \overset{\eqref{eq:stm:dql_no_ec:bounds_close:up_strat}}{\leq} \ProbabilityMDP<\MDP'_t, s><\strategy_\algostep>[\reach \{\targetstate\}], \\
		\ProbabilityMDP<\MDP', s><\strategy_\algostep>[\reach \{\targetstate\}] - \ProbabilityMDP<\MDP, s><\strategy_\algostep>[\reach \overline{\ConvergedUpperBounds<\algostep>}] \overset{\eqref{eq:stm:dql_no_ec:bounds_close:up_probs}}{\leq} \ProbabilityMDP<\MDP, s><\strategy_\algostep>[\reach \{\targetstate\}] \overset{\eqref{eq:stm:dql_no_ec:bounds_close:low_probs}}{\leq} \ProbabilityMDP<\MDP', s><\strategy_\algostep>[\reach \{\targetstate\}] + \ProbabilityMDP<\MDP, s><\strategy_\algostep>[\reach \overline{\ConvergedLowerBounds<\algostep>}], \text{ and} \\
		\ProbabilityMDP<\MDP'_t, s><\strategy_\algostep>[\reach \{\targetstate\}] \overset{\eqref{eq:stm:dql_no_ec:bounds_close:low_strat}}{\leq} 3 \updatestep \cdot \cardinality{\States} p_{\min}^{-\cardinality{\States}} + \ExpectedSumStrat{\strategy_\algostep}{\lowerbound_\algostep}(s) \overset{\ref{proof:dql_no_ec:bounds_close:expected_strat}}{\leq} 3 \updatestep \cdot \cardinality{\States} p_{\min}^{-\cardinality{\States}} + \lowerbound_\algostep(s). \qedhere
	\end{gather*}
\end{proof}
Combining all the above statements now yields the overall result.
\begin{theorem} \label{stm:dql_no_ec_correct}
	\cref{alg:dql_no_ec} terminates and yields a correct result with probability at least $1 - \delta$ after at most $\mathcal{O}(\mathrm{POLY}(\cardinality{\Actions}, p_{\min}^{-\cardinality{\States}}, \varepsilon^{-1}, \ln \delta))$ steps.
\end{theorem}

\begin{proof}
	We only consider executions where \cref{asm:dql_no_ec:sampled_value_close_to_real_value,asm:dql_no_ec:converged_successful_update} hold.
	By \cref{stm:dql_no_ec:sampled_value_close_to_real_value_probability,stm:dql_no_ec:converged_successful_update_probability} together with the union bound, this happens with probability at least $1 - \frac{\delta}{2}$.

	Now, observe that if the algorithm terminates at some step $\algostep$, we have that $\upperbound_\algostep(\initialstate) - \lowerbound_\algostep(\initialstate) < \varepsilon$ by definition.
	With \cref{stm:dql_no_ec:bounds_ordered}, we get $\lowerbound_\algostep(\initialstate) \leq \val(\initialstate) \leq \upperbound_\algostep(\initialstate)$.
    Reordering yields the result.

	We show by contradiction that the algorithm terminates for almost all considered executions.
	Thus, assume that the execution does not halt with non-zero probability.
	Since the MDP $\MDP$ satisfies \cref{asm:mec_free}, almost all episodes eventually visit either $\targetstate$ or $\sinkstate$ due to \cref{stm:mdp_almost_sure_absorption} and thus are of finite length.
	Thus, almost all executions for which the algorithm does not terminate comprise infinitely many episodes.
	We restrict our attention to only those executions.

	Recall that due to \cref{stm:dql_no_ec:attempted_update_count}, there are only finitely many attempted updates on almost all considered executions.
	Consequently, on these executions the algorithm eventually does not change $\upperbound$, since no successful updates can occur from some step $\algostep$ onwards.
	This means that all following samples are obtained by sampling according to the strategy $\strategy_\algostep$.
	Note that both the time of convergence and the actual strategy $\strategy_\algostep$ depends on the execution $\algorithmexecution$.
	Thus, we need to employ \cref{stm:markov_process_repeating}---the algorithm clearly qualifies as Markov process, since its evolution only depends on its current valuations.
	More precisely, it is not difficult to see that the whole execution of the algorithm (with fixed inputs) can be modelled as a (very unwieldy) countable Markov chain, showing that the considered properties are measurable.
	In particular, they are reachability objectives on this induced Markov chain.

	Let us now consider the set of executions for which the upper bounds eventually converge and moreover $\ProbabilityMDP<\MDP, \initialstate><\strategy_\algostep>[\reach \setcomplement{\ConvergedUpperBounds<\algostep>}] \geq \rho > 0$ infinitely often.
	Assume that this set of executions has a non-zero measure.
	By \cref{stm:markov_process_repeating}, on almost all of these executions $\setcomplement{\ConvergedUpperBounds<\algostep>}$ is also reached infinitely often, contradicting \cref{stm:dql_no_ec:non_converged_visit_count}.
	For the lower bounds, we can prove a completely analogous statement.
	Consequently, $\ProbabilityMDP<\MDP, \initialstate><\strategy_{\algostep}>[\reach \setcomplement{\ConvergedUpperBounds<\algostep>}] \to 0$ and $\ProbabilityMDP<\MDP, \initialstate><\strategy_{\algostep}>[\reach \setcomplement{\ConvergedLowerBounds<\algostep>}] \to 0$ on almost all considered executions for $\algostep \to \infty$.

	Inserting the definition of $\updatestep$, we have for a sufficiently large step $\algostep$ that
	\begin{equation*}
		\upperbound_\algostep(\initialstate) - \frac{\varepsilon}{2} < \upperbound_\algostep(\initialstate) - 3 \updatestep \cdot \cardinality{\States} p_{\min}^{-\cardinality{\States}} - \ProbabilityMDP<\MDP, \initialstate><\strategy_\algostep>[\reach \overline{\ConvergedUpperBounds<\algostep>}]
	\end{equation*}
	and dually
	\begin{equation*}
		\lowerbound_\algostep(\initialstate) + 3 \updatestep \cdot \cardinality{\States} p_{\min}^{-\cardinality{\States}} + \ProbabilityMDP<\MDP, \initialstate><\strategy_\algostep>[\reach \overline{\ConvergedLowerBounds<\algostep>}] < \lowerbound_\algostep(\initialstate) + \frac{\varepsilon}{2}
	\end{equation*}
	for all considered executions.
	Thus, by \cref{stm:dql_no_ec:bounds_close}, we have
	\begin{equation*}
		\upperbound_\algostep(\initialstate) - \frac{\varepsilon}{2} < \ProbabilityMDP<\MDP, \initialstate><\strategy_\algostep>[\reach \{\targetstate\}] < \lowerbound_\algostep(\initialstate) + \frac{\varepsilon}{2},
	\end{equation*}
	i.e.\ $\upperbound_\algostep(\initialstate) - \lowerbound_\algostep(\initialstate) < \varepsilon$, contradicting the assumption.

	We have proven that the result is approximately correct with probability $1 - \frac{\delta}{2}$.
	Now, we additionally need to prove the step bound.
	To this end, we first bound the number of sampled paths and then bound the length of each path.
	Central to the following proof is \cref{stm:dql_no_ec:non_converged_visit_count}, bounding the number of visits to non-converged state-action pairs.
	First, we treat the upper bounds.
	Observe that the probability of visiting a non-$\upperbound$-converged state-action pair either is $0$ or at least $p_{\min}^{\cardinality{\States}}$ (due to \cref{stm:markov_chain_minimum_reachability}).
	Moreover, while this probability may fluctuate, once it reaches $0$ it remains at $0$, since then the sampling strategy does not change and all pairs reachable under this strategy are $\upperbound$-converged.
	So, in the worst case, the probability of reaching such a pair is exactly $p_{\min}^{\cardinality{\States}}$ until they are visited often enough.
	We model this process as a series of Bernoulli trials $X_i$, equalling $1$ if at least one $\upperbound$-update happens while sampling the $i$-th path.\footnote{We deliberately use $i$ instead of $\algoepisode$ to emphasize that $X_i$ does not operate on the probability space of the algorithm $(\algorithmspace, \algorithmsigma, \ProbabilityAlgo)$.
	Instead, they represent a crude under-approximation to allow for a feasible analysis.}
	While the exact probabilities are not independent, they are always at least as large as the success probability $p \coloneqq p_{\min}^{\cardinality{\States}}$ of these trials (or $0$ if all reachable pairs are $\upperbound$-converged).
	Hence, we approximate the number of trials we need to perform until we observe at least $c \coloneqq 2 \delay \cdot \frac{\cardinality{\Actions}}{\updatestep}$ successes with high probability---then, all upper bounds necessarily are converged by \cref{stm:dql_no_ec:non_converged_visit_count}.
	Now, we are essentially dealing with a binomially distributed variable $X_n = \sum_{i=1}^n X_i$ and want to find an $n$ such that $\Probability[X_n \geq c] \geq 1 - \frac{\delta}{4}$.
	Since we are interested in the limit behaviour, we can apply the de~Moivre--Laplace theorem, allowing us to replace this binomial distribution with an appropriate normal distribution.
	Thus, we obtain
	\begin{equation*}
		\Probability[X_n \geq c] \approx 1 - \Phi \left( \frac{c - np}{\sqrt{np(1-p)}} \right),
	\end{equation*}
	and rearranging yields
	\begin{equation*}
		n^{-\frac{1}{2}} (c - np) \approx \Phi^{-1}\left( \frac{\delta}{4} \right) \cdot \sqrt{p (1 - p)}.
	\end{equation*}
	For readability, we set $a \coloneqq \Phi^{-1}\left( \frac{\delta}{4} \right)$.
	Solving for $n$ gives us
	\begin{equation*}
		n \approx \frac{c}{p} - \frac{a}{2p} \sqrt{(1-p)^2 a^2 + 4 c (1 - p)} + \frac{(1 - p) r^2}{2p}.
	\end{equation*}
	Inserting the definitions yields that $n \in \mathcal{O}(\mathrm{POLY}(\cardinality{\Actions}, p_{\min}^{-\cardinality{\States}}, \varepsilon^{-1}, \ln \delta))$.
	This bounds the number of paths sampled by the algorithm.
	We furthermore prove that the length of all those paths is polynomial with high probability.
	To this end, we employ \cref{stm:markov_chain_finite_reach_bound}.
	Recall that sampling of a path stops once we reach one of the two special states $\targetstate$ and $\sinkstate$.
	Due to \cref{asm:mec_free}, the probability of eventually reaching them is $1$.
	Hence, $\ProbabilityMDP<\MDP, \initialstate><\strategy_\algoepisode>[\boundedreach<N> \{\targetstate, \sinkstate\}] \geq 1 - \tau$, where $N \geq \ln(\frac{2}{\tau}) \cdot \cardinality{\States} p_{\min}^{-\cardinality{\States}}$ for any sampling strategy $\strategy_\algoepisode$.
	In other words, the probability of a sampled path being longer than $N$ is at most $\tau$.
	Then, by the union bound, the probability of any of the $n$ paths being longer than $N$ is at most $n \cdot \tau$.
	By choosing $\tau = \frac{\delta}{4 n}$, this happens with probability at most $\frac{\delta}{4}$.
	Then, $\ln(\frac{2}{\tau}) = \ln(8 n) - \ln(\delta)$, i.e.\ the length of each path again is bounded by a polynomial in the input values.
	Together, we obtain the results, since polynomials are closed under multiplication. 
\end{proof}

\begin{remark}[Relation to \cite{DBLP:conf/icml/StrehlLWLL06}]
	Before we proceed to the general case, we briefly discuss how our proof structure relates to the one of \cite{DBLP:conf/icml/StrehlLWLL06} and how it can be used to derive a variant of their Theorem~1.
	Most of our proofs are quite analogous.
	For example, \cref{asm:dql_no_ec:converged_successful_update} is practically equivalent to their Assumption~A1, similarly \cref{stm:dql_no_ec:converged_successful_update_probability,stm:dql_no_ec:bounds_ordered,stm:dql_no_ec:unsuccessful_update_means_converged,stm:dql_no_ec:non_converged_visit_count} and the respective proofs correspond to their Lemma~1 to 4 (however, note the different bounds).
	Since we are dealing with unbounded reachability (assuming almost sure absorption by \cref{asm:mec_free}), the purpose of Lemma~5 corresponds to our \cref{stm:assignment_mdp_value}.

	Major differences arise in the actual proof of \cite[Theorem~1]{DBLP:conf/icml/StrehlLWLL06}.
	As we already pointed out, the Hoeffding bound is not applicable to variables indicating whether an update has occurred in a particular step due to the clear dependency.
	The related proof step aims to show that with high probability after a certain number of steps, the number of possible updates is exhausted (by virtue of \cref{stm:dql_no_ec:non_converged_visit_count}) and then bounds the deviation from the true value based on this.
	We prove a similar statement via \cref{stm:dql_no_ec:bounds_close}, connecting the probability of visiting a non-converged state-action pair to the convergence of the bounds.
	Note that the proof of \cref{stm:dql_no_ec:bounds_close} employs the auxiliary \cref{stm:assignment_mdp_value}.
\end{remark}
 	\section{Limited Information -- General Case} \label{sec:dql}

As before, MECs pose an additional challenge, since they introduce superfluous upper fixed points.
The key difference to the full information setting is that MECs cannot be directly identified.
Instead, we identify a set of state-action pairs as an end component if it occurs sufficiently often.
By bounding the probability of falsely identifying such a set as an end component, we can replicate the previous proof structure.
\subsection{Collapsing End Components with Limited Information}
Before we present the complete algorithm, we first show how we identify end components in this section.

\newcommand{\appear}{Appear}

\begin{definition}
	Let $\MDP = (\States, \Actions, \stateactions, \mdptransitions)$ be an MDP, $\infinitepath \in \Infinitepaths<\MDP>$ and $i, j \geq 0$.
	We define the state-action pairs which appear at least $i$ times on the path $\infinitepath$ during the first $j$ steps as
	\begin{equation*}
		\appear(\infinitepath, i, j) = \{(s, a) \in \States \times \stateactions \mid \cardinality{\{k \mid k \leq j \land \infinitepath^a(k) = a\}} \geq i\}.
	\end{equation*}
	We overload the definition of $\appear$ to also accept finite paths of sufficient length.
	Moreover, we also define $\appear$ for paths of Markov chains, which yields the states occurring more than $i$ times.
\end{definition}
For notational convenience, we identify the result of $\appear$ with the corresponding state-action tuple $(R, B)$ since we will use these results as candidates for end components.
With appropriate $i$ and $j$, $\appear$ is an EC with high probability.
\begin{lemma} \label{stm:dql:appear_is_ec_probability}
	Let $\MDP = (\States, \Actions, \stateactions, \mdptransitions)$ be an MDP, $\initialstate \in \States$ an initial state, $\targetset \subseteq \States$ a set of target states, and $\strategy \in \StrategiesMD<\MDP>$ a memoryless strategy on $\MDP$ such that $\ProbabilityMDP<\MDP, s><\strategy>[\reach \overline{\targetset}] = 0$ for all $s \in \targetset$, i.e.\ $\targetset$ is absorbing under $\strategy$.
	Set $\States_\strategy = \Union_{s \in \States} \support(\strategy(s))$, $\kappa = \cardinality{\States_\strategy} + 1$, and pick $i \geq \kappa$.
	Then either $\ProbabilityMDP<\MDP, \initialstate><\strategy>[\boundedreach<2i^3> T] = 1$ or
	\begin{equation*}
		\ProbabilityMDP*<\MDP, \initialstate><\strategy>[{App}_i \mid \overline{\boundedreach<2i^3> T}] \geq 1 - 2 (1 + i^2) \cdot e^{- (i - 1) \frac{\delta_{\min}(\strategy)^\kappa}{\kappa}} \cdot \delta_{\min}(\strategy)^{-\kappa},
	\end{equation*}
	where ${App}_i = \{\infinitepath \in \Infinitepaths<\MDP> \mid \appear(\infinitepath, i, 2i^3) \in \Ecs(\MDP)\}$.
\end{lemma}
Informally, this lemma shows that, when sampling according to a memoryless strategy, paths of sufficient length either end up in an already known set of ECs or frequently reappearing state-action pairs also form an EC with high probability.
\begin{proof}
	If $\ProbabilityMDP<\MDP, \initialstate><\strategy>[\boundedreach<2i^3> T] = 1$, there is nothing to prove, hence we assume the opposite, i.e.\ that $\ProbabilityMDP<\MDP, \initialstate><\strategy>[\overline{\boundedreach<2i^3> T}] > 0$ \plabel{proof:dql:appear_is_ec_probability:assumption}.

	Given an MDP, a designated initial state $\initialstate$, and a memoryless strategy, we can construct a finite state Markov chain which exactly captures the behaviour of the MDP under the given strategy.
	We define the Markov chain $\MC_\strategy = (\{\initialstate\} \union \States_\strategy, \mctransitions_\strategy)$, where $\mctransitions_\strategy$ is defined as
	\begin{align*}
		\mctransitions(\initialstate, a) & = \strategy(\initialstate, a) \text{ for $a \in \support(\strategy(\initialstate))$} \\
		\mctransitions(a, a') & = \mdptransitions(\actionstate<\MDP>(a), a, \actionstate<\MDP>(a')) \cdot \strategy(\actionstate<\MDP>(a'), a').
	\end{align*}
	In other words, $\mctransitions(a, a')$ equals the probability of reaching some state $s'$ after playing action $a$ and then continuing with action $a'$.
	Recall that each action is tied to a unique state.
	As such, the paths in $\MC_\strategy$ exactly correspond to the paths in $\MDP$ following $\strategy$.
	Furthermore, it is easy to see that each BSCC of $\MC_\strategy$ corresponds to an end component in $\MDP$.
	Observe that, by definition, $\kappa$ equals the number of states in $\MC_\strategy$ \plabel{proof:dql:appear_is_ec_probability:kappa_is_states} and $\delta_{\min}(\strategy)$ equals the smallest positive transition probability in $\MC_\strategy$ \plabel{proof:dql:appear_is_ec_probability:delta_min}.
	For readability, we define $c = \exp(- \delta_{\min}(\strategy)^\kappa / \kappa)$.

	Let ${App}_{i, \strategy} \subseteq \Infinitepaths<\MC_\strategy>$ be the event corresponding to ${App}_i$ in the Markov chain $\MC_\strategy$.
	Informally, ${App}_{i, \strategy}$ denotes the set of all (infinite) paths $\infinitepath$ which within $2 i^3$ steps (i)~visit all states of some BSCC at least $i$ times, and (ii)~all other states at most $i - 1$ times, i.e.\ all paths such that $\appear(\infinitepath, i, 2i^3)$ is a BSCC of $\MC_\strategy$.
	We now show that
	\begin{equation*}
		\ProbabilityMC<\MC_\strategy, \initialstate>[{App}_{i, \strategy} \mid \overline{\boundedreach<2 i^3> \targetset}] \geq 1 - 2 c^i i^3 \cdot \delta_{\min}(\strategy)^{-\kappa},
	\end{equation*}
	i.e.\ the probability of ${App}_{i, \strategy}$ given that $\targetset$ is not reached within $2 i^3$ steps is at least $1 - 2 c^i i^3 \cdot \delta_{\min}(\strategy)^{-\kappa}$.
	Since the paths of $\MC_\strategy$ exactly correspond to paths obtained in $\MDP$ by following the strategy $\strategy$, this proves the claim.

	First, we show that \plabel{proof:dql:appear_is_ec_probability:absolute_bound}
	\begin{equation*}
		\ProbabilityMC*<\MC_\strategy, \initialstate>[{App}_{i, \strategy}] \geq 1 - 2 (1 + i^2) \cdot c^{i - 1}.
	\end{equation*}
	Let $B = \Union_{R \in \Bsccs(\MC_\strategy)} R$ be the set of all states in BSCCs of $\MC_\strategy$.
	We have that $\ProbabilityMC<\MC_\strategy, \initialstate>[\reach B] = 1$ by \cref{stm:mc_almost_sure_absorption}.
	We apply \cref{stm:markov_chain_finite_reach_bound} with $N = i - 1$ and $\tau = 2 c^{i - 1}$.
	By \ref{proof:dql:appear_is_ec_probability:kappa_is_states} and \ref{proof:dql:appear_is_ec_probability:delta_min}
	\begin{equation*}
		\cardinality{\States_\strategy} \cdot \ln\left(\frac{2}{\tau} \right) \cdot \delta_{\min}(\strategy)^{-\cardinality{\States_\strategy}} = \kappa \cdot \ln\left(\exp\left( (i-1) \cdot \frac{\delta_{\min}(\strategy)^\kappa}{\kappa}\right)\right) \cdot \delta_{\min}(\strategy)^{-\kappa} = i - 1.
	\end{equation*}
	Thus $\ProbabilityMC<\MC_\strategy, \initialstate>[\boundedreach<i - 1> B] \geq 1 - 2 c^{i - 1}$.
	In other words, an infinite path of $\MC_\strategy$ starting in $\initialstate$ does not visit a BSCC of $\MC_\strategy$ within $i - 1$ steps with probability at most $2 c^{i - 1}$.

	Now, let $R = \{s_1, \dots, s_n\} \subseteq B$ be some BSCC of $\MC_\strategy$ and fix two states $s_i, s_j \in R$.
	Since $R$ is an BSCC, we have $\ProbabilityMC<\MC_\strategy, s_i>[\reach \{s_j\}] = 1$, and we can apply \cref{stm:markov_chain_finite_reach_bound} again to obtain that $\ProbabilityMC<\MC_\strategy, s_i>[\boundedreach<i> \{s_j\}] \geq 1 - 2 c^{i - 1}$.
	Consequently, the probability of visiting all states of $R$, one after another, with at most $i - 1$ steps between visiting the respective next state, is at least $1 - n \cdot 2 c^{i - 1}$.
	Repeating this argument, with probability at least $1 - i \cdot n \cdot 2 c^{i - 1} \geq 1 - i \cdot \kappa \cdot 2 c^{i - 1}$, this round trip is successful $i$ times in a row and has a length of at most $i \cdot n \cdot (i - 1) \leq i^2 \kappa \leq i^3$.
	Using the union bound again, we get that with probability at least $1 - 2 c^{i - 1} - i \kappa \cdot 2 c^{i - 1} = 1 - 2 c^{i - 1} (1 + i \kappa) \geq 1 - 2 (1 + i^2) \cdot c^{i - 1}$ a path of length $i^3$ ends up in a BSCC within $i - 1$ steps and then visits all states of the BSCC at least $i$ times, proving \ref{proof:dql:appear_is_ec_probability:absolute_bound}.

	Let $\targetset_\strategy = \{a \in \States_\strategy \mid \actionstate<\MDP>(a) \in \targetset \}$ the states of $\MC_\strategy$ corresponding to the given state set $\targetset$.
	Recall that we assumed that $\ProbabilityMDP<\MDP, s><\strategy>[\reach \overline{\targetset}] = 0$ for $s \in \targetset$, i.e.\ $\ProbabilityMC<\MC, a>[\reach \overline{\targetset_\strategy}] = 0$ for all $a \in \targetset_\strategy$ (recall that the states of $\MC$ are actions $a$ of $\MDP$).
	Consequently, each BSCC of $\MC_\strategy$ either is contained in $\targetset_\strategy$ or disjoint from it:
	Assume that there exists a BSCC $R$ with states $a, a' \in R$ where $a \in \targetset_\strategy$ and $a' \notin \targetset_\strategy$.
	Since $R$ is a BSCC, we have $\ProbabilityMC<\MC_\strategy, a>[\reach \{a'\}] = 1$, contradicting $\ProbabilityMC<\MC_\strategy, a>[\reach \overline{\targetset_\strategy}] = 0$.

	Due to \ref{proof:dql:appear_is_ec_probability:assumption}, there exists at least one BSCC which is disjoint from $\targetset_\strategy$---otherwise any run would eventually end up in $\targetset_\strategy$.
	Let $s$ be some state in this BSCC.
	By construction, there exists a path of length at most $\kappa$ from $\initialstate$ to $s$ \ref{proof:dql:appear_is_ec_probability:kappa_is_states}, and thus the probability of reaching such a BSCC is bounded from below by $\delta_{\min}(\strategy)^\kappa$, using \ref{proof:dql:appear_is_ec_probability:delta_min}.
	Formally, we have \plabel{proof:dql:appear_is_ec_probability:not_reach_probability}
	\begin{equation*}
		\ProbabilityMC*<\MC_\strategy, \initialstate>[\overline{\boundedreach<2 i^3> \targetset_\strategy}] > \delta_{\min}(\strategy)^\kappa.
	\end{equation*}
	Finally, we obtain
	\begin{align*}
		& \ProbabilityMC*<\MC_\strategy, \initialstate>[{App}_{i, \strategy} \mid \overline{\boundedreach<2 i^3> \targetset}] \overset{\ref{proof:dql:appear_is_ec_probability:assumption}}{~=~} \ProbabilityMC*<\MC_\strategy, \initialstate>[{App}_{i, \strategy} \intersection \overline{\boundedreach<2 i^3> \targetset}] / \ProbabilityMC*<\MC_\strategy, \initialstate>[\overline{\boundedreach<2 i^3> \targetset}] \\
		& \qquad ~=~ \ProbabilityMC*<\MC_\strategy, \initialstate>[{App}_{i, \strategy} \setminus \boundedreach<2 i^3> \targetset] / \ProbabilityMC*<\MC_\strategy, \initialstate>[\overline{\boundedreach<2 i^3> \targetset}] \\
		& \qquad ~=~ (\ProbabilityMC*<\MC_\strategy, \initialstate>[{App}_{i, \strategy}] - \ProbabilityMC*<\MC_\strategy, \initialstate>[{App}_{i, \strategy} \intersection \boundedreach<2 i^3> \targetset]) / \ProbabilityMC*<\MC_\strategy, \initialstate>[\overline{\boundedreach<2 i^3> \targetset}] \\
		& \qquad ~\geq~ (\ProbabilityMC*<\MC_\strategy, \initialstate>[{App}_{i, \strategy}] - \ProbabilityMC*<\MC_\strategy, \initialstate>[\boundedreach<2 i^3> \targetset]) / \ProbabilityMC*<\MC_\strategy, \initialstate>[\overline{\boundedreach<2 i^3> \targetset}] \\
		& \qquad \overset{\mathclap{\ref{proof:dql:appear_is_ec_probability:absolute_bound}}}{~\geq~} (1 - 2 c^{i - 1} (1 + i^2) - (1 - \ProbabilityMC*<\MC_\strategy, \initialstate>[\overline{\boundedreach<2 i^3> \targetset}])) / \ProbabilityMC*<\MC_\strategy, \initialstate>[\overline{\boundedreach<2 i^3> \targetset}] \\
		& \qquad ~=~ (\ProbabilityMC*<\MC_\strategy, \initialstate>[\overline{\boundedreach<2 i^3> \targetset}] - 2 c^{i - 1} (1 + i^2)) / \ProbabilityMC*<\MC_\strategy, \initialstate>[\overline{\boundedreach<2 i^3> \targetset}] \\
		& \qquad ~=~ 1 - (2 c^{i - 1} (1 + i^2)) / \ProbabilityMC*<\MC_\strategy, \initialstate>[\overline{\boundedreach<2 i^3> \targetset}] \\
		& \qquad \overset{\mathclap{\ref{proof:dql:appear_is_ec_probability:not_reach_probability}}}{~\geq~} 1 - 2 (1 + i^2) \cdot c^{i - 1} \cdot \delta_{\min}(\strategy)^{-\kappa}. \qedhere 
	\end{align*}
\end{proof}
\subsection{The General DQL Algorithm}
\newcommand{\algolimit}{\mathsf{i}}
\newcommand{\algocollapsedstates}{\mathsf{collapsed}}
\newcommand{\algozerostates}{\mathsf{Z}}
\newcommand{\algorepresentative}{\mathsf{rep}}

 	\begin{algorithm}[!tp]
     \setcounter{AlgoLine}{0} 
  		\caption{The DQL learning algorithm for general MDPs.} \label{alg:dql}
  		\KwIn{Inputs as given in \cref{def:limited_information}, precision $\varepsilon$, and confidence $\delta$.}
  		\KwOut{Values $(l, u)$ which are $\varepsilon$-optimal, i.e., $\val(\initialstate) \in [l, u]$ and $0 \leq u - l < \varepsilon$, with probability at least $1 - \delta$.}
  		
  		Initialize all variables as in \cref{alg:dql_no_ec}\;
  		$\algoepisode \gets 1$, $\algostep \gets 1$\;
  		\lFor{$s \in \States$}{
  			$\algocollapsedstates_\algoepisode(s) \gets s$
  		}
  		$\States_1 \gets \States$, $\stateactions_1 \gets \stateactions$, $\targetset_1 \gets \targetset$, $\algozerostates_1 \gets \emptyset$\;
  		
  		\While{$\upperbound_\algostep(\initialstate) - \lowerbound_\algostep(\initialstate) \geq \varepsilon$}{
  			\lFor{$s \in \States_\algoepisode$}{
  				$\umaxactions_\algoepisode(s) \gets \argmax_{a \in \stateactions_\algoepisode(s)} \upperbound_\algostep(a)$
  			}
  			
  			$s_\algostep \gets \initialstate$, $\algostep_\algoepisode \gets \algostep$\;
  			
  			\While{$s_\algostep \notin \targetset_\algoepisode \cup \algozerostates_\algoepisode$ \textbf{and} $\algostep - \algostep_\algoepisode < 2 \algolimit^3$\label{alg:dql:line:sample_loop}}{
  				$a_\algostep \gets$ sampled uniformly from $\umaxactions_\algoepisode(s_\algostep)$\tcp*{Pick an action}
  				$s_\algostep'' \gets \successor(a_\algostep)$\label{alg:dql:line:successor_oracle}\tcp*{Query successor oracle}
  				$s_\algostep' \gets \algorepresentative_\algoepisode(s_\algostep'')$\label{alg:dql:line:resolve_rep}\;
  				Perform updates as in \cref{alg:dql_no_ec} \label{alg:dql:line:update_bounds}\tcp*{Update Bounds}
  				$s_{\algostep + 1} \gets s_\algostep'$, $\algostep \gets \algostep + 1$\;
  			}
  			
  			\If(\tcp*[f]{Update ECs}){$\algostep - \algostep_\algoepisode \geq 2 \algolimit^3$}{
  				$(\mathsf{R}, \mathsf{B}) \gets \appear(s_{\algostep_\algoepisode} a_{\algostep_\algoepisode} s_{\algostep_\algoepisode + 1} \dots a_{\algostep - 1} s_\algostep, \algolimit, 2\algolimit^3)$\label{alg:dql:line:appear}\;
  				$C \gets \bigcup_{s \in \mathsf{R}} \stateactions_{\algoepisode}(s) \setminus B$\;
  				
  				\If{$\mathsf{B} \neq \emptyset$}{
  					\If{$\targetset_\algoepisode \cap \mathsf{R} \neq \emptyset$}{
  						$\targetset_{\algoepisode + 1} \gets \targetset_\algoepisode \cup \mathsf{R}$\;
  						\lFor{$a \in \mathsf{B}$}{
  							$\lowerbound_\algostep(a) \gets 1$\label{alg:dql:line:set_lower_one}
  						}
  					}
  					\ElseIf{$C = \emptyset$}{
  						$\algozerostates_{\algoepisode + 1} \gets \algozerostates_\algoepisode \cup \mathsf{R}$\;
  						\lFor{$a \in \mathsf{B}$}{
  							$\upperbound_\algostep(a) \gets 0$ \label{alg:dql:line:set_upper_zero}
  						}
  					}
  					\Else{
  						$\States_{\algoepisode + 1} \gets (\States_\algoepisode \setminus R) \cup \{s_{(R, B)}\}$\label{alg:dql:line:add_rep_state}\;
  						$\stateactions_{\algoepisode + 1}(s_{(R, B)}) \gets C$\label{alg:dql:line:set_representative_actions}\;
  						\lFor{$s \in \mathsf{R} \cup \{ s_{(\mathsf{R}, \mathsf{B})} \}$}{
  							$\algocollapsedstates_{\algoepisode + 1}(s) \gets s_{(\mathsf{R}, \mathsf{B})}$
  						}
  						\lIf{$\initialstate \in \mathsf{R}$}{
  							$\initialstate \gets s_{(\mathsf{R}, \mathsf{B})}$\label{alg:dql:line:update_initial}
  						}
  					}
  				}
  			}
  			
  			$\algoepisode \gets \algoepisode + 1$\;
  		}  		
  		\Return $(\lowerbound_\algostep(\initialstate), \upperbound_\algostep(\initialstate))$\;
  	\end{algorithm}

We define the general DQL algorithm in \cref{alg:dql}.
Essentially, the algorithm works similar to the previous \cref{alg:dql_no_ec}.
The main difference is that it further employs \cref{stm:dql:appear_is_ec_probability} to detect whether the current sample is stuck in a yet to be discovered EC.
To this end, the algorithm introduces a small set of additional auxiliary variables, necessary to track representative states similar to the collapsed MDP of \cref{sec:brtdp}.
In particular, $\algocollapsedstates_\algoepisode$ stores the representatives of each state.
Since we might discover growing ECs, this representative might be part of another already discovered EC.
Thus, we use $\algorepresentative_\algoepisode$ to resolve the current representative of a given state $s$ by repeatedly applying $\algocollapsedstates_\algoepisode$ until a fixed point is reached.
Practically, we would store $\algorepresentative_\algoepisode$ as a map, pointing each \enquote{original} state of the MDP to its current representative.
We introduce the \enquote{layered} representation through $\algocollapsedstates_\algoepisode$ only as notational convenience.
Additionally, $\algozerostates_\algoepisode$ contains all states which are part of a bottom EC without a target state.
We choose the parameter $\algolimit$, controlling the length of each sample and when to check for an EC, such that
\begin{equation}
	\cardinality{\Actions} \cdot 2 (1 + \algolimit^2) \cdot e^{- (\algolimit - 1) \frac{p_{\min}(\strategy)^{\cardinality{\States} + 1}}{\cardinality{\States} + 1}} \cdot p_{\min}(\strategy)^{-(\cardinality{\States}+1)} \leq \frac{\delta}{4} \qquad \text{and} \qquad \algolimit \geq \cardinality{\Actions}. \label{eq:dql:choice_of_limit}
\end{equation}
This technical choice becomes more apparent in the proof of \cref{stm:dql:always_collapse_ec_probability}.
Note that such an $\algolimit$ always exists since the left side of the first inequality converges to $0$ for $\algolimit \to \infty$.
Moreover, we can find such an $\algolimit$ using the values provided by the limited information setting as defined in \cref{def:limited_information}.

\begin{remark}
	In contrast to the previous sections, the domain of the bounds $\upperbound$ and $\lowerbound$ are actions instead of state-action pairs.
	This simplifies notation, since the algorithm frequently changes the state associated with an action.
\end{remark}

\begin{remark}
	We implicitly assume that we can continue sampling with an action of our choice:
	When we collapse, for example, an EC $(\mathsf{R}, \mathsf{B})$ with states $s, s' \in \mathsf{R}$ into a single representative state, we might enter the EC in state $s$ but then continue sampling with an action $a \in \stateactions(s')$.
	This is not an essential restriction:
	Upon entering an already detected EC, we can simply pause the algorithm and randomly pick actions in $\mathsf{B}$ until we reach the state enabling the next action mandated by the algorithm.
\end{remark}

\subsection{Proof of Correctness}
Now, to prove correctness of the algorithm, we again can reuse a lot of the previous reasoning.
However, we need to invest significant effort in the treatment of end components.
First of all, we again prove that the algorithm is well-defined.
\begin{lemma} \label{stm:dql:disjoint_actions}
	During all episodes, we have that $\stateactions_{\algoepisode}(s) \intersection \stateactions_{\algoepisode}(s') = \emptyset$ for all states $s, s' \in \States_\algoepisode$ with $s \neq s'$.
\end{lemma}

\begin{proof}
	The algorithm only modifies the set of available actions $\stateactions_{\algoepisode}$ whenever a new representative state $s_{(R, B)}$ is added.
	In this case, we have $\stateactions_{\algoepisode + 1}(s_{(R, B)}) \gets C \subseteq \Union_{s \in R} \stateactions_{\algoepisode}(s)$ and all states of $R$ are removed. 
\end{proof}

\begin{lemma}
	\cref{alg:dql} is well-defined.
\end{lemma}

\begin{proof}
	To prove this statement, we have to show that (i)~no undefined values are accessed, (ii)~all assignments are free of contradictions, and (iii)~we require no more information than given by \cref{def:limited_information}.

	For (i) and (ii), observe that when assigning the next episode's variables, we only use the variables of the current episode.
	Since we copy all unchanged variables, we only need to take care of the newly introduced arguments, i.e.\ the representative states $s_{(R, B)}$.
	Such a state is only added in \cref{alg:dql:line:add_rep_state}.
	In the following lines, we define the state's actions $\stateactions$, which is non-empty and disjoint from other states by \cref{stm:dql:disjoint_actions}.
	As no new actions are added, the action values in $s_{(R,B)}$ still are defined.
	Observe that in \cref{alg:dql:line:successor_oracle} the successor oracle is only given states of the original MDP.
	Claim (iii) follows immediately. 
\end{proof}
Now, we show several statements related to the newly added handling of end components.
Our goal is to show that the algorithm essentially samples from a collapsed MDP where the ECs identified by the algorithm are collapsed.
Then, we replicate the proof ideas of the EC-free DQL algorithm on this collapsed MDP in order to again obtain the correctness.
\begin{lemma} \label{stm:dql:finite_representative_states}
	\cref{alg:dql} enters \cref{alg:dql:line:appear} at most $\cardinality{\Actions}$ times.
\end{lemma}

\begin{proof}
	First, observe that due to the pigeon-hole principle, $B$ never is empty:
	By \eqref{eq:dql:choice_of_limit}, our choice of $\algolimit$ is larger than $\cardinality{\Actions}$, thus a path of length at least $\algolimit^2$ contains at least one action $\algolimit$ times.
	Consequently, whenever the algorithm enters \cref{alg:dql:line:appear}, $B$ is non-empty.
	Initially, the size of $B$ is bounded by $\sum_{s \in \States_1} \cardinality{\stateactions_1(s)} = \cardinality{\Actions}$.
	We show that in any of the three cases, we remove at least one action which can never occur again as part of $B$.
	Consequently, after at most $\cardinality{\Actions}$ visits to \cref{alg:dql:line:appear}, $B$ would necessarily be empty, contradicting the above.

	Whenever a state is added to either $\targetset_\algoepisode$ or $\algozerostates_\algoepisode$, this state and its actions will not be considered again---in particular, it will not occur as part of $B$.
	For the third case, we show that the number of available actions $\sum_{s \in \States_\algoepisode} \cardinality{\stateactions_{\algoepisode}(s)}$ is reduced whenever a new representative state is added.
	In that case, we have $C \gets \Union_{s \in R} \stateactions_{\algoepisode}(s) \setminus B$, $\States_{\algoepisode + 1} \gets (\States_\algoepisode  \setminus R) \union \{s_{(R, B)}\}$, and $\stateactions_{\algoepisode + 1}(s_{(R, B)}) \gets C$.
	By construction of the algorithm and definition of $\appear$, we have $\emptyset \neq B \subseteq \Union_{s \in R} \stateactions_\algoepisode(s)$.
	Using \cref{stm:dql:disjoint_actions} we thus have $\cardinality{C} < \cardinality{\Union_{s \in R} \stateactions_\algoepisode(s)}$.
	Consequently, $\sum_{s \in \States_{\algoepisode + 1}} \cardinality{\stateactions_{\algoepisode + 1}(s)} < \sum_{s \in \States_\algoepisode} \cardinality{\stateactions_{\algoepisode}(s)}$.\end{proof}

\begin{lemma} \label{stm:dql:terminate_or_infinite_episodes}
	\cref{alg:dql} terminates or experiences an infinite number of episodes.
\end{lemma}

\begin{proof}
	Since the length of each episode is limited, i.e.\ the loop of \cref{alg:dql:line:sample_loop} always terminates after a bounded number of steps, we only need to show that all other loops terminate.
	All for-loops iterate over (sub-)sets of states or actions, which are finite by assumption.
	The only remaining loop is the computation of $\algorepresentative_\algoepisode$ in \cref{alg:dql:line:resolve_rep}, where the representative state is resolved.
	Observe that by construction of the algorithm, we either have that $\algocollapsedstates_\algoepisode(s) = s$ or $\algocollapsedstates_\algoepisode(s) = s_{(R, B)}$ with $s \in R$.
	Since we only modify $\algocollapsedstates$ when a new representative state is added, this happens only finitely often, due to \cref{stm:dql:finite_representative_states}. 
\end{proof}

\begin{lemma} \label{stm:dql:collapsed_states_unchanged}
	If we add a representative state $s_{(R, B)}$ in \cref{alg:dql:line:add_rep_state} after an episode $\algoepisode$ the bounds of any action $a \in B$ are not changed after episode $\algoepisode$.
\end{lemma}

\begin{proof}
	During each episode $\algoepisode$, we only consider states in $\States_\algoepisode$ and actions which are available in such states, as the call to $\algorepresentative_\algoepisode$ in \cref{alg:dql:line:resolve_rep} always yields an element of the current state set $\States_\algoepisode$ due to \cref{stm:dql:representative_in_states}.
	Since all states corresponding to actions in $B$ are removed when adding a representative state $s_{(R, B)}$ and these actions are not enabled in the newly added state, they do not appear again. 
\end{proof}

\begin{lemma} \label{stm:dql:representative_in_states}
	For any execution of the algorithm, we always have that $\algorepresentative_\algoepisode(s) \in \States_\algoepisode$ for any state $s \in \States$.
\end{lemma}

\begin{proof}
	We prove by induction:
	Initially, we have $\algorepresentative_1(s) = \algocollapsedstates_1(s) = s$ for all $s \in \States_1$ by definition.
	Whenever we modify $\States_\algoepisode$, i.e.\ remove some states $R$ and add a representative $s_{(R, B)}$, we set $\algocollapsedstates_{\algoepisode + 1}(s) \gets s_{(R, B)} \in \States_{\algoepisode + 1}$ for all $s \in R$. 
\end{proof}
In order to properly reason about the paths sampled by the algorithm, we introduce a special MDP which corresponds to the current \enquote{view} of the given MDP.
\begin{definition} \label{def:dql:sampling_mdp}
	For any episode $\algoepisode$, let the \emph{sampling MDP} $\MDP_\algoepisode = (\States_\algoepisode, \Actions_\algoepisode, \stateactions_\algoepisode, \mdptransitions_\algoepisode)$,
	\begin{align*}
		\mdptransitions_\algoepisode(s, a) & = \{s \mapsto 1\} \quad \text{for $s \in \States_\algoepisode \intersection (\targetset_\algoepisode \union \algozerostates_\algoepisode)$, $a \in \stateactions_\algoepisode(s)$, and} \\
		\mdptransitions_\algoepisode(s, a, s') & = {\sum}_{\{s'' \in \States \mid \algorepresentative_\algoepisode(s'') = s'\}} \mdptransitions(\actionstate<\MDP>(a), a, s'') \quad \text{for other states $s$, $a \in \stateactions_\algoepisode(s)$},
	\end{align*}
	and $\Actions_\algoepisode = \Union_{s \in \States_\algoepisode} \stateactions_\algoepisode(s)$.
\end{definition}
Note that the sampling MDP is well-defined due to \cref{stm:dql:disjoint_actions,stm:dql:representative_in_states}.
\begin{lemma} \label{stm:dql:sampling_mdp}
	Fix an execution of the algorithm until some episode $\algoepisode$ and let $\finitepath$ be the finite path sampled by the algorithm during episode $\algoepisode$.
	The probability of sampling this path equals the probability of obtaining this path on $\MDP_\algoepisode$ following the strategy $\strategy_\algoepisode$ starting in state $\initialstate$.
\end{lemma}

\begin{proof}
	We prove by induction over the path $\finitepath$, using the Markov property.
	We show that for any finite prefix, the probability of selecting action $a$ and then reaching state $s'$ in the next step is equal in both the algorithm and the sampling MDP.
	Observe that we always have $\initialstate \in \States_\algoepisode$ due to \cref{alg:dql:line:update_initial} and the induction start is trivial.

	For the induction step, suppose we are in a state $s$.
	By construction of the algorithm, $s \notin \targetset_\algoepisode \union \algozerostates_\algoepisode$.
	The algorithm now uniformly selects an action $a$ from $\umaxactions_\algoepisode(s)$, i.e.\ with probability $\cardinality{\umaxactions_\algoepisode(s)}^{-1}$ for any such action.
	Then, a successor $s'' \in \States$ is sampled according to $\successor(s, a)$, i.e.\ with probability $\mdptransitions(s, a, s'')$.
	The overall successor then equals $s' = \algorepresentative_\algoepisode(s'')$.
	We have $s' \in \States_\algoepisode$ by \cref{stm:dql:representative_in_states}.
	Hence, a state $s' \in \States_\algoepisode$ is sampled with probability $\sum_{\{s'' \in \States \mid \algorepresentative_\algoepisode(s'') = s'\}} \mdptransitions(s, a, s'')$, just as in the MDP $\MDP_\algoepisode$ under strategy $\strategy_\algoepisode$. 
\end{proof}

\begin{assumption} \label{asm:dql:always_collapse_ec}
	Whenever the algorithm reaches \cref{alg:dql:line:appear}, $(\mathsf{R}, \mathsf{B})$ is an EC of $\MDP_\algoepisode$.
\end{assumption}

\begin{lemma} \label{stm:dql:always_collapse_ec_probability}
	The probability that \cref{asm:dql:always_collapse_ec} is violated during the execution of \cref{alg:dql} is bounded by $\frac{\delta}{4}$.
\end{lemma}

\begin{proof}
	We apply \cref{stm:dql:appear_is_ec_probability} with $\MDP = \MDP_\algoepisode$, $\targetset = \targetset_\algoepisode \union \algozerostates_\algoepisode$ and $\strategy = \strategy_\algoepisode$.
	By construction of $\MDP_\algoepisode$ and the choice of $\targetset$, we have that $\strategy_\algoepisode$ trivially satisfies the condition of this lemma, since each state in $\targetset$ only has self-loops in $\MDP_\algoepisode$.
	Clearly, we have that $\cardinality{\States_\strategy} \leq \sum_{s \in \States_\algoepisode} \cardinality{\stateactions(s)} \leq \cardinality{\Actions}$, since no actions are added during the execution of the algorithm.
	Consequently, we have that either $\ProbabilityMDP<\MDP_\algoepisode, \initialstate><\strategy_\algoepisode>[\boundedreach<2i^3> (\targetset_\algoepisode \union \algozerostates_\algoepisode)] = 1$ or
	\begin{equation*}
		\ProbabilityMDP*<\MDP_\algoepisode, \initialstate><\strategy_\algoepisode>[{App}_\algolimit \mid \overline{\boundedreach<2\algolimit^3> (\targetset_\algoepisode \union \algozerostates_\algoepisode)}] \geq 1 - 2 (1 + \algolimit^2) \cdot e^{- (\algolimit - 1) \frac{p_{\min}(\strategy)^{\cardinality{\States} + 1}}{\cardinality{\States} + 1}} \cdot p_{\min}(\strategy)^{-(\cardinality{\States}+1)},
	\end{equation*}
	where ${App}_\algolimit$ are all paths $\infinitepath \in \Infinitepaths<\MDP_\algoepisode>$ such that $\appear(\infinitepath, \algolimit, 2\algolimit^3)$ is an EC in $\MDP_\algoepisode$.

	Now, observe that the algorithm only enters \cref{alg:dql:line:appear} if after $2 \algolimit^3$ steps neither $\targetset_\algoepisode$ nor $\algozerostates_\algoepisode$ is reached.
	By applying \cref{stm:dql:sampling_mdp}, we get that the probability of $(R, B)$ being an EC given that \cref{alg:dql:line:appear} is entered exactly equals $\ProbabilityMDP<\MDP_\algoepisode, \initialstate><\strategy_\algoepisode>[{App}_\algolimit \mid \overline{\boundedreach<2\algolimit^3> (\targetset_\algoepisode \union \algozerostates_\algoepisode)}]$.
	Since \cref{alg:dql:line:appear} is entered at most $\cardinality{\Actions}$ times due to \cref{stm:dql:finite_representative_states}, the statement follows by inserting the definition of $\algolimit$ from \eqref{eq:dql:choice_of_limit}. 
\end{proof}

\begin{lemma} \label{stm:dql:values_collapse}
	Assume that \cref{asm:dql:always_collapse_ec} holds and fix some episode $\algoepisode$.
	Let $s \in \States_\algoepisode$ some state of the MDP $\MDP_\algoepisode$ and $s' \in \States$ such that $\algorepresentative_\algoepisode(s') = s$
	Then, $s$ and $s'$ have the same value:
	\begin{equation*}
		\val_\algoepisode(s) = \ProbabilityMDPmax<\MDP_\algoepisode, s>[\reach \targetset_\algoepisode] = \ProbabilityMDPmax<\MDP, s'>[\reach \targetset] = \val(s')
	\end{equation*}
\end{lemma}

\begin{proof}
	We prove by induction over the episode number.
	Initially, we have that $\MDP_1$ is quite similar to the original MDP $\MDP$.
	Recall that $\algozerostates_1 = \emptyset$ and $\algorepresentative_1(s) = s$ for all states.
	Hence, the only difference lies in the transition function of all states $s \in \targetset$.
	These only have self-loops in $\MDP_1$, while in $\MDP$ they may have arbitrary transitions.
	This is irrelevant for the value of the states, since it equals 1 in both cases.

	Now fix an arbitrary episode $\algoepisode$.
	We have that $\val_\algoepisode(s) = \val(s')$ \alabel{proof:dql:values_collapse:ih}{IH} for any two states $s$, $s'$ as in the claim.
	$\MDP_\algoepisode$ is only modified when \cref{alg:dql:line:appear} is entered.
	Let $(R, B)$ the identified set of states and actions.
	Due to \cref{asm:dql:always_collapse_ec}, $(R, B)$ is an EC of $\MDP_\algoepisode$.
	To conclude, we distinguish the three cases in the algorithm:
	\begin{itemize}
		\item
		$\targetset_\algoepisode \intersection R \neq \emptyset$: Since $(R, B)$ is an EC, any state $s \in R$ can reach $\targetset_\algoepisode$ with probability one.
		Hence $\val_{\algoepisode + 1}(s) = 1 = \val_\algoepisode(s) = \val(s')$ \ref{proof:dql:values_collapse:ih}.
		In particular, by adding all states of $R$ to $\targetset_{\algoepisode + 1}$, we do not change their value.

		\item
		$C = \emptyset$: Once in $R$, this EC cannot be left, i.e.\ $\ProbabilityMDPmax<\MDP_\algoepisode, s>[\reach \overline{R}] = 0$ for all $s \in R$.
		Consequently, we have that $\val_\algoepisode(s) = 0 = \val(s')$ \ref{proof:dql:values_collapse:ih}.
		This value is unchanged by adding the states of $R$ to $\algozerostates_{\algoepisode + 1}$ and thus introducing a self-loop in $\MDP_\algoepisode$.

		\item
		Add a representative state:
		By assumption, we have that $\algorepresentative_\algoepisode(s') \in R$ and thus $\algorepresentative_{\algoepisode + 1}(s') = s_{(R, B)}$.
		We need to prove that $\val_{\algoepisode + 1}(s_{(R, B)}) = \val(s')$.
		As $(R, B)$ is an EC by assumption, each state in $R$ has the same value by \cref{stm:ec_same_value}.
		The representative state $s_{(R,B)}$ has this value by applying the same reasoning as in \cref{stm:collapse:reachability_equal}. \qedhere
	\end{itemize}
\end{proof}

\begin{lemma} \label{stm:dql:ec_contained_in_previous}
	Assume that \cref{asm:dql:always_collapse_ec} holds and fix some episode $\algoepisode$.
	For any EC $(R, B) \in \Ecs(\MDP_\algoepisode)$ and $\algoepisode' \leq \algoepisode$ there exists an EC $(R', B') \in \Ecs(\MDP_{\algoepisode'})$ with $B \subseteq B'$.
\end{lemma}

\begin{proof}
	Note that we do not necessarily have that $R \subseteq R'$, since some states of the EC may have been replaced by a representative state.

	We prove by induction on the episode $\algoepisode$.
	Fix any such episode $\algoepisode$ and EC $(R, B) \in \Ecs(\MDP_{\algoepisode + 1})$.
	We only modify the MDP $\MDP_\algoepisode$ when the algorithm enters \cref{alg:dql:line:appear}, hence w.l.o.g.\ we assume that this happened in episode $\algoepisode$.
	Let $(\mathsf{R}, \mathsf{B})$ be the set of states and actions identified in \cref{alg:dql:line:appear} during episode $\algoepisode$.
	By \cref{asm:dql:always_collapse_ec}, $(\mathsf{R}, \mathsf{B})$ is an EC of $\MDP_\algoepisode$.
	As above, we distinguish the three cases in the algorithm:
	\begin{itemize}
		\item $\targetset_\algoepisode \intersection \mathsf{R} \neq \emptyset$:
			Then, all actions in $\mathsf{B}$ are changed to a self-loop in $\MDP_{\algoepisode + 1}$ and hence we either have $B = \{a\} \subseteq \mathsf{B}$ or $B \intersection \mathsf{B} = \emptyset$.
			In the former case, $(\mathsf{R}, \mathsf{B})$ satisfies the conditions of the claim.
			In the latter, the EC $(R, B)$ already existed in $\MDP_\algoepisode$, since no state or action of $(R, B)$ was modified.
		\item $C = \emptyset$:
			Analogously to the above, all actions in $\mathsf{B}$ are now a self-loop in $\MDP_{\algoepisode + 1}$ and the same reasoning applies.
		\item Add a representative state:
			If $s_{(\mathsf{R}, \mathsf{B})} \notin R$, we necessarily have that $\mathsf{B} \intersection B = \emptyset$.
			Hence, the EC $(R, B)$ again already existed in $\MDP_\algoepisode$, since none of its components was modified by this step.
			If instead $s_{(\mathsf{R}, \mathsf{B})} \in R$, we have that $(\mathsf{R} \union R, \mathsf{B} \union B)$ is an EC in $\MDP_\algoepisode$, following the same reasoning as in \cref{stm:collapse:ec_correspondence_collapse_to_normal}. \qedhere 
	\end{itemize}
\end{proof}

\begin{lemma} \label{stm:dql:ec_bounds_one_or_zero}
	Assume that \cref{asm:dql:always_collapse_ec} holds and fix some step $\algostep$ with corresponding episode $\algoepisode$.
	Let $(R, B) \in \Ecs(\MDP_\algoepisode)$ be any EC in $\MDP_\algoepisode$.
	For any $a \in B$ we have that (i)~if $\actionstate<\MDP_\algoepisode>(a) \in \algozerostates_\algoepisode$, then $\upperbound_\algostep(a) = 0$ and (ii)~$\upperbound_\algostep(a) = 1$ otherwise.
\end{lemma}

\begin{proof}
	Item~(i) immediately follows from the definition of the algorithm and $\MDP_\algoepisode$.
	When a state is added to $\algozerostates_\algoepisode$, we set $\upperbound_\algostep(a) = 0$ for all its actions.
	We prove Item~(ii) by induction, showing that the statement holds for all ECs at each step $\algostep$.
	Initially, we have $\upperbound_1(a) = 1$ for all actions by definition of the algorithm.
	For the induction step fix some step $\algostep$.
	We have that $\upperbound_{\algostep'}(a) = 1$ for all actions $a$ in all ECs without zero-states for all $\algostep' \leq \algostep$ \alabel{proof:dql:ec_bounds_one_or_zero:ih}{IH}.
	Now, let $\algoepisode'$ be the episode of step $\algostep + 1$ and fix any EC $(R, B)$ in $\MDP_{\algoepisode'}$ with $R \intersection \algozerostates_\algoepisode = \emptyset$.
	By repeatedly applying \cref{stm:dql:ec_contained_in_previous}, there exists an EC $(R_{\algoepisode'}, B_{\algoepisode'}) \in \Ecs(\MDP_{\algoepisode'})$ with $B \subseteq B_{\algoepisode'}$ for all $\algoepisode' \leq \algoepisode$.
	Since we have no zero-states in the EC in step $\algostep + 1$, none of the $R_{\algoepisode'}$ contain zero-states either, by construction of the algorithm and $\MDP_\algoepisode$.
	Thus, the induction hypothesis \ref{proof:dql:ec_bounds_one_or_zero:ih} is applicable and we have that $\upperbound_{\algostep'}(a) = 1$ for any action $a \in B_{\algoepisode'}$ and $\algostep' \leq \algostep$.
	Hence, we necessarily have that $\upperbound_{\algostep'}(s) = 1$ for all $s \in R_{\algoepisode'}$ and $\algostep' \leq \algostep$ (also using \cref{stm:dql:collapsed_states_unchanged}).
	Whenever any action $a \in B$ is selected at any step $\algostep' \leq \algostep$ during episode $\algoepisode' \leq \algoepisode$, all of its successors are part of the EC $(R_{\algoepisode'}, B_{\algoepisode'})$, thus $\upperbound_{\algostep'}(s) = 1$ for all successors by the above reasoning.
	Consequently, we always add a value of $1$ to $\acc_\algostep^\upperbound(a)$ and whenever an $\upperbound$-update is attempted for action $a$ at some step $\algostep' \leq \algostep$, we would set $\upperbound_{\algostep'}(a) = 1$. 
\end{proof}

\begin{lemma} \label{stm:dql:bounds_monotone}
	Assume that \cref{asm:dql:always_collapse_ec} holds and fix some step $\algostep$ with corresponding episode $\algoepisode$.
	Let $\algostep' \geq \algostep$ with episode $\algoepisode' \geq \algoepisode$.
	We have for any state $s \in \States$ that $\upperbound_{\algostep'}(\algorepresentative_{\algoepisode'}(s)) \leq \upperbound_{\algostep}(\algorepresentative_\algoepisode(s))$ and $\lowerbound_{\algostep}(\algorepresentative_\algoepisode(s)) \leq \lowerbound_{\algostep'}(\algorepresentative_{\algoepisode'}(s))$.
\end{lemma}

\begin{proof}
	The bounds of actions are modified by (i)~the usual update, which only increases or decreases, respectively (ii)~in \cref{alg:dql:line:set_upper_zero,alg:dql:line:set_lower_one}, where upper bounds are set to $0$ and lower bounds set to $1$, or (iii)~when an EC is collapsed and thus the set of available actions is modified in \cref{alg:dql:line:set_representative_actions}.
	Cases (i) and (ii) preserve monotonicity of the state bound by definition.
	Case (iii) is proven separately for upper and lower bounds, with the proof of the lower bound being significantly more involved.
	For the upper bounds, observe that $\stateactions_{\algoepisode'}(s) \subseteq \stateactions_{\algoepisode}(s)$ by definition, i.e.\ we never add new actions to any state.
	Consequently, the maximum over the set of available actions does not increase.
	For the lower bounds, we have to show that while collapsing ECs and thus removing actions, we never remove all those which are optimal w.r.t.\ the lower bound, i.e.\ all actions $a \in \stateactions_\algoepisode(s)$ with $\lowerbound_\algostep(a) = \lowerbound_\algostep(s)$.

	We proceed by additionally proving an auxiliary statement by induction on the step $\algostep$ in parallel.
	In particular, we prove that for any step $\algostep$ with corresponding episode $\algoepisode$ (i)~the statement of the lemma holds \alabel{proof:dql:lower_bounds_monotone:induction_lemma}{IH1} and (ii)~$\lowerbound_{\algostep}(a) \leq \max_{s \in R, a' \in \stateactions_{\algoepisode}(s) \setminus B} \lowerbound_\algostep(a')$ for all actions $a \in B$ (or $0$ if no such actions $a'$ exist) in all ECs $(R, B) \in \Ecs(\MDP_\algoepisode)$ without a target state, i.e.\ $R \intersection \targetset_\algoepisode = \emptyset$. \alabel{proof:dql:lower_bounds_monotone:induction_aux}{IH2}.

	Initially, we have $\lowerbound_1(a) = 0$ by definition of the algorithm and both statements trivially hold.
	For the induction step fix some time step $\algostep$.
	We first treat the case when the lower bound of action an action $a$ is successfully updated in step $\algostep$ and later on deal with the case of an EC being collapsed.
	Note that \ref{proof:dql:lower_bounds_monotone:induction_lemma} trivially holds in this case, since the value of $a$ is never decreased.
	We only need to show the second statement \ref{proof:dql:lower_bounds_monotone:induction_aux}, thus assume that the updated action $a$ is an internal action of some EC $(R, B)$, i.e.\ $a \in B$.
	For readability, denote $C = \Union_{s \in R} \stateactions_\algoepisode(s) \setminus B$ the set of outgoing actions of $(R, B)$.
	If $C = \emptyset$, the statement follows directly:
	Since all lower bounds are initialised to zero, the EC does not contain any target states by assumption, and there are no outgoing actions, the algorithm never updates the lower bound of any action in $B$ to a non-zero value.
	Thus, assume that $C \neq \emptyset$.
	By applying \ref{proof:dql:lower_bounds_monotone:induction_aux} to all states of the EC $(R, B)$, we get that ${\max}_{a' \in C} \lowerbound_\algostep(a') = {\max}_{s \in R} \lowerbound_\algostep(s)$ \plabel{proof:dql:lower_bounds_monotone:applied_hypothesis}.
	Furthermore, let $k_1 < \ldots < k_\delay = \algostep$ the steps of the most recent visits to $a$ with corresponding episodes $\algoepisode_1 \leq \ldots \leq \algoepisode_\delay = \algoepisode$ and sampled successors $s_{k_i}'$.
	Now, let $R_i = \algorepresentative_{\algoepisode_i}(\ecstates_\algoepisode(R))$ for $1 \leq i \leq \delay$ the set of states in episode $\algoepisode_i$ which eventually are collapsed to $R$.
	By applying the reasoning of \cref{stm:dql:ec_contained_in_previous} and \ref{stm:collapse:ec_correspondence_collapse_to_normal}, there exists a set of actions $B_i$ with $B \subseteq B_i$ such that $(R_i, B_i)$ is an EC in $\MDP_{\algoepisode_i}$ and thus $s_{k_i}' \in R_i$ \plabel{proof:dql:lower_bounds_monotone:succ_in_Ri}, since $a \in B_i$.
	By construction, we have that $\algorepresentative_{\algoepisode}(R_i) = R$ \plabel{proof:dql:lower_bounds_monotone:rep_is_R}.
	Finally, we observe that the value of the outgoing actions does not decrease, hence the value we assign to $a$ in step $\algostep$ satisfies
	\begin{align*}
		& \lowerbound_{\algostep + 1}(a) + \updatestep \overset{\text{def}}{~=~} \frac{1}{\delay} {\sum}_{i = 1}^\delay \lowerbound_{k_i}(s_{k_i}') \\
		& \qquad \overset{\ref{proof:dql:lower_bounds_monotone:succ_in_Ri}}{~\leq~} \frac{1}{\delay} {\sum}_{i = 1}^\delay {\max}_{s \in R_i} \lowerbound_{k_i}(s) \\
		& \qquad \overset{\ref{proof:dql:lower_bounds_monotone:induction_lemma}}{~\leq~} \frac{1}{\delay} {\sum}_{i = 1}^\delay {\max}_{s \in R_i} \lowerbound_\algostep(\algorepresentative_\algoepisode(s)) \\
		& \qquad \overset{\ref{proof:dql:lower_bounds_monotone:rep_is_R}}{~=~} \frac{1}{\delay} {\sum}_{i = 1}^\delay {\max}_{s \in R_{\delay}} \lowerbound_\algostep(s) \\
		& \qquad ~=~ {\max}_{s \in R_{\delay}} \lowerbound_\algostep(s) \\
		& \qquad \overset{\ref{proof:dql:lower_bounds_monotone:applied_hypothesis}}{~=~} {\max}_{a' \in C_{\delay}} \lowerbound_\algostep(a').
	\end{align*}
	This concludes proof of the first part.

	For the second part, i.e.\ when a set of states is collapsed by the algorithm, we have that the collapsed set $(R, B)$ is an EC by \cref{asm:dql:always_collapse_ec} and $B$ are only internal actions.
	If the collapsed EC contains target states, the statement trivially holds.
	Otherwise, we apply the result of the first part and get that the lower bound assigned to any action in $B$ is less or equal to outgoing actions.
	Thus, removing the actions in $B$ from the set of available actions does not reduce the value of the obtained representative state. 
\end{proof}
With basic properties about the sampling MDP in place, we can now mimic the previous idea of defining \enquote{converged} state-action pairs and, using those, show that the algorithm eventually converges with high probability.
\begin{definition} \label{def:dql:converged_bounds}
	For every step $\algostep$ during episode $\algoepisode$, define $\ConvergedUpperBounds<\algostep>, \ConvergedLowerBounds<\algostep> \subseteq \Actions_{\algoepisode}$ by
	\begin{align*}
		\ConvergedUpperBounds<\algostep> & \coloneqq \{a \mid \upperbound_\algostep(a) - \ExpectedSumMDP{\mdptransitions_\algoepisode}{\actionstate<\MDP_\algoepisode>(a)}{a}{\ExpectedSumStrat{\strategy_\algostep}{\upperbound_\algostep}} \leq 3 \updatestep\} \text{ and}\\
		\ConvergedLowerBounds<\algostep> & \coloneqq \{a \mid \ExpectedSumMDP{\mdptransitions_\algoepisode}{\actionstate<\MDP_\algoepisode>(a)}{a}{\ExpectedSumStrat{\strategy_\algostep}{\lowerbound_\algostep}} - \lowerbound_\algostep(a) \leq 3 \updatestep\}.
	\end{align*}
	Again, an action $a$ is \emph{$\upperbound$-converged ($\lowerbound$-converged) at step $\algostep$} if $a \in \ConvergedUpperBounds<\algostep>$ ($a \in \ConvergedLowerBounds<\algostep>$).
\end{definition}
\begin{assumption} \label{asm:dql:sampled_value_close_to_real_value}
	Suppose an $\upperbound$-update of the action $a$ is attempted at step $\algostep$.
	Let $k_1 < k_2 < \ldots < k_\delay = \algostep$ be the steps of the $\delay$ most recent visits to $a$, and $e_1 \leq e_2 \leq \ldots \leq e_\delay$ the respective episodes.
	Then $\frac{1}{\delay} \sum_{i=1}^\delay \val_{\algoepisode_i}(s_{k_i}') \geq \val_{\algoepisode_\delay}(a) - \updatestep$.
	Analogously, for an attempted $\lowerbound$-update, we have $\frac{1}{\delay} \sum_{i=1}^\delay \val_{\algoepisode_i}(s_{k_i}') \leq \val_{\algoepisode_\delay}(a) + \updatestep$.
\end{assumption}

\begin{assumption} \label{asm:dql:converged_successful_update}
	Suppose an update of the upper bound (lower bound) of the action $a$ is attempted at step $\algostep$.
	Let $k_1 < k_2 < \ldots < k_\delay = \algostep$ be the steps of the $\delay$ most recent visits to $a$.
	If $a$ is not $\upperbound$-converged ($\lowerbound$-converged) at step $k_1$, the update at step $\algostep$ is successful.
\end{assumption}
We replicate most of the statements from the previous DQL algorithm.

\begin{lemma} \label{stm:dql:properties}
	The following properties hold for \cref{alg:dql}.
	\begin{enumerate}
		\item \label{item:stm:dql:properties:success_updates}
		The number of successful updates of $\upperbound$ and $\lowerbound$ is bounded by $\frac{\cardinality{\Actions}}{\updatestep}$ each.

		\item \label{item:stm:dql:properties:attempted_updates}
		The number of attempted updates of $\upperbound$ and $\lowerbound$ is bounded by $\updatecount$.

		\item \label{item:stm:dql:properties:sampled_value_close_to_real_value_probability}
		Assume that \cref{asm:dql:always_collapse_ec} holds.
		Then, the probability that \cref{asm:dql:sampled_value_close_to_real_value} is violated during the execution of \cref{alg:dql} is bounded by $\frac{\delta}{4}$.

		\item \label{item:stm:dql:properties:bounds_ordered}
		Assume that \cref{asm:dql:always_collapse_ec,asm:dql:sampled_value_close_to_real_value} hold.
		Then, we have $\lowerbound_\algostep(a) \leq \val_\algoepisode(a) \leq \upperbound_\algostep(a)$ for all episodes $\algoepisode$, steps $\algostep \geq \algostep_\algoepisode$, and actions $a \in \Actions_\algoepisode$.

		\item \label{item:stm:dql:properties:strategy_bounds}
		We have for every step $\algostep$ in episode $\algoepisode$ and state $s \in \States_\algoepisode$ that
		\begin{equation*}
			\ExpectedSumStrat{\strategy_\algostep}{\upperbound_\algostep}(s) = \upperbound_\algostep(s) \text{\quad and \quad} \ExpectedSumStrat{\strategy_\algostep}{\lowerbound_\algostep}(s) \leq \lowerbound_\algostep(s).
		\end{equation*}

		\item \label{item:stm:dql:properties:update_succeeds}
		If $a \notin \ConvergedUpperBounds<\algostep>$, then $a \notin \ConvergedUpperBounds<\algostep'>$ for all $\algostep' \geq \algostep$ until an $\upperbound$-update of action $a$ succeeds or the upper bound is set to $0$ in \cref{alg:dql:line:set_upper_zero}.

		\item \label{item:stm:dql:properties:converged_successful_update_probability}
		The probability that \cref{asm:dql:converged_successful_update} is violated during the execution of \cref{alg:dql_no_ec} is bounded by $\frac{\delta}{4}$.

		\item \label{item:stm:dql:properties:unsuccessful_update_means_converged}
		Assume that \cref{asm:dql:converged_successful_update} holds.
		If an attempted $\upperbound$-update  of action $a$ at step $\algostep$ fails and $\learn^\upperbound_{\algostep+1}(a) = \lfalse$, then $a \in \ConvergedUpperBounds<\algostep+1>$.
		Once no more updates of $\upperbound$ succeed, the analogous statement holds true for the lower bounds.

		\item \label{item:stm:dql:properties:non_converged_visit_count}
		Assume that \cref{asm:dql:converged_successful_update} holds.
		Then, there are at most $2 \delay \cdot \frac{\cardinality{\Actions}}{\updatestep}$ visits to state-action pairs which are not $\upperbound$-converged.
		Once the upper bounds are not updated any more, there are at most $2 \delay \cdot \frac{\cardinality{\Actions}}{\updatestep}$ visits to state-action pairs which are not $\lowerbound$-converged.
	\end{enumerate}
\end{lemma}

\begin{proof}
	\Cref{item:stm:dql:properties:success_updates,item:stm:dql:properties:attempted_updates} follow directly as in \cref{stm:dql_no_ec:successful_update_count,stm:dql_no_ec:attempted_update_count}.
	The only additional observation is that the algorithm never adds new actions and that the changes to the bounds outside of \cref{alg:dql:line:update_bounds} never reset the progress of an action's bounds.

	\cref{item:stm:dql:properties:sampled_value_close_to_real_value_probability} can be proven completely analogous to \cref{stm:dql_no_ec:sampled_value_close_to_real_value_probability}, since this proof only relies on the Markov property of the successor sampling.
	We need to adjust the definition of $Y_i$ slightly to incorporate the modifications of the algorithm.
	Let thus $s'_{k_i} \in \States$ denote the states obtained by the successor oracle in \cref{alg:dql:line:successor_oracle}.
	By \cref{stm:dql:values_collapse} we have that $\val_\algoepisode(\algorepresentative_{\algoepisode_i}(s'_{k_i})) = \val_\algoepisode(s''_{k_i})$, and thus $Y_i = \val_\algoepisode(\algorepresentative_{\algoepisode_i}(s'_{k_i}))$ are i.i.d.

	For \cref{item:stm:dql:properties:bounds_ordered}, we first show that all newly introduced updates of $\upperbound$ and $\lowerbound$ are correct.
	Using \cref{asm:dql:always_collapse_ec}, we prove the two special cases.
	The algorithm sets $\upperbound_\algostep(a) \gets 0$ if an EC $(R, B)$ without outgoing transitions and no target state is identified.
	In this case, we clearly have that $\val_\algoepisode(a) = 0$ for all $s \in R$.
	Similarly, setting $\lowerbound_\algoepisode(a) \gets 1$ when any state in the EC $(R, B)$ is an accepting state is correct, since clearly $\val_\algoepisode(a) = 1$ for all $s \in R$, $a \in \stateactions_\algoepisode \intersection B$.
	Due to \cref{stm:dql:values_collapse}, copying the respective bounds to the representative state $s_{(R, B)}$ (which happens implicitly in \cref{alg:dql:line:set_representative_actions}) is correct, too.
	Now, the reasoning of \cref{stm:dql_no_ec:bounds_ordered} applies.

	\Cref{item:stm:dql:properties:strategy_bounds,item:stm:dql:properties:update_succeeds} can be proven as in \cref{stm:dql_no_ec:nonconverged_bounds_monotonic}.

	\Cref{item:stm:dql:properties:converged_successful_update_probability} is proven analogous to \cref{item:stm:dql:properties:sampled_value_close_to_real_value_probability}, following the proof of \cref{stm:dql_no_ec:converged_successful_update_probability}.
	Again, this claim only depends on the sampled successors.
	We define $X_i = \ExpectedSumStrat{\strategy_{k_1}}{\upperbound_{k_1}}(\algorepresentative_{\algoepisode_1}(s_{k_i}''))$.
	Since we do not modify the underlying transition probabilities, from which $s_{k_i}''$ is obtained, these $X_i$ are i.i.d.\ again and we can apply the same reasoning.
	To conclude the proof as before, we need to employ \cref{stm:dql:bounds_monotone}.
	Note that since we only speak about the actual computed bounds $\upperbound$ and $\lowerbound$, we do not need to employ \cref{stm:dql:values_collapse}.

	\cref{item:stm:dql:properties:unsuccessful_update_means_converged} follows directly as in \cref{stm:dql_no_ec:unsuccessful_update_means_converged}.
	Similarly, \cref{item:stm:dql:properties:non_converged_visit_count} follows as in \cref{stm:dql_no_ec:non_converged_visit_count}, using \cref{item:stm:dql:properties:success_updates} instead of \cref{stm:dql_no_ec:successful_update_count}. 
\end{proof}

In the proof of correctness for the no-EC DQL algorithm, we applied \cref{stm:assignment_mdp_value} directly on the MDP to obtain bounds on the reachability of $\targetstate$ based on the values of $\upperbound$ and $\lowerbound$ in \cref{stm:dql_no_ec:bounds_close}.
Now, we cannot apply this lemma directly on either $\MDP$ or $\MDP_\algoepisode$ since both may contain ECs.
Hence, we apply the lemma on an MDP derived from $\MDP_\algoepisode$ to obtain a similar result.
Let us thus first define the set of all actions in \enquote{non-final} ECs as
\begin{equation*}
	E_\algoepisode = {\Union}_{\{(R, B) \in \Ecs(\MDP_\algoepisode) \mid R \intersection (\targetset_\algoepisode \union \algozerostates_\algoepisode) = \emptyset\}} B.
\end{equation*}

\begin{lemma} \label{stm:dql:upper_bound_reachability}
	Assume that \cref{asm:dql:always_collapse_ec,asm:dql:sampled_value_close_to_real_value} hold and fix an episode $\algoepisode$.
	Then, we have for every state $s \in \States_\algoepisode$
	\begin{equation*}
		\upperbound_{\algoepisode}(s) - 3 \updatestep \cdot \cardinality{\States} p_{\min}^{-\cardinality{\States}} - \ProbabilityMDP<\MDP_\algoepisode, s><\strategy_\algoepisode>[\reach \overline{\ConvergedUpperBounds<\algoepisode>}] - \ProbabilityMDP<\MDP_\algoepisode, s><\strategy_\algoepisode>[\reach E_\algoepisode] \leq \ProbabilityMDP<\MDP_\algoepisode, s><\strategy_\algoepisode>[\reach \targetset_\algoepisode].
	\end{equation*}
\end{lemma}

\begin{proof}
	We first want to derive an MDP from $\MDP_\algoepisode$ without any ECs but still capturing its behaviour.
	For this, recall that there are two kinds of ECs in $\MDP_\algoepisode$.
	Firstly, there are ECs which correspond to ECs in the original $\MDP$.
	Secondly, we get a self-loop EC for each identified target- or zero-state, i.e.\ states in $\targetset_\algoepisode$ or $\algozerostates_\algoepisode$.
	We define the derived MDP $\MDP'_\algoepisode = (\States_\algoepisode \union \{\targetstate, \sinkstate\}, \Actions_\algoepisode \union \{a_+, a_-\}, \mdptransitions_\algoepisode', \stateactions_\algoepisode')$, where
	\begin{align*}
		\mdptransitions_\algoepisode'(s_\circ, a_\circ) & = \{s_\circ \mapsto 1\} \quad \text{for $\circ \in \{+, -\}$} \\
		\mdptransitions_\algoepisode'(s, a) & = \{\targetstate \mapsto 1\} \quad \text{for all $s \in \targetset_\algoepisode$, $a \in \stateactions_\algoepisode(s)$,} \\
		\mdptransitions_\algoepisode'(s, a) & = \{\sinkstate \mapsto 1\} \quad \text{for all $s \in \algozerostates_\algoepisode$, $a \in \stateactions_\algoepisode(s)$,} \\
		\mdptransitions_\algoepisode'(s, a) & = \{\targetstate \mapsto 1\} \quad \text{for all $a \in E$, $s = \actionstate<\MDP_\algoepisode>(a)$,} \\
		\mdptransitions_\algoepisode'(s, a) & = \mdptransitions_\algoepisode(s, a) \quad \text{for all other $s \in \States_\algoepisode$, $a \in \stateactions_\algoepisode(s)$},
	\end{align*}
	and $\stateactions_\algoepisode'(s) = \stateactions_\algoepisode(s)$ for $s \in \States_\algoepisode$ and $\stateactions_\algoepisode'(s_\circ) = \{a_\circ\}$ for $\circ \in \{+, -\}$.
	In essence, $\MDP'_\algoepisode$ equals $\MDP_\algoepisode$ except that we (i)~added the special states $\targetstate$ and $\sinkstate$, (ii)~all states in $\targetset_\algoepisode$ and $\algozerostates_\algoepisode$ move to $\targetstate$ and $\sinkstate$, respectively, and (iii)~all actions in ECs outside of $\targetset_\algoepisode$ and $\algozerostates_\algoepisode$ move to $\targetstate$, in the spirit of \cref{stm:dql:ec_bounds_one_or_zero}.

	Clearly, $\MDP'_\algoepisode$ has no ECs except the special states $\targetstate$ and $\sinkstate$ and thus satisfies \cref{asm:mec_free}.
	Moreover, the probability of reaching $\targetstate$ in $\MDP'_\algoepisode$ equals the probability of reaching $\targetset_\algoepisode \union E_\algoepisode$ in $\MDP_\algoepisode$ by construction of $\MDP'_\algoepisode$ \plabel{proof:dql:upper_bound:reach_probability_correspondence}.

	Now, we extend $\strategy_\algoepisode$ to select action $a_\circ$ in the special state $s_\circ$ to obtain $\strategy'_\algoepisode$.
	Furthermore, we set $X(s, a) = \upperbound_{\algoepisode}(a)$ for all states $s \in \States_\algoepisode$, $a \in \stateactions_\algoepisode(s)$, $X(\targetstate, a_+) = 1$, and $X(\sinkstate, a_-) = 0$.
	We apply \cref{stm:assignment_mdp_value} with $\MDP = \MDP'_\algoepisode$, $\strategy = \strategy'_\algoepisode$, $\kappa_l = -1$, and $\kappa_u = 3 \updatestep$.
	As a result, for each state $s \in \States_\algoepisode$ we have
	\begin{equation*}
		\ExpectedSumStrat{\strategy'_\algoepisode}{X}(s) - \ProbabilityMDP<\MDP', s><\strategy'_\algoepisode>[\reach \{\targetstate\}] \leq 3 \updatestep \cdot \cardinality{\States} p_{\min}^{-\cardinality{\States}},
	\end{equation*}
	where $\MDP'$ is the MDP defined in the lemma.
	Observe that for $s \in \States_\algoepisode$ \plabel{proof:dql:upper_bound:expected_sum_states}
	\begin{equation*}
		\ExpectedSumStrat{\strategy'_\algoepisode}{X}(s) = {\sum}_{a \in \stateactions'_\algoepisode(s)} \strategy'_\algoepisode(s, a) \cdot X(s, a) = {\sum}_{a \in \stateactions_\algoepisode(s)} \strategy_\algoepisode(s, a) \cdot \upperbound_\algoepisode(a) = \ExpectedSumStrat{\strategy_\algoepisode}{\upperbound_\algoepisode}(s).
	\end{equation*}
	To analyse how $\MDP'$ and $\MDP'_\algoepisode$ are related, we first need to derive the structure of $\mathcal{K}$ from the lemma.
	Thus, we now prove that $\mathcal{K} = \ConvergedUpperBounds<\algoepisode> \union \{a_+, a_-\}$.
	Recall that
	\begin{equation*}
		\mathcal{K} = \{a \in \Actions_\algoepisode \union \{a_+, a_-\} \mid X(s, a) - \ExpectedSumMDP{\mdptransitions_\algoepisode'}{s}{a}{\strategy'_\algoepisode[X]} \leq 3 \updatestep\}
	\end{equation*}
	and
	\begin{equation*}
		\ExpectedSumMDP{\mdptransitions_\algoepisode'}{s}{a}{\strategy'_\algoepisode[X]} = {\sum}_{s' \in \States_\algoepisode \union \{\targetstate, \sinkstate\}} \mdptransitions_\algoepisode'(s, a, s') \cdot {\sum}_{a' \in \stateactions_\algoepisode'(s')} \strategy(s', a') \cdot X(s', a').
	\end{equation*}
	Clearly, $a_+$ and $a_-$ satisfy the requirements due to their self-loop.
	Furthermore, we have $\ExpectedSumStrat{\strategy'_\algoepisode}{X}(\targetstate) = 1$, $\ExpectedSumStrat{\strategy'_\algoepisode}{X}(\sinkstate) = 0$ \plabel{proof:dql:upper_bound:expected_sum_special}.
	Now, let $a \in \Actions_\algoepisode$ and $s \in \States_\algoepisode$ the corresponding state.
	By definition, we have $X(s, a) = \upperbound_{\algoepisode}(a)$, hence we need to show that $\ExpectedSumMDP{\mdptransitions_\algoepisode'}{s}{a}{\ExpectedSumStrat{\strategy'_\algoepisode}{X}} = \ExpectedSumMDP{\mdptransitions_\algoepisode}{s}{a}{\ExpectedSumStrat{\strategy_\algoepisode}{\upperbound_{\algoepisode}}}$.
	We proceed with a case distinction.
	\begin{itemize}
		\item
		$s \in \targetset_\algoepisode \union \algozerostates_\algoepisode$: By definition of the algorithm, we have $\upperbound_\algoepisode(s) = 1$ or $0$, respectively.
		The unique successor under any action $a \in \stateactions_\algoepisode(s)$ in $\MDP_\algoepisode$ equals $s$ by definition, thus $\ExpectedSumMDP{\mdptransitions_\algoepisode}{s}{a}{\ExpectedSumStrat{\strategy_\algoepisode}{\upperbound_{\algoepisode}}} = \upperbound_\algoepisode(s)$.
		In $\MDP'_\algoepisode$, the unique successor equals $\targetstate$ or $\sinkstate$, respectively.
		Thus, with \ref{proof:dql:upper_bound:expected_sum_special}, we have $\ExpectedSumStrat{\strategy'_\algoepisode}{X}(s) = \ExpectedSumStrat{\strategy_\algoepisode}{\upperbound_{\algoepisode}}(s)$.
		The claim follows.

		\item
		$a \in E$: Note that this case implies that $s \notin \targetset_\algoepisode \union \algozerostates_\algoepisode$.
		Due to \cref{stm:dql:ec_bounds_one_or_zero}, we have that $\upperbound_\algoepisode(a) = 1$ for all such actions.
		Recall that $\strategy_\algoepisode$ follows actions maximizing $\upperbound_\algoepisode$.
		Consequently, $\ExpectedSumStrat{\strategy_\algoepisode}{\upperbound_{\algoepisode}}(s') = \ExpectedSumStrat{\strategy'_\algoepisode}{X}(s') = \upperbound_{\algoepisode}(s') = 1$ for all states $s'$ inside an non-trivial EC of $\MDP_\algoepisode$.
		Thus, we also have $\ExpectedSumMDP{\mdptransitions_\algoepisode}{s}{a}{\ExpectedSumStrat{\strategy_\algoepisode}{\upperbound_{\algoepisode}}} = 1$.
		From the definition of $\MDP'_\algoepisode$ and \ref{proof:dql:upper_bound:expected_sum_special}, we directly get $\ExpectedSumMDP{\mdptransitions_\algoepisode'}{s}{a}{\ExpectedSumStrat{\strategy'_\algoepisode}{X}} = 1$.

		\item
		$s \notin \targetset_\algoepisode \union \algozerostates_\algoepisode$, $a \notin E$: By definition, we have $\mdptransitions_\algoepisode(s, a) = \mdptransitions_\algoepisode'(s, a)$.
		Together with \ref{proof:dql:upper_bound:expected_sum_states} and \ref{proof:dql:upper_bound:expected_sum_special}, the statement follows.
	\end{itemize}

	Recall that $\MDP'$ is defined as $\MDP'_\algoepisode$ except that $\mdptransitions'(s, a) = \{\targetstate \mapsto X(s, a), \sinkstate \mapsto 1 - X(s, a)\}$ for all $a \notin \mathcal{K}$.
	Hence, as in \cref{stm:dql_no_ec:bounds_close}, we get that for all states $s \in \States_\algoepisode$
	\begin{equation*}
		\ProbabilityMDP<\MDP', s><\strategy'_\algoepisode>[\reach \{\targetstate\}] - \ProbabilityMDP<\MDP'_\algoepisode, s><\strategy'_\algoepisode>[\reach \overline{\ConvergedUpperBounds<\algoepisode>}] \leq \ProbabilityMDP<\MDP'_\algoepisode, s><\strategy'_\algoepisode>[\reach \{\targetstate\}],
	\end{equation*}
	and thus with \ref{proof:dql:upper_bound:reach_probability_correspondence} we get \plabel{proof:dql:upper_bound:probability_bound}
	\begin{equation*}
		\ExpectedSumStrat{\strategy'_\algoepisode}{X}(s) - 3 \updatestep \cdot \cardinality{\States} p_{\min}^{-\cardinality{\States}} - \ProbabilityMDP<\MDP'_\algoepisode, s><\strategy'_\algoepisode>[\reach \overline{\ConvergedUpperBounds<\algoepisode>}] \leq \ProbabilityMDP<\MDP_\algoepisode, s><\strategy_\algoepisode>[\reach (\targetset_\algoepisode \union E_\algoepisode)].
	\end{equation*}
	Further, we have $\ExpectedSumStrat{\strategy'_\algoepisode}{X}(s) = \ExpectedSumStrat{\strategy_\algoepisode}{\upperbound_\algoepisode}(s) = \upperbound_\algoepisode(s)$ by \cref{stm:dql:properties}, \cref{item:stm:dql:properties:strategy_bounds} \plabel{proof:dql:upper_bound:strategy_correspondence}.

	To conclude the proof, we show that $\ProbabilityMDP<\MDP'_\algoepisode, s><\strategy'_\algoepisode>[\reach \overline{\ConvergedUpperBounds<\algoepisode>}] \leq \ProbabilityMDP<\MDP_\algoepisode, s><\strategy_\algoepisode>[\reach \overline{\ConvergedUpperBounds<\algoepisode>}]$ \plabel{proof:dql:upper_bound:K_reach_correspondence}.
	To this end, observe that (i)~for each state $s \in \States_\algoepisode$ and action $a \in \stateactions_\algoepisode(s)$ we either have $\mdptransitions_\algoepisode(s, a) = \mdptransitions'_\algoepisode(s, a)$ or $\support \mdptransitions'_\algoepisode(s, a) \subseteq \{\targetstate, \sinkstate\}$ and (ii)~the added states $\targetstate$ and $\sinkstate$ are absorbing.
	Thus, each run reaching $\overline{\ConvergedUpperBounds<\algoepisode>}$ in $\MDP'_\algoepisode$ has a corresponding, equally probable path in $\MDP_\algoepisode$.

	The overall claim follows from the above equations and a union bound.
	\begin{align*}
		& \upperbound_{\algoepisode}(s) - 3 \updatestep \cdot \cardinality{\States} p_{\min}^{-\cardinality{\States}} - \ProbabilityMDP<\MDP_\algoepisode, s><\strategy_\algoepisode>[\reach \overline{\ConvergedUpperBounds<\algoepisode>}] \\
		& \qquad \overset{\mathclap{\ref{proof:dql:upper_bound:strategy_correspondence}}}{~=~} \ExpectedSumStrat{\strategy'_\algoepisode}{X}(s) - 3 \updatestep \cdot \cardinality{\States} p_{\min}^{-\cardinality{\States}} - \ProbabilityMDP<\MDP_\algoepisode, s><\strategy_\algoepisode>[\reach \overline{\ConvergedUpperBounds<\algoepisode>}] \\
		& \qquad \overset{\mathclap{\ref{proof:dql:upper_bound:K_reach_correspondence}}}{~\leq~} \ExpectedSumStrat{\strategy'_\algoepisode}{X}(s) - 3 \updatestep \cdot \cardinality{\States} p_{\min}^{-\cardinality{\States}} - \ProbabilityMDP<\MDP'_\algoepisode, s><\strategy'_\algoepisode>[\reach \overline{\ConvergedUpperBounds<\algoepisode>}] \\
		& \qquad \overset{\mathclap{\ref{proof:dql:upper_bound:probability_bound}}}{~\leq~} \ProbabilityMDP<\MDP_\algoepisode, s><\strategy_\algoepisode>[\reach (\targetset_\algoepisode \union E_\algoepisode)]. \qedhere
	\end{align*}
\end{proof}

\begin{lemma} \label{stm:dql:lower_bound_reachability}
	Assume that \cref{asm:dql:always_collapse_ec,asm:dql:sampled_value_close_to_real_value} hold and fix an episode $\algoepisode$.
	Then, we have for every state $s \in \States_\algoepisode$
	\begin{equation*}
		\ProbabilityMDP<\MDP_\algoepisode, s><\strategy_\algoepisode>[\reach \targetset_\algoepisode] \leq \lowerbound_\algoepisode(s) + 3 \updatestep \cdot \cardinality{\States} p_{\min}^{-\cardinality{\States}} + \ProbabilityMDP<\MDP_\algoepisode, s><\strategy_\algoepisode>[\reach \overline{\ConvergedLowerBounds<\algoepisode>}] + \ProbabilityMDP<\MDP_\algoepisode, s><\strategy_\algoepisode>[\reach E_\algoepisode].
	\end{equation*}
\end{lemma}

\begin{proof}
	As in \cref{stm:dql:upper_bound_reachability}, we construct a second MDP without ECs, but slightly modify the transition function.
	In particular, let $\MDP'_\algoepisode = (\States_\algoepisode \union \{\targetstate, \sinkstate\}, \Actions_\algoepisode \union \{a_+, a_-\}, \mdptransitions_\algoepisode', \stateactions_\algoepisode')$ be defined as before.
	However, for $a \in E_\algoepisode$ and $s = \actionstate<\MDP_\algoepisode>(a)$, we define
	\begin{equation*}
		\mdptransitions_\algoepisode'(s, a) = \{\targetstate \mapsto \ExpectedSumMDP{\mdptransitions_\algoepisode}{s}{a}{\ExpectedSumStrat{\strategy_\algoepisode}{\lowerbound_\algoepisode}}, \sinkstate \mapsto 1 - \ExpectedSumMDP{\mdptransitions_\algoepisode}{s}{a}{\ExpectedSumStrat{\strategy_\algoepisode}{\lowerbound_\algoepisode}}\}.
	\end{equation*}
	Again, $\MDP'_\algoepisode$ has no ECs except in the two special states and thus \cref{stm:assignment_mdp_value} is applicable.
	We set $X(s, a) = \lowerbound_{\algoepisode}(a)$ for all states $s \in \States_\algoepisode$, $X(\targetstate, a_+) = 1$, and $X(\sinkstate, a_-) = 0$.
	As above, we have that $\ExpectedSumStrat{\strategy'_\algoepisode}{X}(s) = \ExpectedSumStrat{\strategy_\algoepisode}{\lowerbound_{\algoepisode}}(s)$ for all $s \in \States_\algoepisode$.
	We apply the lemma with $\MDP = \MDP'_\algoepisode$, $\strategy = \strategy'_\algoepisode$, $\kappa_l = -3 \updatestep$, and $\kappa_u = 1$.
	Thus, for each state $s \in \States_\algoepisode$
	\begin{equation*}
		\ProbabilityMDP<\MDP', s><\strategy'_\algoepisode>[\reach \{\targetstate\}] - \ExpectedSumStrat{\strategy'_\algoepisode}{X}(s) \leq 3 \updatestep \cdot \cardinality{\States} p_{\min}^{-\cardinality{\States}},
	\end{equation*}
	where $\MDP'$ is the MDP defined in the lemma.
	We again show that $\mathcal{K} = \ConvergedLowerBounds<\algoepisode> \union \{a_+, a_-\}$ by case distinction as follows:
	\begin{itemize}
		\item
		Trivially, $a_+, a_- \in \mathcal{K}$, $\ExpectedSumStrat{\strategy'_\algoepisode}{X}(\targetstate) = 1$, and $\ExpectedSumStrat{\strategy'_\algoepisode}{X}(\sinkstate) = 0$.

		\item
		$s \in \targetset_\algoepisode \union \algozerostates_\algoepisode$: The claims follow by an analogous argument.
		Recall that for these states we have $\upperbound_\algoepisode(a) = \lowerbound_\algoepisode(a)$ for all $a \in \stateactions_\algoepisode(s)$.

		\item
		$a \in E$: Inserting the definitions, we get
		\begin{align*}
			& \ExpectedSumMDP{\mdptransitions'_\algoepisode}{s}{a}{\ExpectedSumStrat{\strategy'_\algoepisode}{X}} = \mdptransitions'_\algoepisode(s, a, \targetstate) \cdot \ExpectedSumStrat{\strategy'_\algoepisode}{X}(\targetstate) + \mdptransitions'_\algoepisode(s, a, \sinkstate) \cdot \ExpectedSumStrat{\strategy'_\algoepisode}{X}(\sinkstate) \\
			& \qquad = \ExpectedSumMDP{\mdptransitions_\algoepisode}{s}{a}{\ExpectedSumStrat{\strategy_\algoepisode}{\lowerbound_\algoepisode}} \cdot 1 + (1 - \ExpectedSumMDP{\mdptransitions_\algoepisode}{s}{a}{\ExpectedSumStrat{\strategy_\algoepisode}{\lowerbound_\algoepisode}}) \cdot 0 \\
			& \qquad = \ExpectedSumMDP{\mdptransitions_\algoepisode}{s}{a}{\ExpectedSumStrat{\strategy_\algoepisode}{\lowerbound_\algoepisode}}.
		\end{align*}

		\item
		$s \notin \targetset_\algoepisode \union \algozerostates_\algoepisode$, $a \notin E$: Follows analogously.
	\end{itemize}

	As in \cref{stm:dql_no_ec:bounds_close}, we also get for all states $s \in \States_\algoepisode$ that
	\begin{equation*}
		\ProbabilityMDP<\MDP'_\algoepisode, s><\strategy'_\algoepisode>[\reach \{\targetstate\}] \leq \ProbabilityMDP<\MDP', s><\strategy'_\algoepisode>[\reach \{\targetstate\}] + \ProbabilityMDP<\MDP'_\algoepisode, s><\strategy'_\algoepisode>[\reach \overline{\ConvergedLowerBounds<\algoepisode>}].
	\end{equation*}
	Similar to the above proof, we have $\ExpectedSumStrat{\strategy'_\algoepisode}{X}(s) = \ExpectedSumStrat{\strategy_\algoepisode}{\lowerbound_\algoepisode}(s) \leq \lowerbound_\algoepisode(s)$ by \cref{stm:dql:properties}, \cref{item:stm:dql:properties:strategy_bounds}.
	With completely analogous reasoning, we can show that $\ProbabilityMDP<\MDP'_\algoepisode, s><\strategy'_\algoepisode>[\reach \overline{\ConvergedLowerBounds<\algoepisode>}] \leq \ProbabilityMDP<\MDP_\algoepisode, s><\strategy_\algoepisode>[\reach \overline{\ConvergedLowerBounds<\algoepisode>}]$.
	Putting all equations together, we get that
	\begin{equation*}
		\ProbabilityMDP<\MDP'_\algoepisode, s><\strategy'_\algoepisode>[\reach \{\targetstate\}] \leq \lowerbound_\algoepisode(s) + 3 \updatestep \cdot \cardinality{\States} p_{\min}^{-\cardinality{\States}} + \ProbabilityMDP<\MDP_\algoepisode, s><\strategy_\algoepisode>[\reach \overline{\ConvergedLowerBounds<\algoepisode>}].
	\end{equation*}
	Now, it remains to show that $\ProbabilityMDP<\MDP_\algoepisode, s><\strategy_\algoepisode>[\reach \targetset_\algoepisode] - \ProbabilityMDP<\MDP_\algoepisode, s><\strategy_\algoepisode>[\reach E_\algoepisode] \leq \ProbabilityMDP<\MDP'_\algoepisode, s><\strategy'_\algoepisode>[\reach \{\targetstate\}]$.
	This claim follows with the same reasoning as before, since we have that $\mdptransitions_\algoepisode(s, a) = \mdptransitions'_\algoepisode(s, a)$ for $a \notin E_\algoepisode, s = \actionstate<\MDP_\algoepisode>(a)$.
	Thus, every path in $\MDP_\algoepisode$ which does not visit $E$ has a corresponding, equally probable path in $\MDP_\algoepisode'$.
	The overall claim follows. 
\end{proof}

\begin{theorem} \label{stm:dql_correct}
	\cref{alg:dql} terminates and yields a correct result with probability at least $1 - \delta$ after at most $\mathcal{O}(\mathrm{POLY}(\cardinality{\Actions}, p_{\min}^{-\cardinality{\States}}, \varepsilon^{-1}, \ln \delta))$ steps.
\end{theorem}

\begin{proof}
	This proof is largely analogous to the proof of \cref{stm:dql_no_ec_correct}, and we shorten some of its parts.
	Again, we only consider executions where \cref{asm:dql:converged_successful_update,asm:dql:sampled_value_close_to_real_value,asm:dql:always_collapse_ec} hold.
	By \cref{stm:dql:always_collapse_ec_probability,stm:dql:properties}, \cref{item:stm:dql:properties:converged_successful_update_probability,item:stm:dql:properties:sampled_value_close_to_real_value_probability} together with the union bound, this happens with probability at least $1 - \delta$.
	Correctness of the result upon termination follows from \cref{stm:dql:properties}, \cref{item:stm:dql:properties:bounds_ordered}.

	We show by contradiction that the algorithm terminates for almost all considered executions.
	Thus, assume that the execution does not halt with non-zero probability.
	By \cref{stm:dql:terminate_or_infinite_episodes}, all of these executions experience an infinite number of episodes.

	Due to \cref{stm:dql:properties}, \cref{item:stm:dql:properties:attempted_updates}, there are only finitely many attempted updates on all considered executions and the algorithm eventually does not change $\upperbound$, since no successful updates can occur from some step $\algostep$ onwards.
	Similarly, there are only finitely many EC collapses due to \cref{stm:dql:finite_representative_states}, and eventually the sampling MDP $\MDP_\algoepisode$ stabilizes.
	This means that all following samples are obtained by sampling according to the strategy $\strategy_\algostep$ on the MDP $\MDP_\algoepisode$.
	Again, we employ \cref{stm:markov_process_repeating} to continue the proof and we get $\ProbabilityMDP<\MDP_\algoepisode, \initialstate><\strategy_\algostep>[\reach \overline{\ConvergedUpperBounds<\algostep>}] = 0$ and $\ProbabilityMDP<\MDP_\algoepisode, \initialstate><\strategy_\algostep>[\reach \overline{\ConvergedLowerBounds<\algostep>}] = 0$ on almost all considered executions.
	By an analogous argument, we can show that $\ProbabilityMDP<\MDP_\algoepisode, \initialstate><\strategy_\algostep>[\reach E_\algoepisode] = 0$, since otherwise by \cref{stm:dql:appear_is_ec_probability} (with $\targetset = \targetset_\algoepisode \union \algozerostates_\algoepisode$) we have a non-zero probability of detecting a new EC, contradicting our assumption.

	Thus, by applying \cref{stm:dql:upper_bound_reachability}
	\begin{equation*}
		\ProbabilityMDP<\MDP_\algoepisode, \initialstate><\strategy_\algoepisode>[\reach \targetset_\algoepisode] \geq \upperbound_{\algoepisode}(\initialstate) - 3 \updatestep \cdot \cardinality{\States} p_{\min}^{-\cardinality{\States}} - \ProbabilityMDP<\MDP_\algoepisode, \initialstate><\strategy_\algoepisode>[\reach \overline{\ConvergedUpperBounds<\algoepisode>}] - \ProbabilityMDP<\MDP_\algoepisode, \initialstate><\strategy_\algoepisode>[\reach E_\algoepisode] > \upperbound_{\algoepisode}(\initialstate) - \frac{\varepsilon}{2}.
	\end{equation*}
	Dually, with \cref{stm:dql:lower_bound_reachability} we get
	\begin{equation*}
		\ProbabilityMDP<\MDP_\algoepisode, \initialstate><\strategy_\algoepisode>[\reach \targetset_\algoepisode] \leq \lowerbound_\algoepisode(\initialstate) + 3 \updatestep \cdot \cardinality{\States} p_{\min}^{-\cardinality{\States}} + \ProbabilityMDP<\MDP_\algoepisode, \initialstate><\strategy_\algoepisode>[\reach \overline{\ConvergedLowerBounds<\algoepisode>}] + \ProbabilityMDP<\MDP_\algoepisode, \initialstate><\strategy_\algoepisode>[\reach E_\algoepisode] < \lowerbound_{\algoepisode}(\initialstate) + \frac{\varepsilon}{2}.
	\end{equation*}
	Together, $\upperbound_\algoepisode(\initialstate) - \lowerbound_\algoepisode(\initialstate) < \varepsilon$, contradicting the assumption.

	For the step bound, we can mostly replicate the idea of the DQL variant without ECs.
	In particular, we can bound the number of paths by the same argument:
	The probability of reaching a non-$\upperbound$- / non-$\lowerbound$-converged action within $\cardinality{\States}$ steps is at least $p_{\min}^{\cardinality{\States}}$ (or 0).
	By \cref{stm:dql:properties}, \cref{item:stm:dql:properties:non_converged_visit_count} we again get that the number of visits to such actions is bounded.
	Since $\algolimit \geq \cardinality{\Actions} \geq \cardinality{\States}$ and thus the sampling is not stopped early due to that condition, we again can bound the maximal number of paths by the same $n$.
	For the length of the paths, observe that they are bounded by $2 \algolimit^3$ by construction of the algorithm.
	From the definition of $\algolimit$ in Equation~\eqref{eq:dql:choice_of_limit}, we see that this bound is polynomial, too, by considering the Taylor expansion of the exponential. 
\end{proof}

\begin{remark}
	To conclude, we briefly outline extensions to other objectives.

	For safety, i.e.\ maximizing the probability of remaining inside a given set of states forever (or, equivalently, minimizing reachability of unsafe states), we only need to change the treatment of end components slightly.
	Assume w.l.o.g.\ that any unsafe state is collapsed into one sink state $\sinkstate$ (e.g.\ by testing for every encountered state whether it is safe and, if not, replace it by $\sinkstate$).
	Then, whenever we identify an end component, we know that this end component does not contain a sink state but rather only comprises safe states.
	However, this actually is exactly what we are looking for: a possibility of staying safe forever.
	Thus, we assign a value of $1$ to all actions in such an EC.
	And indeed, by \cref{stm:mdp_almost_sure_absorption}, we know that ECs are the \emph{only} place that allow us to stay safe forever.
	Together, we can derive the desired result.

	Extending to total reward has two major hurdles.
	Firstly, the total reward can be infinite, and we would first need to identify whether this is the case.
	To this end, we need to identify \emph{all} end components in the system and check for each that it yields zero reward.
	Here, we would need to employ graph-based reasoning akin to \cite{DBLP:conf/cav/AshokKW19}, as we need to ensure that we have not missed any transition.
	Once this is established (or guaranteed due to domain knowledge), we can derive an upper bound on the total reward if we are given an upper bound on the reward that can be obtained in one step $r_{\max}$.
	This bound is in the order $\mathcal{O}(p_{\min}^{-\cardinality{\States}} \cdot r_{\max})$.
	Using this bound as initial value for the upper bound then would lead to a correct algorithm.
	(See also \cite[Appendix~B]{DBLP:conf/cav/MeggendorferW24} and \cite[Section~4]{DBLP:journals/fmsd/ChenFKPS13} for related discussions.)

	Finally, an extension to mean payoff (aka.\ long run average reward) or general $\omega$-regular objectives in a model-free setting seems to be rather unlikely.
	Both inherently are infinite horizon objectives, while sampling only ever gives us finite information.
	As such, we likely need to use graph-based reasoning to reach meaningful conclusions.
	In particular, for $\omega$-regular objectives, we would need to know at least the graph structure of identified end components to decide whether they are winning or not, and for mean payoff we even would need bounds on the transition probabilities.
	As a special case, models where each end component is guaranteed to only comprise a single state could be tractable.
\end{remark}
 	\section{Conclusion and Future Work} \label{sec:conclusion}

In this work, we improved and extended the ideas of \cite{DBLP:conf/atva/BrazdilCCFKKPU14}, fixing several imprecisions and issues of the proofs.
This results in a framework for verifying MDP, using learning algorithms.
Building upon exiting methods, we thus provide novel techniques to analyse infinite-horizon reachability properties of arbitrary MDPs, yielding either exact bounds in the white-box scenario or probabilistically correct bounds in the black-box scenario.
Moreover, we presented a generalization of the methods of \cite{DBLP:conf/atva/BrazdilCCFKKPU14}, allowing for further, more sophisticated applications.

We deliberately omit an experimental evaluation.
Since the inception of the presented idea, multiple tools have implemented variants and extensions thereof for several objectives and model classes.
In particular, we want to point to the tool PET \cite{DBLP:conf/atva/Meggendorfer22,DBLP:conf/cav/MeggendorferW24}, which implements and evaluates the general complete information algorithm and presents a detailed evaluation.
Moreover, as already mentioned, for DQL the associated constants are infeasible for practical application:
Already for an MDP with $10$ States, $20$ actions and $p_{\min} = 0.1$, we obtain $\delay \approx 10^{26}$ for $\varepsilon = 0.1$ and $\delta = 0.01$.

Given this framework, an interesting direction for future work would be to extend this approach with more sophisticated learning algorithms.
Another, orthogonal direction is to explore whether our approach can be combined with symbolic methods.

	\printbibliography
	\appendix
	\section{Auxiliary Statements} \label{sec:appendix:auxiliary}

In this chapter we provide some general statements about Markov chains and decision processes which are used in various proofs for the DQL algorithms.

\subsubsection*{From Reachability to Step-bounded Reachability}
In this section we prove several statements relating the infinite-horizon reachability with the reachability after a sufficiently large number of steps.
\begin{lemma} \label{stm:markov_chain_minimum_reachability}
	For any Markov chain $\MC = (\States, \mctransitions)$, state $s$, and target set $\targetset$, we have that either $\ProbabilityMC<\MC, s>[\reach \targetset] = 0$ or $\ProbabilityMC<\MC, s>[\boundedreach<\cardinality{\States}> \targetset] \geq \delta_{\min}^{\cardinality{\States}}$, where $\delta_{\min}$ is the minimal transition probability, i.e.\ $\delta_{\min} = \min \{ \mctransitions(s, s') \mid s \in \States, s' \in \support \mctransitions(s) \}$.
\end{lemma}
\begin{proof}
	Fix the Markov chain $\MC$, state $s$, and target set $\targetset$ as in the lemma.
	In the first case there is nothing to prove, thus assume that $\ProbabilityMC<\MC, s>[\reach \targetset] > 0$.
	This means that there exists a finite path $\finitepath$ from $s$ to some state in $\targetset$.
	By the pigeon-hole principle, we can assume this path has length at most $\cardinality{\States}$.
	Clearly, the probability of any single transition on this path is at least $\delta_{\min}$ and thus the overall probability of this path is at least $\delta_{\min}^{\cardinality{\States}}$. 
\end{proof}

\begin{corollary} \label{stm:mdp_minimum_reachability}
	For any MDP $\MDP = (\States,  \Actions, \stateactions, \mdptransitions)$, memoryless strategy $\strategy \in \StrategiesMD<\MDP>$, state $s$, and target set $\targetset$, we have that either $\ProbabilityMDP<\MDP, s><\strategy>[\reach \targetset] = 0$ or $\ProbabilityMDP<\MDP, s><\strategy>[\boundedreach<\cardinality{\States}> \targetset] \geq \delta_{\min}(\strategy)^{\cardinality{\States}}$, where $\delta_{\min}(\strategy) = \min \{ \strategy(s, a) \cdot \mdptransitions(s, a, s') \mid s \in \States, a \in \stateactions(s), \strategy(s, a) > 0, s' \in \support \mdptransitions(s, a, s') \}$.
\end{corollary}
\begin{proof}
	Follows directly from the above lemma by applying it to $\MDP^\strategy$. 
\end{proof}
The following lemma shows that on a large enough horizon, step-bounded and unbounded reachability values coincide up to a small error, similar in spirit to \cite[Lemma~2]{DBLP:journals/ml/KearnsS02}.
\begin{lemma} \label{stm:markov_chain_finite_reach_bound}
	Given a Markov chain $\MC = (\States, \mctransitions)$, a state $s \in \States$, a constant $\tau \in (0, 1]$, and a target set $\targetset$, for $N \geq \ln(\frac{2}{\tau}) \cdot \cardinality{\States} \delta_{\min}^{-\cardinality{\States}}$ we have
	\begin{equation*}
		\ProbabilityMC<\MC, s>[\reach \targetset] - \ProbabilityMC<\MC, s>[\boundedreach<N> \targetset] \leq \tau.
	\end{equation*}
\end{lemma}

\begin{proof}
	We can express $\ProbabilityMC<\MC, s>[\reach \targetset]$ as a sum of $\ProbabilityMC<\MC, s>[\boundedreach<N> \targetset]$ and $\ProbabilityMC<\MC, s>[\reach^{> N} \targetset]$, where $\reach^{> N} \targetset = \reach \targetset \setminus \boundedreach<N> \targetset$ are all paths which reach the set $\targetset$ but only after at least $N + 1$ steps.
	Clearly,
	\begin{equation*}
		\ProbabilityMC<\MC, s>[\reach \targetset] - \ProbabilityMC<\MC, s>[\boundedreach<N> \targetset] = \ProbabilityMC<\MC, s>[\reach^{> N} \targetset].
	\end{equation*}
	By \cite[Lemma~5.1]{DBLP:journals/jacm/BrazdilKK14} we have that $\ProbabilityMC<\MC, s>[\reach^{> N} \targetset] \leq 2 \cdot c^N$, where $c = \exp(-\cardinality{\States}^{-1} \delta_{\min}^{\cardinality{\States}})$.
	\begin{align*}
		2 \cdot c^N \leq \tau \quad & \Leftrightarrow \quad N \cdot \ln c \geq \ln \frac{\tau}{2} \quad \Leftrightarrow \quad N \geq \ln \frac{\tau}{2} \cdot (\ln c)^{-1} \\
			& \Leftrightarrow \quad N \geq \ln \frac{\tau}{2} \cdot - \cardinality{\States} \delta_{\min}^{-\cardinality{\States}} \quad \Leftrightarrow \quad N \geq \ln \frac{2}{\tau} \cdot  \cardinality{\States} \delta_{\min}^{-\cardinality{\States}} \qedhere
	\end{align*}
\end{proof}

\subsubsection*{Unique Solution of Bellman Equations}
Now, we prove that a particular class of Bellman equations has a unique solution by proving that the associated functor is a contraction.
\begin{lemma} \label{stm:bellman_unique_solution}
	Let $\MDP$ be an MDP, $\stateactions_? : \States \to \Actions$ a function mapping a state $s$ to a subset of its available actions $\stateactions_?(s) \subseteq \stateactions(s)$, $c : \States \to \Reals$ a cost function, and $\strategy$ a memoryless strategy on $\MDP$.
	Define $\States_= = \{s \mid \stateactions_?(s) = \emptyset\}$.

	If $\ProbabilityMDP<\MDP, s><\strategy>[\reach \States_=] > 0$ for all states $s \in \States$, then the system of Bellman equations
	\begin{equation*}
		f(s) = c(s) + {\sum}_{a \in \stateactions_?(s)} \strategy(s, a) \cdot \ExpectedSumMDP{\mdptransitions}{s}{a}{f}
	\end{equation*}
	has a unique solution $f$.
\end{lemma}

\begin{proof}
	Define the iteration operator $F$ as
	\begin{align*}
		F(f)(s) = c(s) + {\sum}_{a \in \stateactions_?(s)} \strategy(s, a) \cdot \ExpectedSumMDP{\mdptransitions}{s}{a}{f}.
	\end{align*}
	Trivially, a function $f : \States \to \Reals$ is a solution to the equation system if and only if it is a fixed point of $F$, i.e.\ $F(f)(s) = f(s)$ for all states $s \in \States$.

	We show that $F^{\cardinality{\States}}$, i.e.\ $F$ applied $\cardinality{S}$ times, is a contraction and thus has a unique fixed point, obtainable by iterating $F$.
	This means that there exists a contraction factor $0 \leq \gamma < 1$ such that for two arbitrary $f, g : \States \to \Reals$, we have
	\begin{equation}
		\max_{s \in \States} \abs*{F^{\cardinality{\States}}(f)(s) - F^{\cardinality{\States}}(g)(s)} \leq \gamma \cdot \max_{s \in \States} \abs*{f(s) - g(s)}. \label{eq:stm:bellman_unique_solution:contraction}
	\end{equation}
	Let $P(s, s', k)$ be the probability of reaching state $s'$ starting from $s$ in exactly $k$ steps using the strategy $\strategy$ by using only actions from $\stateactions_?$.
	Note that for $s \in \States_=$ this implies $P(s, s', k) = 0$ for any $s' \in \States$ and any number $k$.
	For $s \in \States_? \coloneqq \States \setminus \States_=$, we have that
	\begin{equation*}
		F^{\cardinality{S}}(f)(s) = {\sum}_{s' \in \States} \left( {\sum}_{i = 0}^{\cardinality{S} - 1} P(s, s', i) \cdot c(s') \right) + {\sum}_{s' \in \States_?} P(s, s', \cardinality{S}) \cdot f(s')
	\end{equation*}
	Observe that the first term is independent of $f$, hence for $s \in S_?$ we have
	\begin{align*}
		& \abs*{F^{\cardinality{\States}}(f)(s) - F^{\cardinality{\States}}(g)(s)} \\
		& \qquad = \abs*{{\sum}_{s' \in \States_?} P(s, s', \cardinality{S}) \cdot f(s') - {\sum}_{s' \in \States_?} P(s, s', \cardinality{S}) \cdot g(s')} \\
		& \qquad \leq {\sum}_{s' \in \States_?} P(s, s', \cardinality{S}) \cdot \abs*{f(s') - g(s')} \\
		& \qquad \leq \left( {\sum}_{s' \in \States_?} P(s, s', \cardinality{S})  \right) \cdot \max_{s' \in \States} \abs*{f(s') - g(s')}.
	\end{align*}
	By assumption, we have that $\ProbabilityMDP<\MDP, s><\strategy>[\reach \States_=] > 0$.
	This implies that $\ProbabilityMDP<\MDP, s><\strategy>[\boundedreach<\cardinality{\States}> \States_=] \geq \delta_{\min}(\strategy) > 0$ by \cref{stm:mdp_minimum_reachability}.
	For $s \in S_=$, observe that $F^{\cardinality{S}}(f)(s) = f(s) = c(s)$ and hence
	\begin{equation*}
		\abs*{F^{\cardinality{\States}}(f)(s) - F^{\cardinality{\States}}(g)(s)} = \abs*{f(s) - g(s)} = \abs*{c(s) - c(s)} = 0.
	\end{equation*}
	Consequently, $\gamma = \max_{s \in \States_?} \sum_{s' \in \States_?} P(s, s', \cardinality{S}) \leq \delta_{\min}(\strategy) < 1$ satisfies Inequality~\eqref{eq:stm:bellman_unique_solution:contraction} and we have that $F^{\cardinality{\States}}$ is a contraction.
	By the Banach fixed point theorem we get that $F^{\cardinality{\States}}$ has a unique fixed point and thus the equation system has a unique solution. 
\end{proof}

\subsubsection*{From Local to Global Error Bounds}
The next lemma bounds the overall error of an approximation in an MDP given that the approximation is \enquote{close} locally.
By definition
\begin{equation*}
	\ExpectedSumMDP{\mdptransitions}{s}{a}{\ExpectedSumStrat{\strategy}{X}} = {\sum}_{s' \in \States} \mdptransitions(s, a, s') \cdot {\sum}_{a' \in \stateactions(s')} \strategy(s', a') \cdot f(s', a').
\end{equation*}
Thus, the term $X(s, a) - \ExpectedSumMDP{\mdptransitions}{s}{a}{\ExpectedSumStrat{\strategy}{X}}$ in the lemma essentially denotes the difference between the state-action value $X(s, a)$ and the expected value obtained from $X$ in the successors of $(s, a)$ following $\strategy$.
Consequently, $\mathcal{K}$ contains those state-action pairs for which the value under $X$ is consistent with the value of its successors up to some error.
\begin{lemma} \label{stm:assignment_mdp_value}
	Let $\MDP = (\States, \Actions, \stateactions, \mdptransitions)$ be an MDP satisfying \cref{asm:mec_free}, $X : \States \times \stateactions \to [0, 1]$ a function assigning a value between $0$ and $1$ to each state-action pair, $\strategy$ a memoryless strategy on $\MDP$, and $\kappa_l \leq \kappa_u$ two error bounds.
	Set
	\begin{equation*}
		\mathcal{K} \coloneqq \{(s, a) \mid \kappa_l \leq X(s, a) - \ExpectedSumMDP{\mdptransitions}{s}{a}{\ExpectedSumStrat{\strategy}{X}} \leq \kappa_u\}.
	\end{equation*}
	Define a new MDP $\MDP' = (\States, \Actions, \stateactions, \mdptransitions')$ where
	\begin{equation*}
		\mdptransitions'(s, a) = \begin{dcases*}
			\mdptransitions(s, a) & if $(s, a) \in \mathcal{K}$, and \\
			\{\targetstate \mapsto X(s, a), \sinkstate \mapsto 1 - X(s, a)\} & otherwise.
		\end{dcases*}
	\end{equation*}
	Then, for each state $s \in \States$ we have
	\begin{equation*}
		\kappa_l \leq \frac{\delta_{\min}(\strategy)^{\cardinality{\States}}}{\cardinality{\States}} \left( \ExpectedSumStrat{\strategy}{X}(s) - \ProbabilityMDP<\MDP', s><\strategy>[\reach \{\targetstate\}] \right) \leq \kappa_u,
	\end{equation*}
	where $\delta_{\min}(\strategy) = \min\{ \strategy(s, a) \cdot \mdptransitions(s, a, s') \mid s \in \States, a \in \stateactions(s), \strategy(s, a) > 0, s' \in \support(\mdptransitions(s, a)) \}$ is the smallest transition probability in the Markov chain $\MDP^\strategy$.
\end{lemma}

\begin{proof}
	Define $v'(s) = \ProbabilityMDP<\MDP', s><\strategy>[\reach \{\targetstate\}]$.
	Furthermore, let $\mathcal{K}(s) = \{a \in \stateactions(s) \mid (s, a) \in \mathcal{K}\}$ and $\lnot \mathcal{K}(s) = \overline{\mathcal{K}(s)} \intersection \stateactions(s)$ the sets of all actions $a \in \stateactions(s)$ such that $(s, a) \in \mathcal{K}$ and $(s, a) \notin \mathcal{K}$, respectively.
	Observe that $v'$ is a solution to the following system of equations:
	\begin{align*}
		v'(\targetstate) & = 1 \\
		v'(\sinkstate) & = 0 \\
		v'(s) & = {\sum}_{a \in \mathcal{K}(s)} \strategy(s, a) \cdot \ExpectedSumMDP{\mdptransitions}{s}{a}{v'} + {\sum}_{a \in \lnot \mathcal{K}(s)} \strategy(s, a) \cdot X(s, a)
	\end{align*}
	We apply \cref{stm:bellman_unique_solution} to show that $v'$ is the unique solution.
	Let $\varepsilon(\targetstate) = 1$, $\varepsilon(\sinkstate) = 0$, and $\varepsilon(s) = {\sum}_{a \in \lnot \mathcal{K}(s)} \strategy(s, a) \cdot X(s, a)$ for all other $s \in \States$.
	Further, set $\stateactions_?(\targetstate) = \stateactions_?(\sinkstate) = \emptyset$ and $\stateactions_?(s) = \mathcal{K}(s)$ for all other $s \in \States$.
	Then, $\{\targetstate, \sinkstate\} \subseteq \States_=$.
	The MDP $\MDP'$ also satisfies \cref{asm:mec_free}, since no new ECs are introduced, and thus $\ProbabilityMDP<\MDP, s><\strategy>[\reach \States_=] = 1 > 0$ for all $s \in \States$ by \cref{stm:mdp_almost_sure_absorption}.
	Consequently, \cref{stm:bellman_unique_solution} is applicable and $v'$ is the unique solution of the above equations.

	$\ExpectedSumStrat{\strategy}{X}$ satisfies a similar set of equations:
	\begin{align*}
		\ExpectedSumStrat{\strategy}{X}(\targetstate) & = 1 \\
		\ExpectedSumStrat{\strategy}{X}(\sinkstate) & = 0 \\
		\ExpectedSumStrat{\strategy}{X}(s) & = {\sum}_{a \in \stateactions(s)} \strategy(s, a) \cdot X(s, a) \\
			& = {\sum}_{a \in \mathcal{K}(s)} \strategy(s, a) \cdot X(s, a) + {\sum}_{a \in \lnot \mathcal{K}(s)} \strategy(s, a) \cdot X(s, a) \\
			& = \kappa(s) + {\sum}_{a \in \mathcal{K}(s)} \strategy(s, a) \cdot \ExpectedSumMDP{\mdptransitions}{s}{a}{\ExpectedSumStrat{\strategy}{X}} + {\sum}_{a \in \lnot \mathcal{K}(s)} \strategy(s, a) \cdot X(s, a)
	\end{align*}
	where $\kappa(s) = \sum_{a \in \mathcal{K}(s)} \strategy(s, a) \cdot (X(s, a) - \ExpectedSumMDP{\mdptransitions}{s}{a}{\ExpectedSumStrat{\strategy}{X}})$ is bounded by $\kappa_l \leq \kappa(s) \leq \kappa_u$.
	Again, by \cref{stm:bellman_unique_solution}, these equations then have a unique fixed point, setting $\varepsilon(s) = \kappa(s) + {\sum}_{a \in \lnot \mathcal{K}(s)} \strategy(s, a) \cdot X(s, a)$.

	Now, we prove a bound for the difference between $X$ and $v'$ using the above characterizations.
	Observe that the above equation systems only differ structurally by the error term $\kappa(s)$.
	Let thus $f(s) = \ExpectedSumStrat{\strategy}{X}(s) - v'(s)$.
	This $f$ is a fixed point of the following equation system:
	\begin{align*}
		f(\targetstate) & = f(\sinkstate) = 0 \\
		f(s) & = \kappa(s) + {\sum}_{a \in \mathcal{K}(s)} \strategy(s, a) \cdot \ExpectedSumMDP{\mdptransitions}{s}{a}{f}
	\end{align*}
	Clearly, $f$ again is unique by \cref{stm:bellman_unique_solution}.

	Given a state $s$, the probability to reach the terminal states $\targetstate$ and $\sinkstate$ in $\cardinality{S}$ steps following strategy $\strategy$ is bounded from below by $\delta_{\min}(\strategy)^{\cardinality{\States}}$ due to \cref{stm:mdp_minimum_reachability}.
	Consequently, the probability of not reaching these states in $\cardinality{\States}$ steps is bounded from above by $1 - \delta_{\min}(\strategy)^{\cardinality{S}} < 1$.
	Hence, we can bound the difference between $\ExpectedSumStrat{\strategy}{X}$ and $v'$ by
	\begin{equation*}
		\kappa(s) \cdot {\sum}_{n=0}^\infty \cardinality{S} \left( 1 - \delta_{\min}(\strategy)^{\cardinality{S}} \right)^n = \kappa(s) \cdot \cardinality{S} \delta_{\min}(\strategy)^{-\cardinality{S}}. \qedhere 
	\end{equation*}
\end{proof}

\subsubsection*{Bounding Reachability on Similar MDP}
In this lemma, we show that MDP which are sufficiently \enquote{similar} also have similar reachability values.
In particular, we are concerned with MDP that agree on a subset of states.
For another notion of similarity (same transition structure but different transition probabilities) see \cite{DBLP:conf/fossacs/Chatterjee12}.
\begin{lemma} \label{stm:similar_mdp_bound}
	Let $\MDP = (\States, \Actions, \stateactions, \mdptransitions)$ be an MDP, $\targetset \subseteq \States$ a set of target states, $\mathcal{K} \subseteq \States \times \stateactions$ a set of state-action pairs, and $\MDP' = (\States', \Actions', \stateactions', \mdptransitions')$ an arbitrary MDP with $\mathcal{K} \subseteq \States' \times \stateactions'$ that coincides with $\MDP$ on $\mathcal{K}$ and $\targetset$, i.e.\ (i)~$\stateactions(s) = \stateactions'(s)$ for all $s \in \mathcal{K}$, (ii)~$\mdptransitions(s, a) = \mdptransitions'(s, a)$ for all $(s, a) \in \mathcal{K}$, and (iii)~$\targetset \subseteq \States'$.
	Moreover, let $\strategy$ be a strategy in $\MDP$, $s \in \States \intersection \States'$ an arbitrary state in both MDP, and $N \in \Naturals$ a natural number.
	Then,
	\begin{equation*}
		\ProbabilityMDP<\MDP, s><\strategy>[\boundedreach<N> \targetset] \geq \ProbabilityMDP<\MDP', s><\strategy'>[\boundedreach<N> \targetset] - \ProbabilityMDP<\MDP, s><\strategy>[\boundedreach<N> \overline{\mathcal{K}}],
	\end{equation*}
	where $\strategy'$ is an arbitrary strategy equal to $\strategy$ on all finite paths over $\mathcal{K}$, i.e.\ $\strategy(\finitepath) = \strategy'(\finitepath)$ for all $\finitepath \in \mathcal{K}^\star \times \States \intersection \Finitepaths<\MDP>$.
\end{lemma}
\begin{proof}
	For a finite path $\finitepath = s_1 a_1 \ldots a_{n-1} s_n \in \Finitepaths<\MDP>$, let $\ProbabilityMDP<\MDP, s><\strategy>[\finitepath]$ denote the probability of path $\finitepath$ occurring when following strategy $\strategy$ from state $s$.
	Let $\mathcal{K}_N$ denote the (finite) set of all finite paths $\finitepath$ of length $N$ starting in $s$ such that all state-action pairs $(s_i, a_i)$ in $\finitepath$ are in $\mathcal{K}$.
	Similarly, let $\lnot \mathcal{K}_N$ denote the set of all such paths containing at least one state-action pair not in $\mathcal{K}$.
	Let $\mathcal{R}(\finitepath)$ be a function which returns $1$ if some target state of $\targetset$ is in path $\finitepath$ and $0$ otherwise.
	Then, we have the following:
	\begin{align}
		& \ProbabilityMDP<\MDP', s><\strategy'>[\boundedreach<N> \targetset] - \ProbabilityMDP<\MDP, s><\strategy>[\boundedreach<N> \targetset] \label{eq:stm:dql_no_ec:similar_mdp_bound:first} \\
		& \qquad = \begin{aligned}
				\sum_{\mathclap{\finitepath \in \mathcal{K}_N}} \left( \ProbabilityMDP<\MDP', s><\strategy'>[\finitepath] \cdot \mathcal{R}(\finitepath) - \ProbabilityMDP<\MDP, s><\strategy>[\finitepath] \cdot \mathcal{R}(\finitepath) \right) + & \\
				\qquad \sum_{\mathclap{\finitepath \in \lnot \mathcal{K}_N}} \left( \ProbabilityMDP<\MDP', s><\strategy'>[\finitepath] \cdot \mathcal{R}(\finitepath) - \ProbabilityMDP<\MDP, s><\strategy>[\finitepath] \cdot \mathcal{R}(\finitepath) \right) &
			\end{aligned} \label{eq:stm:dql_no_ec:similar_mdp_bound:second} \\
		& \qquad = \sum_{\finitepath \in \lnot \mathcal{K}_N} \left( \ProbabilityMDP<\MDP', s><\strategy'>[\finitepath] \cdot \mathcal{R}(\finitepath) - \ProbabilityMDP<\MDP, s><\strategy>[\finitepath] \cdot \mathcal{R}(\finitepath) \right) \label{eq:stm:dql_no_ec:similar_mdp_bound:third} \\
		& \qquad \leq \sum_{\finitepath \in \lnot \mathcal{K}_N} \ProbabilityMDP<\MDP', s><\strategy'>[\finitepath] \cdot \mathcal{R}(\finitepath) \label{eq:stm:dql_no_ec:similar_mdp_bound:fourth} \\
		& \qquad \leq \sum_{\finitepath \in \lnot \mathcal{K}_N} \ProbabilityMDP<\MDP', s><\strategy'>[\finitepath] \label{eq:stm:dql_no_ec:similar_mdp_bound:fifth} \\
		& \qquad = \ProbabilityMDP<\MDP, s><\strategy>[\boundedreach<N> \overline{\mathcal{K}}] \label{eq:stm:dql_no_ec:similar_mdp_bound:sixth}
	\end{align}
	In Equation~\eqref{eq:stm:dql_no_ec:similar_mdp_bound:second}, we simply split the set of all paths of length $N$ into $\mathcal{K}_N$ and $\lnot \mathcal{K}_N$.
	For Equations~\eqref{eq:stm:dql_no_ec:similar_mdp_bound:third} and \eqref{eq:stm:dql_no_ec:similar_mdp_bound:sixth}, note that $\ProbabilityMDP<\MDP', s><\strategy'>$ and $\ProbabilityMDP<\MDP, s><\strategy>$ agree on $\mathcal{K}_N$ by choice of $\MDP'$ and $\strategy'$. 
\end{proof}

\subsubsection*{Repeating Events in Markov Processes}
Finally, we prove a general statement of Markov processes.
The statement itself seems to be quite obvious, yet surprisingly tricky to prove.
In essence, we want to show the following.
Suppose that we are given a Markov process $X_t$ on some probability space $\Omega$ together with a sequence of events $A_t$.
Moreover, assume that for a significant set of atoms $\omega \in \Omega$ there is an infinite set of times $T$ such that the \emph{conditional} probability of $A_t$ occurring is at least $\varepsilon > 0$, i.e.\ $\Probability[X_t \in A_t \mid X_{t-1}(\omega)] > \varepsilon$.
Then, the set of atoms for which infinitely many $A_t$ actually occur is also significant.
The subtle difficulty of this statement arises from the fact that (i)~conditional probabilities are considered, and (ii)~the set $T$ depends on the particular atom $\omega$.

\begin{lemma} \label{stm:markov_process_repeating}
	Fix some probability space $(\Omega, \mathcal{F}, \Probability)$ and a measure space $(S, \mathcal{S})$.
	Let $X_t : \Omega \to S$ be a Markov process on $\Omega$ and $A_t \in \mathcal{S}$ measurable events in $S$.
	Assume that the set $\Omega' = \{\omega \in \Omega \mid \exists T.~\cardinality{T} = \infty \land \forall t \in T.~\Probability[X_t \in A_t \mid X_{t-1}](\omega) > \varepsilon\}$ has positive measure, i.e.\ $\Probability[\Omega'] > 0$, and that $\Omega'_t = \{\omega \in \Omega \mid \Probability[X_t \in A_t \mid X_{t-1}](\omega) > \varepsilon\}$ is measurable for all $t \in \Naturals$.
	Then, $\Probability[\{\omega \in \Omega \mid \exists T.~\cardinality{T} = \infty \land \forall t \in T.~X_\algostep(\omega) \in A_t\}] = \Probability[\Omega']$.
\end{lemma}
\begin{proof}
	\DeclareDocumentCommand{\tries}{r()}{\mathsf{Tries}(#1)}
	\DeclareDocumentCommand{\successes}{r()}{\mathsf{Succs}(#1)}
	\DeclareDocumentCommand{\ithtry}{r<> d()}{\mathsf{Try}_{#1}\IfNoValueF{#2}{(#2)}}
	\DeclareDocumentCommand{\jthsuccess}{r<> d()}{\mathsf{Succ}_{#1}\IfNoValueF{#2}{(#2)}}
	\DeclareDocumentCommand{\triesafter}{r<> d()}{\mathsf{TriesJ}_{#1}\IfNoValueF{#2}{(#2)}}
	\DeclareDocumentCommand{\ithtryafter}{r<> r<> d()}{\mathsf{TryJ}_{#1,#2}\IfNoValueF{#3}{(#3)}}
	\DeclareDocumentCommand{\ithtryattafter}{r<> r<> r<>}{\mathsf{TryAtTJ}^{#2}_{#1,#3}}
	\DeclareDocumentCommand{\atleastitriesafter}{r<> r<>}{\mathsf{TriesJ}_{#1,#2}}
	\DeclareDocumentCommand{\ithtrytimesafter}{r<> r<>}{\mathsf{TryTimesJ}_{#1,#2}}

	Let $\omega \in \Omega'$.
	By assumption, for each such $\omega$, there exists an infinite set of time-points $\tries(\omega) = \{t_1, t_2, \cdots\}$ with $1 \leq t_1 < t_2 < \cdots$ where $\Probability[X_t \in A_t \mid X_{t-1}](\omega) > \varepsilon$.
	We call such an event a \emph{try} of $\omega$.
	Denote $\ithtry<i>(\omega) = t_i$ or $\infty$ if no such $t_i$ exists, e.g.\ for $\omega \notin \Omega'$.
	Informally, $\ithtry<i>$ is the time of the $i$-th try of some outcome $\omega$.
	$\ithtry<i>$ is measurable by assumption, since its pre-images can be constructed using $\Omega'_t$.
	Moreover, let $\successes(\omega) = \{s_1, s_2, \cdots\} \subseteq \tries(\omega)$ be the times where $X_{s_j}(\omega) \in A_{s_j}$, called \emph{$j$-th success(ful try)}.
	Note that $\successes(\omega)$ possibly is finite or even empty for some outcomes $\omega$, even for $\omega \in \Omega'$, since infinitely many tries may fail.
	Now, let $\jthsuccess<j>(\omega) = s_j \in \successes(\omega)$ the time of the $j$-th success or $\infty$ if no such $s_j$ exists, i.e.\ $j > \cardinality{\successes(\omega)}$.
	$\jthsuccess<j>$ is measurable since $\ithtry<i>$, $X_t$ and $A_t$ are measurable.
	To succinctly capture corner-cases, we further define $\jthsuccess<0> = 0$.
	The successes $\successes(\omega)$ naturally partition the set $\tries(\omega)$ into $\triesafter<j>(\omega) = \{t \in \tries(\omega) \mid \jthsuccess<j>(\omega) < t \leq \jthsuccess<j+1>(\omega)\}$.
	We use $\ithtryafter<i><j>(\omega)$ to refer to the $i$-th element of $\triesafter<j>(\omega)$, or $\infty$ if no such element exists.
	$\ithtryafter<i><j>$ is measurable due to $\jthsuccess<j>$ being measurable.
	Informally, $\ithtryafter<i><j>(\omega)$ denotes the time of the $i$-th try since the $j$-th success.

	We show that after a sufficient number of tries, there is a success with high probability.
	Repeating this argument inductively, we then show that there are infinitely many successes for almost all outcomes $\omega$ in $\Omega'$.

	Let thus $\ithtryattafter<i><t><j>$ denote the set of runs which at time $t$ have succeeded $j$ times and since the $j$-th success experienced $i$-th tries, where this $i$-th try happens exactly at time $t$.
	Formally,
	\begin{equation*}
		\ithtryattafter<i><t><j> \coloneqq \{\omega \in \Omega' \mid \ithtryafter<i><j>(\omega) = t\}.
	\end{equation*}
	Note that this definition implicitly includes the condition $\jthsuccess<j>(\omega) \leq t < \jthsuccess<j+1>(\omega)$ by definition of $\ithtryafter<i><j>$.
	Thus, $\ithtryattafter<i><t><j>$ are disjoint for fixed $i$ and $j$.

	We furthermore define $\atleastitriesafter<i><j> = \Union_{t=1}^\infty \ithtryattafter<i><t><j> = \{\omega \in \Omega' \mid \ithtryafter<i><j>(\omega) < \infty\}$ as the set of outcomes which after their $j$-th success experienced at least $i - 1$\footnote{$\ithtryafter<i><j>(\omega) = t$ does not exclude that the try at time t is successful.} unsuccessful tries.
	We have $\atleastitriesafter<i><j> = \atleastitriesafter<i+1><j> \union \atleastitriesafter<1><j+1>$, since the $i$-th try either fails and the $i+1$-th try is experienced later (since $\atleastitriesafter<i><j> \subseteq \Omega'$, implying infinitely many tries) or the try succeeds.
	Observe that $\atleastitriesafter<i+1><j>$ and $\atleastitriesafter<1><j+1>$ are not disjoint, since, for example, the runs succeeding at the $i+1$-th try also are an element of $\atleastitriesafter<1><j+1>$.
	On the contrary, we show that $\Probability[\atleastitriesafter<i><j> \setminus \atleastitriesafter<1><j+1>] = 0$, i.e.\ almost all runs in $\atleastitriesafter<i><j>$ will eventually succeed again.

	\newcommand{\tryset}{\mathsf{Tries}}
	\newcommand{\successset}{\mathsf{Succs}}
	\newcommand{\trytime}{\mathsf{TryTime}}
	\newcommand{\successtime}{\mathsf{SuccTime}}
	\newcommand{\successruns}{\mathsf{SuccSet}}

	To this end, we argue that for any fixed $j$ we have that $\lim_{i \to \infty} \Probability[\atleastitriesafter<i><j>] = 0$.
	Fix some $j$ and $i$ with $\Probability[\atleastitriesafter<i><j>] > 0$ (otherwise there is nothing to prove, since $\atleastitriesafter<i><j>$ is monotonically decreasing in $i$).
	Let $\ithtrytimesafter<i><j> = \{t \mid \Probability[\ithtryattafter<i><t><j>] > 0\}$ which is non-empty by the previous condition.
	Clearly, $\Probability[\atleastitriesafter<i><j>] = \sum_{t=1}^\infty \Probability[\ithtryattafter<i><t><j>] = \sum_{t \in \ithtrytimesafter<i><j>} \Probability[\ithtryattafter<i><t><j>]$, as $\ithtryattafter<i><t><j>$ are disjoint.
	Observe that $\ithtryattafter<i><t><j>$ is the intersection of several conditions on $X_{t'}$ for $t' < t$ and requiring that $\Probability[X_t \in A_t \mid X_{t-1}] > \varepsilon$.
	Hence, by the Markov property
	\begin{equation*}
		\Probability[X_t \notin A_t \mid \ithtryattafter<i><t><j>] = 1 - \Probability[X_t \in A_t \mid \ithtryattafter<i><t><j>] = 1 - \Probability[X_t \in A_t \mid X_{t-1}] < 1 - \varepsilon.
	\end{equation*}
	Intuitively, this means that the probability of a try at time $t$ succeeding does not depend on the number of previous tries and successes.
	Thus, for all $t \in \ithtrytimesafter<i><j>$, we have $\Probability[X_t \notin A_t \intersection \ithtryattafter<i><t><j>] < (1 - \varepsilon) \cdot \Probability[\ithtryattafter<i><t><j>]$.
	Observe that $\Union_{t=1}^\infty (X_t \notin A_t \intersection \ithtryattafter<i><t><j>) = \atleastitriesafter<i+1><j>$ since the intersection implies that the $i$-th try at time $t$ was unsuccessful.
	Together,
	\begin{align*}
		\Probability[\atleastitriesafter<i+1><j>] & = \Probability[{\Union}_{t=1}^\infty X_t \notin A_t \intersection \ithtryattafter<i><t><j>] = {\sum}_{t=1}^\infty \Probability[X_t \notin A_t \intersection \ithtryattafter<i><t><j>] \\
			& = {\sum}_{t \in \ithtrytimesafter<i><j>} \Probability[X_t \notin A_t \intersection \ithtryattafter<i><t><j>] \\
			& < {\sum}_{t=1}^\infty (1 - \varepsilon) \cdot \Probability[\ithtryattafter<i><t><j>] = (1 - \varepsilon) \cdot \Probability[{\Union}_{t=1}^\infty \ithtryattafter<i><t><j>] \\
			& = (1 - \varepsilon) \cdot \Probability[\atleastitriesafter<i><j>].
	\end{align*}
	Consequently, $\lim_{i \to \infty} \Probability[\atleastitriesafter<i><j>] = 0$ for any fixed $j$.

	As argued before, we have $\atleastitriesafter<i><j> = \atleastitriesafter<i+1><j> \union \atleastitriesafter<1><j+1>$.
	Iterating this equation yields $\atleastitriesafter<i><j> = \atleastitriesafter<i+k><j> \union \atleastitriesafter<1><j+1>$ for any $k \geq 1$ and consequently $\atleastitriesafter<1><j> = \Intersection_{i=1}^\infty \atleastitriesafter<i><j> \union \atleastitriesafter<1><j+1>$.
	Informally, this equation can be read as \enquote{all outcomes which succeed at least $j$ times either try infinitely often or succeed at least $j+1$ times.}
	Let $\atleastitriesafter<\infty><j> = \Intersection_{i=1}^\infty \atleastitriesafter<i><j> = \{\omega \in \Omega' \mid \jthsuccess<j>(\omega) < \infty = \jthsuccess<j+1>(\omega)\}$.
	Clearly, $\atleastitriesafter<\infty><j> \intersection \atleastitriesafter<1><j+1> = \emptyset$, thus we have $\Probability[\atleastitriesafter<1><j+1> \setminus \atleastitriesafter<1><j>] = \Probability[\atleastitriesafter<\infty><j>]$.
	Additionally, we have $\Probability[\atleastitriesafter<\infty><j>] = \inf_{i \in \Naturals} \Probability[\atleastitriesafter<i><j>] = 0$ by the above reasoning.
	Hence $\Probability[\atleastitriesafter<1><j+1> \setminus \atleastitriesafter<1><j>] = 0$.
	This implies that almost all runs in $\Omega'$ succeed infinitely often, concluding the proof. 
\end{proof}
\end{document}